\pdfoutput=1
\RequirePackage{ifpdf}
\ifpdf 
\documentclass[pdftex]{sigma}
\else
\documentclass{sigma}
\fi

\numberwithin{equation}{section}

\newtheorem{Theorem}{Theorem}[section]
\newtheorem{Corollary}[Theorem]{Corollary}
\newtheorem{Lemma}[Theorem]{Lemma}
\newtheorem{Proposition}[Theorem]{Proposition}
 { \theoremstyle{definition}
\newtheorem{Definition}[Theorem]{Definition}
\newtheorem{Example}[Theorem]{Example}
\newtheorem{Remark}[Theorem]{Remark} }

\usepackage{bbm}
\usepackage{stmaryrd}
\usepackage{yfonts}
\input xy
\xyoption{all} 

\newcommand{\Dleft}{[\hspace{-1.5pt}[}
\newcommand{\Dright}{]\hspace{-1.5pt}]}
\newcommand{\SN}[1]{\Dleft #1 \Dright}

\newcommand{\Id}{\mathbbmss{1}}

\newcommand{\InHom}{\mbox{$\underline{\Hom}$}}
\newcommand{\rmi}{ \textnormal{i}}
\newcommand{\rmd}{\textnormal{d}}

\newcommand{\N}{\mathbb{N}}

\newcommand{\op}[1]{\!\!\mathop{\rm ~#1}\nolimits}
\newcommand{\cY}{\mathcal{Y}}
\newcommand{\rmo}{\textnormal{op}}

\newcommand{\catname}[1]{\textnormal{\texttt{#1}}}

\font\black=cmbx10 \font\sblack=cmbx7 \font\ssblack=cmbx5 \font\blackital=cmmib10 \skewchar\blackital='177
\font\sblackital=cmmib7 \skewchar\sblackital='177 \font\ssblackital=cmmib5 \skewchar\ssblackital='177
\font\sanss=cmss10 \font\ssanss=cmss8 
\font\sssanss=cmss8 scaled 600 \font\blackboard=msbm10 \font\sblackboard=msbm7 \font\ssblackboard=msbm5
\font\caligr=eusm10 \font\scaligr=eusm7 \font\sscaligr=eusm5  \font\fraktur=eufm10
\font\sfraktur=eufm7 \font\ssfraktur=eufm5 
\font\bsymb=cmsy10 scaled\magstep2
\def\all#1{\setbox0=\hbox{\lower1.5pt\hbox{\bsymb
 \char"38}}\setbox1=\hbox{$_{#1}$} \box0\lower2pt\box1}
\def\exi#1{\setbox0=\hbox{\lower1.5pt\hbox{\bsymb \char"39}}
 \setbox1=\hbox{$_{#1}$} \box0\lower2pt\box1}

\def\tx#1{{\fam0\relax#1}}

\newfam\bifam
\textfont\bifam=\blackital \scriptfont\bifam=\sblackital \scriptscriptfont\bifam=\ssblackital

\newfam\blfam
\textfont\blfam=\black \scriptfont\blfam=\sblack \scriptscriptfont\blfam=\ssblack

\newfam\bbfam
\textfont\bbfam=\blackboard \scriptfont\bbfam=\sblackboard \scriptscriptfont\bbfam=\ssblackboard

\newfam\ssfam
\textfont\ssfam=\sanss \scriptfont\ssfam=\ssanss \scriptscriptfont\ssfam=\sssanss
\def\sss#1{{\fam\ssfam\relax#1}}

\newfam\clfam
\textfont\clfam=\caligr \scriptfont\clfam=\scaligr \scriptscriptfont\clfam=\sscaligr

\newfam\frfam
\textfont\frfam=\fraktur \scriptfont\frfam=\sfraktur \scriptscriptfont\frfam=\ssfraktur

\def\hpb#1{\setbox0=\hbox{${#1}$}
 \copy0 \kern-\wd0 \kern.2pt \box0}
\def\vpb#1{\setbox0=\hbox{${#1}$}
 \copy0 \kern-\wd0 \raise.08pt \box0}

\def\pmb#1{\setbox0\hbox{${#1}$} \copy0 \kern-\wd0 \kern.2pt \box0}
\def\pmbb#1{\setbox0\hbox{${#1}$} \copy0 \kern-\wd0
 \kern.2pt \copy0 \kern-\wd0 \kern.2pt \box0}
\def\pmbbb#1{\setbox0\hbox{${#1}$} \copy0 \kern-\wd0
 \kern.2pt \copy0 \kern-\wd0 \kern.2pt
 \copy0 \kern-\wd0 \kern.2pt \box0}
\def\pmxb#1{\setbox0\hbox{${#1}$} \copy0 \kern-\wd0
 \kern.2pt \copy0 \kern-\wd0 \kern.2pt
 \copy0 \kern-\wd0 \kern.2pt \copy0 \kern-\wd0 \kern.2pt \box0}
\def\pmxbb#1{\setbox0\hbox{${#1}$} \copy0 \kern-\wd0 \kern.2pt
 \copy0 \kern-\wd0 \kern.2pt
 \copy0 \kern-\wd0 \kern.2pt \copy0 \kern-\wd0 \kern.2pt
 \copy0 \kern-\wd0 \kern.2pt \box0}

\mathchardef\za="710B 
\mathchardef\zb="710C 
\mathchardef\zg="710D 
\mathchardef\zd="710E 
\mathchardef\zve="710F 
\mathchardef\zz="7110 
\mathchardef\zh="7111 
\mathchardef\zvy="7112 
\mathchardef\zi="7113 
\mathchardef\zk="7114 
\mathchardef\zl="7115 
\mathchardef\zm="7116 
\mathchardef\zn="7117 
\mathchardef\zx="7118 
\mathchardef\zp="7119 
\mathchardef\zr="711A 
\mathchardef\zs="711B 
\mathchardef\zt="711C 
\mathchardef\zu="711D 
\mathchardef\zvf="711E 
\mathchardef\zq="711F 
\mathchardef\zc="7120 
\mathchardef\zw="7121 
\mathchardef\ze="7122 
\mathchardef\zy="7123 
\mathchardef\zf="7124 
\mathchardef\zvr="7125 
\mathchardef\zvs="7126 
\mathchardef\zf="7127 
\mathchardef\zG="7000 
\mathchardef\zD="7001 
\mathchardef\zY="7002 
\mathchardef\zL="7003 
\mathchardef\zX="7004 
\mathchardef\zP="7005 
\mathchardef\zS="7006 
\mathchardef\zU="7007 
\mathchardef\zF="7008 
\mathchardef\zW="700A 
\mathchardef\zC="7009 

\newcommand{\R}{{\mathbb R}}

\newcommand{\C}{{\mathbb C}}
\newcommand{\Z}{{\mathbb Z}}

\newcommand{\s}{{\textstyle *}}

\newcommand{\A}{{\mathcal A}}

\def\Hom{\textsf{Hom}}

\def\cD{{\mathcal{D}}}

\def\ul{\underline}

\def\la{\langle}

\def\sT{{\sss T}}

\def\st{{\sss t}}

\def\xi{\tx{i}}

\def\cM{\cal M}

\def\cD{\cal D}

\def\s*{{\scriptstyle *}}
\def\cO{\mathcal{O}}

\def\cF{\mathcal{F}}

\def\cM{\mathcal{M}}

\def\ul{\underline}

\newcommand{\0}{\otimes}
\newcommand{\id}{\op{id}}

\begin{document}
\allowdisplaybreaks

\newcommand{\arXivNumber}{1906.09834}

\renewcommand{\PaperNumber}{002}

\FirstPageHeading

\ShortArticleName{The Schwarz--Voronov Embedding of ${\mathbb Z}_{2}^{n}$-Manifolds}

\ArticleName{The Schwarz--Voronov Embedding of $\boldsymbol{{\mathbb Z}_{2}^{n}}$-Manifolds}

\Author{Andrew James BRUCE, Eduardo IBARGUENGOYTIA and Norbert PONCIN}
\AuthorNameForHeading{A.J.~Bruce, E.~Ibarguengoytia and N.~Poncin}
\Address{Mathematics Research Unit, University of Luxembourg, Maison du Nombre 6,\\ avenue de la Fonte, L-4364 Esch-sur-Alzette, Luxembourg}
\Email{\href{mailto:andrew.bruce@uni.lu}{andrew.bruce@uni.lu}, \href{mailto:eduardo.ibarguengoytia@uni.lu}{eduardo.ibarguengoytia@uni.lu}, \href{mailto:norbert.poncin@uni.lu}{norbert.poncin@uni.lu}}

\ArticleDates{Received July 10, 2019, in final form December 30, 2019; Published online January 08, 2020}

\Abstract{Informally, ${\mathbb Z}_2^n$-manifolds are `manifolds' with ${\mathbb Z}_2^n$-graded coordinates and a sign rule determined by the standard scalar product of their ${\mathbb Z}_2^n$-degrees. Such manifolds can be understood in a sheaf-theoretic framework, as supermanifolds can, but with significant differences, in particular in integration theory. In this paper, we reformulate the notion of a~${\mathbb Z}_2^n$-manifold within a categorical framework via the functor of points. We show that it is sufficient to consider ${\mathbb Z}_2^n$-points, i.e., trivial ${\mathbb Z}_2^n$-manifolds for which the reduced manifold is just a single point, as `probes' when employing the functor of points. This allows us to construct a fully faithful restricted Yoneda embedding of the category of ${\mathbb Z}_2^n$-manifolds into a~subcategory of contravariant functors from the category of ${\mathbb Z}_2^n$-points to a category of Fr\'echet manifolds over algebras. We refer to this embedding as the \emph{Schwarz--Voronov embedding}. We further prove that the category of ${\mathbb Z}_2^n$-manifolds is equivalent to the full subcategory of locally trivial functors in the preceding subcategory.}


\Classification{58C50; 58D1; 14A22}

\section{Introduction}
Various notions of \emph{graded geometry} play an important r\^ole in mathematical physics and can often provide further insight into classical geometric constructions. For example, supermanifolds, as pioneered by Berezin and collaborators, are essential in describing quasi-classical systems with both bosonic and fermionic degrees of freedom. Very loosely, supermanifolds are `manifolds' for which the structure sheaf is $\Z_2$-graded. Such geometries are of fundamental importance in perturbative string theory, supergravity, and the BV-formalism, for example. While the theory of supermanifolds is firmly rooted in theoretical physics, it has since become a respectable area of mathematical research. Indeed, supermanifolds allow for an economical description of Lie algebroids, Courant algebroids as well as various related structures, many of which are of direct interest to physics. We will not elaborate any further and urge the reader to consult the ever-expanding literature.

Interestingly, {\it $\Z_{2}^{n}$-gradings} ($\Z_2^n=\Z_2^{\times n}$, $n \geq 2$) can be found in the theory of parastatistics, see for example \cite{Druhl:1970,Green:1953,Greenberg:1965,Yang:2001}, behind an alternative approach to supersymmetry \cite{Tolstoy:2013}, in relation to the symmetries of the L\'{e}vy-Lebond equation \cite{Aizawa:2016}, and behind the theory of mixed symmetry tensors \cite{Bruce:2018a}. Generalizations of the super Schr\"{o}dinger algebra (see \cite{Aizawa:2017}) and the super Poincar\'{e} algebra (see~\cite{Bruce:2019}) have also appeared in the literature. That said, it is unknown if these `higher gradings' are of the same importance in fundamental physics as $\Z_2$-gradings. It must also be remarked that the quaternions and more general Clifford algebras can be understood as $\Z_2^n$-graded $\Z_2^n$-commutative (see below) algebras \cite{Albuquerque:1999, Albuquerque:2002}. Thus, one may expect $\Z_{2}^{n}$-gradings to be important in studying Clifford algebras and modules, though the implications for classical and quantum field theory remain as of yet unexplored. It should be further mentioned that any `sign rule' can be understood in terms of a $\Z_{2}^{n}$-grading (see \cite{Covolo:2016}). A natural question here is to what extent can {\it $\Z_2^n$-graded geometry} be developed.

\looseness=-1 A locally ringed space approach to {\it $\Z_{2}^{n}$-manifolds} has been constructed in a series of papers by Bruce, Covolo, Grabowski, Kwok, Ovsienko \& Poncin \cite{Bruce:2018a, Bruce:2018b, Covolo:2016,Covolo:2016a,Covolo:2016b,Covolo:2016c,COP, Poncin:2016}. It includes the $\Z_2^n$-differential-calculus, the $\Z_2^n$-Berezinian, as well as a low-dimensional $\Z_2^n$-integration-theory. Integration on $\Z_2^n$-manifolds turns out to be fundamentally different from integration on $\Z_2^1$-manifolds (i.e., supermanifolds) and is currently being constructed in full gene\-ra\-lity by authors of the present paper. The novel aspect of integration on $\Z_2^n$-manifolds is integration with respect to the non-zero degree even parameters (for some preliminary results see~\cite{Poncin:2016}).

Loosely, $\Z_{2}^{n}$-manifolds are `manifolds' for which the structure sheaf has a $\Z_{2}^{n}$-grading and the commutation rule for the local coordinates comes from the standard scalar product of their $\Z_2^n$-degrees. This is not just a trivial or straightforward generalization of the notion of a supermanifold as one has to deal with formal coordinates that anticommute with other formal coordinates, but are themselves \emph{not} nilpotent. Due to the presence of formal variables that are not nilpotent, formal power series are used rather than polynomials (for standard supermanifolds all functions are polynomial in the Grassmann odd variables). The use of formal power series is unavoidable in order to have a well-defined local theory (see \cite{Covolo:2016}), and a well-defined differential calculus (see \cite{Covolo:2016b}). Heuristically, one can view supermanifolds as `mild' noncommutative geo\-metries: the noncommutativity is seen simply as anticommutativity of the odd coordinates. In a~similar vein, one can view $\Z_2^n$-manifolds ($n >1$) as examples of `mild' nonsupercommutative geometries: the sign rule involved is not determined by the coordinates being even or odd, i.e., by their total degree, but by their $\Z_2^n$-degree.

The idea of understanding supermanifolds, i.e., $\Z_2^1$-manifolds, as `Grassmann algebra valued manifolds' can be traced back to the pioneering work of Berezin \cite{Berezin:1987}. An informal understanding along these lines has continuously been employed in physics, where one chooses a `large enough' Grassmann algebra to capture the aspects to the theory needed. This informal understanding leads to the DeWitt--Rogers approach to supermanifolds which seemed to avoid the theory of locally ringed spaces altogether. However, arbitrariness in the choice of the underlying Grassmann algebra is somewhat displeasing. Furthermore, developing the mathematical consistency of DeWitt--Rogers supermanifolds takes one back to the sheaf-theoretic approach of Berezin \& Leites: for a comparison of these approaches, the reader can consult Rogers \cite{Rogers:2007} or Schmitt \cite{Schmitt:1997}. From a physics perspective, there seems no compelling reason to think that there is any physical significance to the choice of underlying Grassmann algebra. To quote Schmitt \cite{Schmitt:1997}: ``However, no one has ever measured a Grassmann number, everyone measures real numbers". The solution here is, following Schwarz \& Voronov \cite{Schwarz:1982,Schwarz:1984,Voronov:1984}, not to fix the underlying Grassmann algebra, but rather understand supermanifolds as functors from the category of finite-dimensional Grassmann algebras to, in the first instance, the category of sets. For a given, but arbitrary, Grassmann algebra $\Lambda$, one speaks of the set of $\Lambda$-points of a supermanifold. It is well known that the set of $\Lambda$-points of a given supermanifold comes with the further structure of a $\Lambda_0$-smooth manifold. That is we, in fact, do not only have a set, but also the structure of a finite-dimensional manifold whose tangent spaces are $\Lambda_0$-modules. Moreover, thinking of supermanifolds as functors, not all natural transformations between the $\Lambda$-points correspond to genuine supermanifold morphisms, only those that respect the $\Lambda_0$-smooth structure do. A~similar approach is used by Molotkov~\cite{Molotkov:1986}, who defines Banach supermanifolds roughly speaking as specific functors from the category of finite-dimensional Grassmann algebras to the category of smooth Banach manifolds of a particular type. The classical roots of these ideas go back to Weil~\cite{Weil:1953} who considered the $A$-points of a~mani\-fold as the set of maps from the algebra of smooth functions on the mani\-fold to a~specified finite-dimensional commutative local algebra $A$. Today one refers to \emph{Weil functors} and these have long been utilised in the theory of jet structures over manifolds, see for example~\cite{Kolar:1993}.

\looseness=1 In this paper, we study {\it Grothendieck's functor of points} \cite{Grothendieck:1973a} of a $\Z_2^n$-manifold $M$, which is a contravariant functor $M(-)$ from the category of $\Z_2^n$-manifolds to the category of sets, and restrict it to the category of \emph{$\Z_2^n$-points}, i.e., trivial $\Z_2^n$-manifolds $\R^{0|\ul q}$ that have no degree zero coordinates. More precisely, we consider the restricted \emph{Yoneda functor} $M \mapsto M(-)$ from the category of $\Z_2^n$-manifolds to the category of contravariant functors from $\Z_2^n$-points to sets. Dual to $\Z_2^n$-points $\R^{0|\ul q}$ are what we will call \emph{$\Z_2^n$-Grassmann algebras} $\zL$ (see Definition~\ref{def:Z2nGrassmann}). The aim of this paper is to carefully prove and generalise the main results of Schwarz \& Voronov \cite{Schwarz:1984,Voronov:1984} to the `higher graded' setting. In particular, we show that $\Z_2^n$-points $\R^{0|\ul q}\simeq \zL$ are actually sufficient to act as `probes' when employing the functor of points (see Theorem \ref{thm:generatingset}). However, not all natural transformations $\zh_\zL\colon M(\zL)\to N(\zL)$ (where $\zL$ is a variable) between the sets $M(\zL)$, $N(\zL)$ of $\Lambda$-points correspond to morphisms $\phi\colon M\to N$ of the underlying $\Z_2^n$-manifolds. By carefully analysing the image of the functor of points, we prove that the set $M(\zL)$ of $\Lambda$-points of a $\Z_2^n$-manifold $M$ comes with the extra structure of a Fr\'echet $\Lambda_{0}$-manifold (see Theorem~\ref{thm:ManStruct}; by $\Lambda_{0}$ we mean the subalgebra of degree zero elements of the $\Z_2^n$-Grassmann algebra $\zL$). Note that we are not trying to define infinite-dimensional $\Z_2^n$-manifolds, yet infinite-dimensional manifolds, specifically Fr\'{e}chet manifolds, are fundamental to our paper. Moreover, we show that natural transformations $\zh_\zL$ between sets of $\Lambda$-points arise from morphisms $\phi$ of $\Z_2^n$-manifolds if and only if they respect the Fr\'{e}chet $\Lambda_{0}$-manifold structures (see Proposition~\ref{prop:NatMorp}). By restricting accordingly the natural transformations allowed, we get a {\it full and faithful embedding} of the category of $\Z^n_2$-manifolds into the category of contravariant functors from the category of $\Z_2^n$-points to the category of nuclear Fr\'{e}chet manifolds over nuclear Fr\'echet algebras. This embedding we refer to as the \emph{Schwarz--Voronov embedding} (see Definition~\ref{def:SVMembedding}). We finally study representability of such contravariant functors and prove that the category of $\Z_2^n$-manifolds is equivalent to the full subcategory of locally trivial functors in the just depicted subcategory of contravariant functors from $\Z_2^n$-points to nuclear Fr\'echet manifolds (see Theorem~\ref{trm:ReprCon}).

\textbf{Methodology:} As $\Z_2^n$-manifolds have well defined local models, we work with $\Z_2^n$-domains and then `globalize' the results to general $\Z_2^n$-manifolds. We modify the approach of Schwarz \& Voronov \cite{Schwarz:1984,Voronov:1984} and draw on Balduzzi, Carmeli \& Fioresi \cite{Balduzzi:2010,Balduzzi:2013} and Konechny \& Schwarz \cite{Konechny:1998,Konechny:2000}, making all changes necessary to encompass $\Z_2^n$-manifolds. Let us mention that Balduzzi, Carmeli \& Fioresi study functors from the category of super Weil algebras and not that of Grassmann algebras. However, if we truly want to build a restricted Yoneda embedding, the source category of the functors of points must be a category of algebras that is opposite to some category of supermanifolds~-- and super Weil algebras are not the algebras of functions of some class of supermanifolds (unless they are Grassmann algebras). Moreover, the idea behind our restriction of the Yoneda embedding is `the smaller the class of test algebras, the better'~-- which points again to Grassmann algebras as being the somewhat privileged objects. The most striking difference between supermanifolds and $\Z_2^n$-manifolds ($n>1$) is that we are forced, due to the presence of non-zero degree even coordinates, to work with (infinite-dimensional) Fr\'{e}chet spaces, algebras and manifolds. Interestingly, nuclearity of the values $M(\zL)$ of the functor of points of a $\Z_2^n$-manifold $M$, i.e., nuclearity of the local models of the Fr\'echet $\zL_0$-manifolds~$M(\zL)$ or of their tangent spaces, does not play a r\^{o}le in the proofs of the statements in this paper. More precisely, the functor of points $M(-)$ has values $M(\zL)$ that {\it are} nuclear Fr\'echet $\zL_0$-manifolds. Conversely, a functor $\cF(-)$ whose values $\cF(\zL)$ are Fr\'echet $\zL_0$-manifolds {\it and} which is representable, has nuclear values (nuclearity is encrypted in the representability condition (see Theorem \ref{trm:ReprCon})). Although nuclearity of the tangent spaces of the mani\-folds~$M(\zL)$ is not explicitly used throughout this work, we do {\it not} at all claim that nuclearity is not of importance in the theory of $\Z_2^n$-manifolds. For instance, the function sheaf of a $\Z_2^n$-manifold {\it is} a~nuclear Fr\'echet sheaf of $\Z_2^n$-graded $\Z_2^n$-commutative algebras -- a fact that is crucial for product $\Z_2^n$-manifolds and $\Z_2^n$-Lie groups \cite{Bruce:2018b}.

\textbf{Applications:} The functor of points has been used informally in physics as from the very beginning. It is actually of importance in situations where there is no good notion of point (see also Section~\ref{FoP}), for instance in algebraic geometry and in super- and $\Z_2^n$-geometry. Construc\-ting a~set-valued functor and showing that it is representable as a locally ringed space, e.g., a~scheme or a $\Z_2^n$-manifold, is often easier than building that scheme or manifold directly. Functors that are not representable can be interpreted as generalised schemes or generalised $\Z_2^n$-manifolds. Further, the category of functors is better behaved than the corresponding category of supermanifolds or of other types of spaces. Also homotopical algebraic geometry \cite{TVI, TVII}, as well as its generalisation that goes under the name of homotopical algebraic $\cD$-geometry (where~$\cD$ refers to differential operators) \cite{BPP:KTR, BPP:HAC}, are fully based on the functor of points approach. Finally, the functor of points turns out to be an indispensable tool when it comes to the investigation of $\Z_2^n$-Lie groups and their actions on $\Z_2^n$-manifolds, of geometric $\Z_2^n$-vector bundles \dots. These concepts are explored in upcoming texts that are currently being written down.

 \textbf{Arrangement:} In Section \ref{sec:Z2nGeometry}, we review the basic tenets of $\Z_2^n$-geometry and the theory of $\Z_2^n$-manifolds. The bulk of this paper is to be found in Section \ref{sec:Z2nPtsSVMemb}. We rely on two appendices: in Appendix \ref{appx:GenSet} we recall the notion of a generating set of a category, and in Appendix \ref{appx:AManifolds} we review indispensable concepts from the theory of Fr\'{e}chet spaces, algebras and manifolds.

\section[Rudiments of $\Z_{2}^{n}$-graded geometry]{Rudiments of $\boldsymbol{\Z_{2}^{n}}$-graded geometry}\label{sec:Z2nGeometry}

\subsection[The category of $\Z_{2}^{n}$-manifolds]{The category of $\boldsymbol{\Z_{2}^{n}}$-manifolds}
The locally ringed space approach to $\Z_{2}^{n}$-manifolds is presented in a series of papers \cite{Covolo:2016,Covolo:2016a, Covolo:2016b, Covolo:2016c,COP, Poncin:2016} by Covolo, Grabowski, Kwok, Ovsienko, and Poncin. We will draw upon these works heavily and not present proofs of any formal statements.

\begin{Definition}A \emph{locally} $\Z_{2}^{n}$-\emph{ringed space}, $n \in \mathbb{N}$, is a pair $X := (|X|, \mathcal{O}_{X} )$, where $|X|$ is a~second-countable Hausdorff space, and $\mathcal{O}_{X}$ is a~sheaf of $\Z_{2}^{n}$-graded $\Z_{2}^{n}$-commutative associative unital $\mathbb{R}$-algebras, such that the stalks $\mathcal{O}_{p}$, $p \in |X|$, are local rings.
\end{Definition}

In this context, $\Z_{2}^{n}$-commutative means that any two sections $a, b \in \mathcal{O}_{X}(|U|)$, $|U| \subset |X|$ open, of homogeneous degrees $\deg(a) = \underline{a} \in \Z_{2}^{n}$ and $\deg(b) = \underline{b} \in \Z_{2}^{n}$ commute according to the sign rule
\begin{gather*} ab = (-1)^{\langle \underline{a}, \underline{b}\rangle} ba,\end{gather*}
where $\langle - , -\rangle$ is the standard scalar product on $\Z_{2}^{n}$. We will say that a section $a$ is \emph{even} or \emph{odd} if $\langle \underline{a} , \underline{a} \rangle \in \Z_{2}$ is $0$ or $1$.

Just as in standard supergeometry, which we recover for $n=1$, a locally $\Z_2^n$-ringed space is a $\Z_2^n$-manifold if it is locally isomorphic to a specific local model. Given the central r\^ole of (finite-dimensional) Grassmann algebras in the theory of supermanifolds, we consider here $\Z_2^n$-Grassmann algebras.

\begin{Remark}In the following, we order the elements in $\Z_{2}^{n}$ \emph{lexicographically}, and refer to this ordering as the \emph{standard ordering}. For example, we thus get
\begin{gather*} \Z_{2}^{2} = \{ (0,0), (0,1), (1,0), (1,1)\} .\end{gather*} \end{Remark}

\begin{Definition}\label{def:Z2nGrassmann}A {\it $\Z_2^n$-Grassmann algebra} $\Lambda^{ \underline{q}} := \R[[\zx]]$ is the $\Z_2^n$-graded $\Z_2^n$-commutative associative unital $\R$-algebra of all formal power series with coefficients in $\R$ generated by homogeneous parameters $\zx^{\alpha}$ subject to the commutation relation
\begin{gather*} \zx^\alpha \zx^\beta = (-1)^{\langle \underline{\alpha} , \underline{\beta} \rangle } \zx^{\beta} \zx^\alpha,\end{gather*}
where $\ul\za:=\deg(\zx^\alpha) \in \Z_2^n\setminus \underline{0}$, $\ul{0}=(0,\ldots,0)$. The tuple $\underline{q} = (q_{1}, q_{2}, \dots , q_{N})$, $N = 2^{n}-1$, provides the number $q_i$ of generators $\zx^\za$, which have the $i$-th degree in $\Z_2^n\setminus\ul{0}$ (endowed with its standard order).

A {\it morphism of $ \Z_2^n$-Grassmann algebras}, $\psi^{*} \colon \Lambda^{\underline{q}}\rightarrow \Lambda^{\underline{p}}$, is a map of $\R$-algebras that preserves the $\Z_2^n$-grading and the units.

We denote the category of $\Z_2^n$-Grassmann algebras and corresponding morphisms by $\Z_2^n\catname{GrAlg}$.
\end{Definition}
\begin{Example}For $n=0$, we simply get $\R$ considered as an algebra over itself.
\end{Example}

\begin{Example}If $n=1$, we recover the classical concept of Grassmann algebra with the standard supercommutation rule for generators. In this case, all formal power series truncate to polynomials. In particular, the Grassmann algebra generated by a single odd generator is isomorphic to the algebra of dual numbers.
\end{Example}

\begin{Example}The $\Z_2^2$-Grassmann algebra $\Lambda^{(1,1,1)}$ is described by three generators
\begin{gather*}\big (\underbrace{\zx}_{(0,1)}, \underbrace{\theta}_{(1,0)}, \underbrace{z}_{(1,1)} \big),\end{gather*}
where we have indicated the $\Z_2^2$-degree. Note that $\zx \theta = \theta \zx$, while $\zx^2 = 0$ and $\theta^2 =0$. Moreover, $\zx z = {-} z \zx$ and $\theta z = {-} z\theta$, while $z$ is not nilpotent. A general (inhomogeneous) element of $\Lambda^{(1,1,1)}$ is then of the form
\begin{gather*}f(\zx, \theta, z) = f_z(z) + \zx f_\zx(z) + \theta f_\theta (z) + \zx \theta f_{\zx\theta}(z),\end{gather*}
where $f_z(z)$, $f_\zx(z)$, $f_\theta(z)$ and $f_{\zx\theta}(z)$ are formal power series in~$z$. As a subalgebra we can consider~$\Lambda^{(1,1,0)}$, whose generators are $\zx$ and $\theta$. A general element of this subalgebra is a~polynomial in these generators.
\end{Example}

Within any $\Z_2^n$-Grassmann algebra $\Lambda := \Lambda^{\underline{q}}$, we have the ideal generated by the generators of $\zL$, which we will denote as $\mathring{\Lambda}$. In particular we have the decomposition
\begin{gather*}\Lambda = \R \oplus \mathring{\Lambda} ,\end{gather*}
which will be used later on. Moreover, the set of degree $0$ elements, $\Lambda_{0} \subset \Lambda$, is a commutative associative unital $\R$-algebra.

Very informally, a $\Z_{2}^{n}$-manifold is a smooth manifold whose structure sheaf has been `deformed' to now include the generators of a $\Z_2^n$-Grassmann algebra.

\begin{Definition}\label{LRSDefi}A (smooth) $\Z_{2}^{n}$-\emph{manifold} of dimension $p |\underline{q}$ is a locally $\Z_{2}^{n}$-ringed space $ M := \left(|M|, \cO_M \right)$, which is locally isomorphic to the locally $\Z_{2}^{n}$-ringed space $\mathbb{R}^{p |\underline{q}} := \left( \mathbb{R}^{p}, C^{\infty}_{\mathbb{R}^{p}}[[\zx]] \right)$. Local sections of $\cO_M$ are thus formal power series in the $\Z_{2}^{n}$-graded variables $\zx$ with smooth coefficients,
\begin{gather*} \cO_M(|U|) \simeq C^{\infty}_{\R^p}(|U|)[[\zx]] := \bigg\{ \sum_{\alpha \in \mathbb{N}^{\sum_i\! q_i}} \zx^{\alpha}f_{\alpha} \colon f_{\alpha} \in C_{\R^p}^{\infty}(|U|)\bigg \},\end{gather*}
for `small enough' open subsets $|U|\subset |M|$. A $\Z_2^n$-\emph{morphism}, i.e., a morphism between two $\Z_{2}^{n}$-mani\-folds, say $M$ and $N$, is a morphism of $\Z_{2}^{n}$-ringed spaces, that is, a pair $\phi = (|\phi|, \phi^{*}) \colon \allowbreak (|M|, \cO_M) \rightarrow (|N|, \cO_N)$ consisting of a continuous map $|\phi|\colon |M| \rightarrow |N|$ and a sheaf morphism $\phi^{*} \colon \cO_N \rightarrow |\zvf|_*\cO_M$, i.e., a family of $\Z_2^n$-graded unital $\R$-algebra morphisms $\phi^*_{|V|}\colon \cO_N(|V|) \rightarrow \cO_M(|\phi|^{-1}(|V|))$ ($|V| \subset |N|$ open), which commute with restrictions. We will refer to the global sections of the structure sheaf $\cO_M$ as \emph{functions} on $M$ and denote them as $C^{\infty}(M) := \cO_{M}(|M|)$.
\end{Definition}

\begin{Example}[the local model]\label{exp:SuperDom}
The locally $\Z_{2}^{n}$-ringed space $\mathcal{U}^{p|\underline{q}} := \big(\mathcal{U}^p , C^\infty_{\mathcal{U}^p}[[\zx]] \big)$, where $\mathcal{U}^p \subset \R^p$ is open, is naturally a $\Z_2^n$-manifold~-- we refer to such $\Z_2^n$-manifolds as \emph{$\Z_2^n$-domains} of dimension $p|\underline{q}$. We can employ (natural) coordinates $(x^a, \zx^\alpha)$ on any $\Z_2^n$-domain, where the $x^a$ form a coordinate system on $\mathcal{U}^p$ and the $\zx^\alpha$ are formal coordinates.
\end{Example}

Canonically associated to any $\Z_{2}^{n}$-graded algebra $\mathcal{A}$ is the homogeneous ideal $J$ of $\mathcal{A}$ generated by
all homogeneous elements of $\mathcal{A}$ having nonzero degree. If $f\colon \mathcal{A} \rightarrow \mathcal{A}^{\prime}$ is a morphism of $\Z_{2}^{n}$-graded algebras, then $f (J_{\mathcal{A}} ) \subset J_{\mathcal{A}^{\prime}}$. The $J$-adic topology plays a fundamental r\^ole in the theory of $\Z_{2}^{n}$-manifolds. In particular, these notions can be `sheafified'. That is, for any $\Z_{2}^{n}$-manifold~$M$, there exists an ideal sheaf~$\mathcal{J}_M$, defined by
$\mathcal{J}(|U| ) = \la f\in \cO_M(|U|)\colon \deg(f)\neq 0 \rangle$. The $\mathcal{J}_M$-adic topology on $\cO_M$ can then be defined in the obvious way.

Many of the standard results from the theory of supermanifolds pass over to $\Z_{2}^{n}$-manifolds. For example, the topological space $|M|$ comes with the structure of a smooth manifold of dimension $p$ and the continuous base map of any $\Z_2^n$-morphism is actually smooth. Further, for any $\Z_2^n$-manifold~$M$, there exists a short exact sequence of sheaves of $\Z_2^n$-graded $\Z_2^n$-commutative associative $\R$-algebras
\begin{gather*}0\longrightarrow\ker\zve\longrightarrow\cO_M\stackrel{\zve}{\longrightarrow}C^\infty_{|M|}\longrightarrow 0,\end{gather*}
 such that $\ker \zve=\mathcal{J}_M$.

The immediate problem with $\Z_{2}^{n}$-manifolds is that $\mathcal{J}_M$ is \emph{not} nilpotent -- for supermanifolds the ideal sheaf is nilpotent and this is a fundamental property that makes the theory of supermanifolds so well-behaved. However, this loss of nilpotency is compensated by Hausdorff completeness of $\cO_M$ with respect to the $\mathcal{J}_M$-adic topology.

\begin{Proposition} Let $M$ be a $\Z_{2}^{n}$-manifold. Then $\cO_{M}$ is $\mathcal{J}_M$-adically Hausdorff complete as a~sheaf of $\Z_{2}^{n}$-commutative associative unital $\R$-algebras, i.e., the morphism
\begin{gather*}\cO_M \rightarrow \lim_{\leftarrow k} \cO_M / \mathcal{J}_M^{k} ,\end{gather*}
naturally induced by the filtration of $\cO_{M}$ by the powers of $\mathcal{J}_M$, is an isomorphism.
\end{Proposition}

The presence of formal power series in the coordinate rings of $\Z_{2}^{n}$-manifolds forces one to rely on the Hausdorff-completeness of the $\mathcal{J}$-adic topology. This completeness replaces the standard fact that supermanifold functions of Grassmann odd variables are always polynomials~-- a result that is often used in extending results from smooth manifolds to supermanifolds.

What makes $\Z_{2}^{n}$-manifolds a very workable form of noncommutative geometry is the fact that we have well-defined local models. Much like the theory of manifolds, one can construct global geometric concepts via the gluing of local geometric concepts. That is, we can consider a $\Z_{2}^{n}$-manifold as being covered by $\Z_2^n$-domains together with specified gluing information, i.e., coordinate transformations. Moreover, we have the \emph{chart theorem} \cite[Theorem~7.10]{Covolo:2016} that says that $\Z_2^n$-morphisms from a $\Z_2^n$-manifold $(|M|,\cO_M)$ to a $\Z_2^n$-domain $(\mathcal{U}^p,C^\infty_{\mathcal{U}^p}[[\zx]])$, are completely described by the pullbacks of the coordinates $(x^a,\zx^\za)$. In other words, to define a $\Z_2^n$-morphism valued in a $\Z_2^n$-domain, we only need to provide total sections $(s^a,s^\za)\in\cO_M(|M|)$ of the source structure sheaf, whose degrees coincide with those of the target coordinates $(x^a,\zx^\za)$. Let us stress the condition $(\ldots,\zve s^a,\ldots)(|M|)\subset \mathcal{U}^p$, which is often understood in the literature.

\newcommand{\Ci}{C^\infty}\newcommand{\cU}{\mathcal{U}}\newcommand{\cV}{\mathcal{V}}

A few words about the atlas definition of a $\Z_2^n$-manifold are necessary. Let $p|\ul q$ be as above. A~$p|\ul q$-\emph{chart} (or $p|\ul q$-coordinate-system) over a (second-countable Hausdorff) smooth manifold $|M|$ is a $\Z_2^n$-domain \begin{gather*} \mathcal{U}^{p|\ul q}=\big(\cU^p,\Ci_{\cU^p}[[\zx]]\big) ,\end{gather*} together with a diffeomorphism $|\psi|\colon |U|\to \cU^p$, where $|U|$ is an open subset of $|M|$. Given two $p|\ul q$-charts
\begin{gather}\label{Charts}\big(\mathcal{U}^{p|\ul q}_\za,|\psi_\za|\big)\qquad\text{and}\qquad \big(\mathcal{U}^{p|\ul q}_\zb,|\psi_\zb|\big)\end{gather}
 over $|M|$, we set $V_{\alpha \beta} : = |\psi_\alpha|(|U_{\alpha \beta}|)$ and $V_{ \beta \alpha} : = |\psi_\beta|(|U_{\alpha \beta}|)$, where $|U_{\alpha \beta}| := |U_{\alpha}|\cap |U_{\beta}| $. We then denote by $|\psi_{\zb\za}|$ the diffeomorphism
\begin{gather}\label{Diffeo}|\psi_{\zb\za}|:=|\psi_\zb|\circ|\psi_\za|^{-1}\colon \ V_{\za\zb} \to V_{\zb\za} .\end{gather}
Whereas in classical differential geometry the coordinate transformations are completely defined by the coordinate systems, in $\Z_2^n$-geometry, they have to be specified separately. A \emph{coordinate transformation} between two charts, say the ones of~\eqref{Charts}, is an isomorphism of $\Z_2^n$-manifolds \begin{gather}\label{CoordTrans}\psi_{\zb\za}=(|\psi_{\zb\za}|,\psi^*_{\zb\za})\colon \ \mathcal{ U}^{p|\ul q}_{\za}|_{V_{\za\zb}}\to \mathcal{ U}^{p|\ul q}_{\zb}|_{V_{\zb\za}} ,\end{gather} where the source and target are the open $\Z_2^n$-submanifolds \begin{gather*} \cU^{p|\ul q}_\za|_{V_{\za\zb}}=\big(V_{\za\zb}, \Ci_{V_{\za\zb}}[[\zx]]\big)\end{gather*} (note that the underlying diffeomorphism is~\eqref{Diffeo}). A $p|\ul q$-\emph{atlas} over $|M|$ is a covering $\big(\mathcal{ U}^{p|\ul q}_\za,\!|\psi_\za|\big)_\za$ by charts together with a coordinate transformation~\eqref{CoordTrans} for each pair of charts, such that the usual {cocycle} condition $\psi_{\zb\zg}\psi_{\zg\za}=\psi_{\zb\za}$ holds (appropriate restrictions are understood).

\begin{Definition}\label{AtlasDefi} A (smooth) \emph{$\Z_2^n$-manifold} of dimension $p|\ul q$ is a (second-countable Hausdorff) smooth manifold $|M|$ together with a preferred $p|\ul q$-atlas over it.\end{Definition}

As in standard supergeometry, the Definitions~\ref{LRSDefi} and~\ref{AtlasDefi} are equivalent \cite{Leites:2011}. For instance, if $M=(|M|,\cO_M)$ is a $\Z_2^n$-manifold of dimension $p|\ul q$ in the sense of Definition~\ref{LRSDefi}, there are $\Z_2^n$-isomorphisms (isomorphisms of $\Z_2^n$-manifolds) \begin{gather*} h_\za=(|h_\za|,h^*_\za)\colon \ U_\za=(|U_\za|,\cO_M|_{|U_\za|})\to \cU_\za^{p|\ul q}=\big(\cU^p_\za,\Ci_{\R^p}|_{\cU^p_\za}[[\zx]]\big) ,\end{gather*} such that $(|U_\za|)_\za$ is an open cover of $|M|$. For any two indices $\za,\zb$, the restriction $h_\za|_{U_{\za\zb}}$ of $h_\za$ to the open $\Z_2^n$-submanifold $U_{\za\zb}=(|U_{\za\zb}|,\cO_M|_{|U_{\za\zb}|})$, $|U_{\za\zb}|=|U_\za|\cap |U_\zb|$, is a $\Z_2^n$-isomorphism between $U_{\za\zb}$ and
\begin{gather*} \cU_\za^{p|\ul q}|_{V_{\za\zb}}=\big(V_{\za\zb},\Ci_{\R^p}|_{V_{\za\zb}}[[\zx]]\big),\qquad V_{\za\zb}=|h_\za|(|U_{\za\zb}|) .\end{gather*} Therefore, the composite \begin{gather*} 
\psi_{\zb\za}=h_\zb|_{U_{\zb\za}}h_\za|_{U_{\za\zb}}^{-1}\end{gather*} is a $\Z_2^n$-isomorphism \begin{gather*} \psi_{\zb\za}\colon \ \cU_\za^{p|\ul q}|_{V_{\za\zb}}\to \cU_\zb^{p|\ul q}|_{V_{\zb\za}} ,\end{gather*} such that the cocycle condition is satisfied.

\newcommand{\Zn}{\Z_2^n}
As a matter of some formality, $\Z_{2}^{n}$-manifolds and their morphisms form a category. The category of $\Z_{2}^{n}$-manifolds we will denote as $\Z_{2}^{n}\catname{Man}$. We remark this category is locally small. Moreover, as shown in \cite[Theorem~19]{Bruce:2018b}, the category of $\Z_{2}^{n}$-manifolds admits (finite) products. More precisely, let $M_i$, $i\in\{1,2\}$, be $\Zn$-manifolds. Then there exists a $\Zn$-manifold $M_1 \times M_2$ and $\Zn$-morphisms $\pi_i\colon M_1 \times M_2 \to M_i$ (with underlying smooth manifold $|M_1\times M_2|=|M_1|\times |M_2|$ and with underlying smooth morphisms $|\pi_i|\colon |M_1|\times|M_2|\to|M_i|$ given by the canonical projections), such that for any $\Zn$-manifold $N$ and $\Zn$-morphisms $f_i\colon N \to M_i$, there exists a~unique morphism $h\colon N \to M_1 \times M_2$ making the obvious diagram commute. It follows that, if $\zvf\in\Hom_{\Zn\tt Man}(M,M')$ and $\psi\in\Hom_{\Zn\tt Man}(N,N')$, there is a unique morphism $\zvf\times\psi\in\Hom_{\Zn\tt Man}(M\times N,M'\times N')$.

\begin{Remark}It is known that an analogue of the Batchelor--Gaw\c{e}dzki theorem holds in the category of (real) $\Z_{2}^{n}$-manifolds, see \cite[Theorem~3.2]{Covolo:2016a}. That is, any $\Z_2^n$-manifold is noncanonically isomorphic to a $\Z_2^n\setminus \{\underline{0}\}$-graded vector bundle over a smooth manifold. While this result is quite remarkable, we will not exploit it at all in this paper.
\end{Remark}

\subsection{The functor of points}\label{FoP}

Similar to what happens in classical supergeometry, a $\Zn$-manifold $M$ is not fully described by its topological points in $|M|$. To remedy this defect, we broaden the notion of `point', as was suggested by Grothendieck in the context of algebraic geometry.

More precisely, set $V=\{z\in\C^n\colon P(z)=0\}\in{\tt Aff}$, where $P$ denotes a polynomial in~$n$ indeterminates with complex coefficients and $\tt Aff$ denotes the category of affine varieties. Grothendieck insisted on solving the equation $P(z)=0$ not only in $\C^n$, but in~$A^n$, for any algebra $A$ in the category $\tt CA$ of commutative (associative unital) algebras (over~$\C$). This leads to an arrow
\begin{gather*} \op{Sol}_P\colon \ {\tt CA}\ni A\mapsto\op{Sol}_P(A)=\big\{a\in A^n\colon P(a)=0\big\}\in{\tt Set},\end{gather*} which turns out to be a functor \begin{gather*} \op{Sol}_P\simeq \Hom_{\tt CA}(\C[V],-)\in[{\tt CA},{\tt Set}] ,\end{gather*} where $\C[V]$ is the algebra of polynomial functions of $V$. The dual of this functor, whose value $\op{Sol}_P(A)$ is the set of $A$-points of~$V$, is the functor
\begin{gather*} \Hom_{\tt Aff}(-,V)\in \big[{\tt Aff}^{\op{op}},{\tt Set}\big] ,\end{gather*} whose value $\Hom_{\tt Aff}(W,V)$ is the set of $W$-points of~$V$.

The latter functor can be considered not only in $\tt Aff$, but in any locally small category, in particular in $\Zn{\tt Man}$. We thus obtain a covariant functor (functor in $\bullet$) \begin{gather}\label{YonEmb}\bullet(-)=\Hom(-,\bullet)\colon \ \Zn{\tt Man}\ni M\mapsto M(-)=\Hom_{\Zn\tt Man}(-,M)\in \big[{\tt \Zn{\tt Man}^{\op{op}}, Set}\big] .\end{gather} As suggested above, the contravariant functor $\Hom(-,M)$ (we omit the subscript $\Zn{\tt Man}$) (functor in $-$) is referred to as the {\it functor of points} of $M$. If $S\in\Zn{\tt Man}$, an {\it $S$-point} of $M$ is just a morphism $\pi_S\in\Hom(S,M)$. One may regard an $S$-point of $M$ as a `family of points of $M$ parameterised by the points of~$S$'. The functor $\bullet(-)$ is known as the {\it Yoneda embedding}. For any underlying locally small category $\tt C$ (here ${\tt C}=\Zn{\tt Man}$), the functor $\bullet(-)$ is {\it fully faithful}, what means that, for any $M,N\in\Zn{\tt Man}$, the map
\begin{gather*} \bullet_{M,N}(-)\colon \ \Hom(M,N)\ni \phi\mapsto \Hom(-,\phi)\in\op{Nat}(\Hom(-,M),\Hom(-,N))\end{gather*} is bijective (here $\op{Nat}$ denotes the set of natural transformations). It can be checked that the correspondence $\bullet_{M,N}(-)$ is natural in~$M$ and in~$N$. Moreover, any fully faithful functor is automatically injective up to isomorphism on objects: $M(-)\simeq N(-)$ implies $M\simeq N$. Of course, the functor $\bullet(-)$ is not surjective up to isomorphism on objects, i.e., not every functor $X\in[\Zn{\tt Man}^{\op{op}},{\tt Set}]$ is isomorphic to a functor of the type $M(-)$. However, if such $M$ does exist, it is, due to the mentioned injectivity, unique up to isomorphism and it is called {\it `the' representing $\Zn$-manifold} of~$X$. Further, if $X,Y\in[\Zn{\tt Man}^{\op{op}},{\tt Set}]$ are two representable functors, represented by $M$, $N$ respectively, a morphism or natural transformation between them, provides, due to the mentioned bijectivity, a {\it unique morphism between the representing $\Zn$-manifolds}~$M$ and~$N$. It follows that, instead of studying the category $\Zn{\tt Man}$, we can just as well focus on the functor category $[{\tt \Zn{\tt Man}^{\op{op}}, Set}]$ (which has better properties, in particular it has all limits and colimits).
A \emph{generalized $\Z_2^n$-manifold} is an object in the functor category $[\Z_2^n\catname{Man}^{\textnormal{op}}, \catname{Set}]$ and morphisms of such objects are natural transformations. The category $[\Zn{\tt Man}^{\op{op}},\tt Set]$ of generalised $\Zn$-manifolds has finite products. Indeed, if $F,G$ are two generalized manifolds, we define the functor $F\times G$, given on objects $S$, by $(F\times G)(S)=F(S)\times G(S)$, and on morphisms $\Psi\colon S\to T$, by\looseness=-1
\begin{gather*} (F\times G)(\Psi)=F(\Psi)\times G(\Psi)\colon \ F(T)\times G(T)\to F(S)\times G(S) .\end{gather*}
It is easily seen that $F\times G$ respects compositions and identities. Further, there are canonical natural transformations $\zh_1\colon F\times G\to F$ and $\zh_2\colon F\times G\to G$. If now $(H,\za_1,\za_2)$ is another functor with natural transformations from it to~$F$ and~$G$, respectively, it is straightforwardly checked that there exists a unique natural transformation $\zb\colon H\to F\times G$, such that $\za_i=\zh_i\circ \zb$.
One passes from the category of $\Z_2^n$-manifolds to the larger category of generalised $\Z_2^n$-manifolds in order to understand, for example, the \emph{internal Hom} objects. In particular, there always exists a generalised $\Z_2^n$-manifold such that the so-called \emph{adjunction formula} holds
\begin{gather*}
\InHom_{\Z_2^n\catname{Man}}(M,N)(-) := \Hom_{\Z_2^n\catname{Man}}(- \times M,N) .
\end{gather*}
This internal Hom functor is defined on $\phi \in \Hom_{\Z_2^n\catname{Man}}(P, S)$ by
\begin{align*}
 \InHom_{\Z_2^n\catname{Man}}(M,N)(\phi) \colon \ \InHom_{\Z_2^n\catname{Man}}(M,N)(S)& \longrightarrow \InHom_{\Z_2^n\catname{Man}}(M,N)(P) ,\\
 \Psi_{S} &\longmapsto \Psi_{S}\circ (\phi\times \Id_{M}) .
\end{align*}
In general, a mapping $\Z_2^n$-manifold $\InHom_{\Z_2^n\catname{Man}}(M,N)$ will not be representable. We will refer to `elements' of a mapping $\Z_2^n$-manifold as \emph{maps} reserving morphisms for the categorical morphisms of $\Z_2^n$-manifolds.

Composition of maps between $\Zn$-manifolds is naturally defined as a natural transformation
\begin{gather*}
\underline{\circ} \colon \ \InHom(M,N) \times \InHom(N,L) \longrightarrow \InHom(M,L) ,
\end{gather*}
defined, for any $S \in \Z_2^n\catname{Man}$, by
\begin{align*}
 \Hom(S\times M, N) \times \Hom(S\times N,L)& \longrightarrow \Hom(S\times M, L), \\
 (\Psi_{S}, \Phi_{S})& \longmapsto (\Phi \underline{\circ} \Psi)_{S}:= \Phi_{S} \circ (\Id_{S}\times \Psi_{S})\circ (\Delta\times \Id_{M}) ,
\end{align*}
where $\Delta\colon S\longrightarrow S\times S$ is the diagonal of $S$ and $\Id_S\colon S\longrightarrow S$ is its identity.

Similarly to the cases of smooth manifolds and supermanifolds, morphisms between $\Z_2^n$-manifolds are completely determined by the corresponding maps between the global functions. We remark that this is not, in general, true for complex (super)manifolds. More carefully, we have the following proposition that was proved in \cite[Theorem~3.7]{Bruce:2018b}.

\begin{Proposition}\label{Global-to-local} Let $M = (|M|, \cO_M)$ and $N = (|N|, \cO_N)$ be $\Z_2^n$-manifolds. Then the natural map
\begin{gather*}\emph{\Hom}_{\Z_2^n\catname{Man}}\big(M,N \big) \longrightarrow \emph{\Hom}_{\Z_2^n\catname{Alg}}\big(\cO(|N|), \cO(|M|) \big) ,\end{gather*}
where $\tt\Z_2^nAlg$ denotes the category of $\Z_2^n$-graded $\Z_2^n$-commutative associative unital $\R$-algebras, is a bijection.
\end{Proposition}

It is worth recalling how a morphism $\psi\in\Hom_{\Z_2^n\catname{Alg}}\big(\cO(|N|), \cO(|M|))$ defines a continuous base map $|\phi|\colon |M|\to|N|$. We denote by $\zve_m\in\Hom_{\Z_2^n\catname{Alg}}\big(\cO(|M|),\R)$, $m\in|M|$, the morphism
\begin{gather*} \zve_m\colon \ \cO(|M|)\ni f\mapsto (\zve_{|M|}f)(m)\in\R ,\end{gather*} and by $\op{Spm}(\cO(|M|))$ the maximal spectrum of the algebra $\cO(|M|)$. The map \begin{gather*}\flat\colon \ |M|\ni m\mapsto \ker\zve_m\in \op{Spm}(\cO(|M|))\end{gather*}
 is a homeomorphism, both, when the target space is endowed with its Zariski topology and when it is endowed with its Gel'fand topology. The continuous map $|\phi|\colon |M|\to |N|$ that is induced by the morphism $\psi$ is now given by \begin{gather*} |\phi|\colon \ |M|\simeq \op{Spm}(\cO(|M|))\ni m\simeq \ker\zve_m\mapsto \ker(\zve_m\circ\psi)\simeq n\in \op{Spm}(\cO(|N|))\simeq |N| .\end{gather*}

The fact that the functor $\Hom_{\Z_2^n\tt Man}(S,-)$ respects limits and in particular products directly implies that
\begin{gather}\label{eqn:CartProdSpoints}
\big( M \times N \big)(S) \simeq M(S) \times N(S).
\end{gather}
The latter result is essential in dealing with $\Z_2^n$-Lie groups. A (super) Lie group can be defined as a group object in the category of smooth (super)manifolds. This leads us to the following definition.

\begin{Definition}\label{def:Z2nLieGroup}
A \emph{$\Z_2^n$-Lie group} is a group object in the category of $\Z_2^n$-manifolds.
\end{Definition}

A convenient fact here is that, if $G$ is a $\Z_2^n$-Lie group, then the set $G(S)$ is a group (see (\ref{eqn:CartProdSpoints})). In other words, $G(-)$ is a functor from $\Z_2^n\catname{Man}^{\rmo} \rightarrow \catname{Grp}$.

\begin{Remark}We leave details and examples of $\Z_2^n$-Lie groups for future publications. However, we will remark at this point that the idea of ``colour supergroup manifolds" has already appeared in the physics literature, albeit without a proper mathematical definition (see \cite{Aizawa:2018,Aizawa:2017, Rittenburg:1978a, Rittenburg:1978b}, for example). Another approach to $\Z_2^n$-Lie groups is via a generalisation of Harish-Chandra pairs (see~\cite{Mohammadi:2017} for work in this direction).
\end{Remark}

\section[$\Z_2^n$-points and the functor of points]{$\boldsymbol{\Z_2^n}$-points and the functor of points}\label{sec:Z2nPtsSVMemb}

In view of \eqref{YonEmb}, we need to `probe' a given $\Z_{2}^{n}$-manifold $M\simeq M(-)$ with \emph{all} $ \Z_{2}^{n}$-manifolds. We will show that this is however not the case, since, much like for the category of supermanifolds, we have a rather convenient generating set that we can employ, namely the set of $\Z_2^n$-points.

\subsection[The category of $\Z_2^n$-points]{The category of $\boldsymbol{\Z_2^n}$-points}

\begin{Definition}\label{def:superpoints}
A \emph{$\Z_2^n$-point} is a $\Z_2^n$-manifold $\R^{0|\ul{m}}$ with vanishing ordinary dimension. We denote by $\Z_2^n\tt Pts$ the full subcategory of $\Z_2^n\tt Man$, whose collection of objects is the (countable) set of $\Z_2^n$-points.
\end{Definition}

Morphisms $\zvf\colon \R^{0|\ul{m}}\to\R^{0|\ul{n}}$ of $\Z_2^n$-points are exactly morphisms $\phi^*\colon \Lambda^{\ul{n}} \rightarrow\Lambda^{\ul{m}}$ of $\Z_2^n$-Grassmann algebras:

\begin{Proposition}\label{prop:PtsGrAlg}There is an isomorphism of categories
\begin{gather*} \Z_2^n\catname{Pts}\simeq\Z_2^n\catname{GrAlg}^{\textnormal{op}} .\end{gather*}
\end{Proposition}

We can think of $\Z_2^n$-points as \emph{formal thickenings} of an ordinary point by the non-zero degree generators. The simplest $\Z_2^n$-point is the one with trivial formal thickening, $\R^{0|\ul{0}} := \big(\R^0, \R \big)$:
\begin{Proposition}
The $\Z_2^n$-point $\R^{0|\ul{0}}=\R^0$ is a terminal object in both, $\Z_2^n\catname{Man}$ and $\Z_2^n\catname{Pts}$.
\end{Proposition}

\begin{proof}The unique morphism $M \longrightarrow \R^{0|\ul{0}}$ corresponds to the morphism $\R\ni r\cdot 1\mapsto r \cdot \Id_M\in\cO_M(|M|)$, where $\Id_M$ is the unit function.\end{proof}

\begin{Proposition}\label{prop:DirSet}
The object set $\textnormal{Ob}(\Z_2^n\catname{Pts}) \simeq \textnormal{Ob}(\Z_2^n\catname{GrAlg})$ is a directed set.
\end{Proposition}

\begin{proof}Given any $\ul{m} = (m_1, m_2, \dots, m_N)$ and $\ul{n} = (n_1, n_2, \dots, n_N)$, we write $\zL^{\ul{m}} \leq \zL^{\ul{n}}$ if and only if $m_i \leq n_i$, for all~$i$. This preorder makes the non-empty set of $\Z_2^n$-Grassmann algebras into a directed set, since, any $\Lambda^{\ul{m}}$ and $\Lambda^{\ul{n}}$ admit $\zL^{\ul{p}}$, where $p_i=\sup\{m_i,n_i\}$, as upper bound.
\end{proof}

 We will need the following functional analytic result in later sections of this paper. See De\-fi\-ni\-tions~\ref{Appx:FSpace} and~\ref{Appx:FAlgebra} for the notion of Fr\'{e}chet space and Fr\'{e}chet algebra, respectively.

\begin{Proposition}\label{Prop:GAlgFrec}
The algebra of functions of any $\Z_2^n$-point is a $\Z_2^n$-graded $\Z_2^n$-commutative nuclear Fr\'{e}chet algebra.
\end{Proposition}

The proposition is a special case of the fact that the structure sheaf of \emph{any} $\Z_2^n$-manifold is a~nuclear Fr\'{e}chet sheaf of $\Z_2^n$-graded $\Z_2^n$-commutative algebras \cite[Theorem~14]{Bruce:2018}.

Moreover, as a direct consequence of \cite[Theorem~19, Definition~13]{Bruce:2018b}, we observe that the category of $\Z_2^n$-points admits all finite categorical products; in particular: $\R^{0|\ul{m}}\times \R^{0|\ul{n}} \simeq \R^{0|\ul{m}+\ul{n}}$.
By restricting attention to elements of degree $0 \in \Z_2^n$, we get the following corollary. See Definition~\ref{Appex:FModule} for the concept of Fr\'{e}chet module.

\begin{Corollary}\label{Coro:Lambda0Frec}
The set $\Lambda_0$ of degree $0$ elements of an arbitrary $\Z_2^n$-Grassmann algebra $\zL$ is a~commutative nuclear Fr\'{e}chet algebra. Moreover, the algebra $\Lambda$ can canonically be considered as a Fr\'{e}chet $\Lambda_{0}$-module.
\end{Corollary}

\begin{Remark}Specialising to the $n=1$ case, we recover the standard and well-known facts about superpoints and their relation with Grassmann algebras.
\end{Remark}

\subsection[A convenient generating set of $\Z_2^n\catname{Man}$]{A convenient generating set of $\boldsymbol{\Z_2^n\catname{Man}}$}

It is clear that studying just the underlying topological points of a $\Z_2^n$-manifold is inadequate to probe the graded structure. Much like the category of supermanifolds, where the set of superpoints forms a generating set, the set of $\Z_2^n$-points forms a generating set for the category of $\Z_2^n$-manifolds. For the classical case of standard supermanifolds, see for example \cite[Theorem~3.3.3]{Sachse:2007}. For the general notion of a generating set, see Definition~\ref{def:GeneratingSet}.

\begin{Theorem}\label{thm:generatingset}
The set $\textnormal{Ob} \big(\Z_2^n\catname{Pts}\big)$ constitutes a generating set for $\Z_2^n\catname{Man}$.
\end{Theorem}

\begin{proof}Let $\zvf=(|\zvf|,\zvf^*)$ and $\psi=(|\psi|,\psi^*)$ be two distinct $\Z_2^n$-morphisms $\zvf,\psi\colon M\to N$ between two $\Z_2^n$-manifolds $M=(|M|,\cO_M)$ and $N=(|N|,\cO_N)$. These morphisms have distinct smooth base maps \begin{gather*} |\zvf|,|\psi|\colon \ |M|\to|N|,\end{gather*} or, if $|\zvf|=|\psi|$, they have distinct pullback morphisms of sheaves of algebras \begin{gather*}\zvf^*,\psi^*\colon \ \cO_N\to |\zvf|_*\cO_M .\end{gather*}

If $|\zvf|\neq|\psi|$, there is at least one point $m\in|M|$, such that $|\zvf|(m)\neq|\psi|(m).$ Let now $s\colon \R^{0|\ul{0}}\to M$ be the $\Z_2^n$-morphism, which corresponds to the $\Z_2^n\tt Alg$ morphism $s^*\colon \cO_M(|M|)\ni f\mapsto(\zve f)(m)\in\R,$ where $\zve$ is the sheaf morphism $\zve\colon \cO_M\to\Ci_{|M|}$. It follows from the reconstruction theorem \cite[Theorem~9]{Bruce:2018b} that the base morphism $|s|\colon \{\star\}\to|M|$ maps $\star$ to $m$. Hence, the $\Z_2^n$-morphisms $\zvf\circ s$ and $\psi\circ s$ have distinct base maps.

Assume now that $|\zvf|=|\psi|$, so that there exists $|V|\subset |N|$, such that $\zvf^*_{|V|}\neq\psi^*_{|V|},$ i.e., such that $\zvf^*_{|V|}f\neq\psi^*_{|V|}f,$ for some function $f\in\cO_N(|V|)$. A cover of $|V|$ by coordinate patches $(\cV_i)_i$, induces a cover $|U_i|:=|\phi|^{-1}(\cV_i)$ of $|U|:=|\phi|^{-1}(|V|)$. It follows that \begin{gather*} (\zvf^*_{|V|}f)|_{|U_i|}\neq(\psi^*_{|V|}f)|_{|U_i|} ,\end{gather*} for some fixed $i$, i.e., that \begin{gather*} \zvf^*_{\cV_i}(f|_{\cV_i})\neq\psi^*_{\cV_i}(f|_{\cV_i}) ,\end{gather*} so that $\zvf^*_{\cV_i}\neq\psi^*_{\cV_i}$.

Recall that, for any open subset $|X|\subset|M|$, there is a $\Z_2^n$-morphism \begin{gather*}\iota_X\colon \ (|X|,\cO_M|_{|X|})\to (|M|,\cO_M),\end{gather*} whose base map $|\iota_X|$ is the inclusion and whose pullback $\iota_X^*$ is the obvious restriction. Further, any $\Z_2^n$-morphism $\phi\colon M\to N$, whose base map $|\phi|\colon |M|\to |N|$ is valued in an open subset $|Y|$ of $|N|$, induces a $\Z_2^n$-morphism \begin{gather*}\phi_Y\colon \ (|M|,\cO_M)\to(|Y|,\cO_N|_{|Y|}),\end{gather*} whose base map $|\phi_Y|$ is the map $|\phi|\colon |M|\to |Y|$ and whose pullback $\phi_Y^*$ is the pullback $\phi^*$ restricted to $\cO_N|_{|Y|}$.

In view of the above, if $(\cU_j)_j$ is a cover of $|U_i|$ by coordinate domains, we have \begin{gather}\label{dist}(\zvf^*_{\cV_i}(f|_{\cV_i}))|_{\cU_j}\neq(\psi^*_{\cV_i}(f|_{\cV_i}))|_{\cU_j},\end{gather} for some fixed $j$. This implies that the $\Z_2^n$-morphisms $(\phi\circ\iota_{\cU_j})_{\cV_i}$ and $(\psi\circ\iota_{\cU_j})_{\cV_i}$ from the $\Z_2^n$-domain $\cU_j=(\cU_j,\Ci_{\cU_j}[[\zx]])$ to the $\Z_2^n$-domain $\cV_i=(\cV_i,\Ci_{\cV_i}[[\theta]])$ are different. More precisely, they have the same base map $|\phi|=|\psi|\colon \cU_j\to\cV_i$, but their pullbacks are distinct. Indeed, these sheaf morphisms' algebra maps at $\cV_i$ are the maps $\iota^*_{\cU_j,|U_i|}\circ \phi^*_{\cV_i}$ and $\iota^*_{\cU_j,|U_i|}\circ \psi^*_{\cV_i}$ from $\Ci_{\cV_i}(y)[[\theta]]$ to $\Ci_{\cU_j}(x)[[\zx]]$, where $y$ runs through $\cV_i$ and $x$ through $\cU_j$, and the values of these algebra maps at $f|_{\cV_i}$ are different (see equation~\eqref{dist}).

In view of Lemma~\ref{LemGenerating}, there is a $\Z_2^n$-morphism $s\colon \R^{0|\ul{m}}\to \cU_j$, such that \begin{gather*}(\phi\circ\iota_{\cU_j})_{\cV_i}\circ s\neq(\psi\circ\iota_{\cU_j})_{\cV_i}\circ s.\end{gather*} However, then the $\Z_2^n$-morphism $\iota_{\cU_j}\circ s\colon \R^{0|\ul{m}}\to M$ separates $\phi$ and $\psi$, since the algebra maps at $\cV_i$ of the pullbacks $(s^*\circ\iota_{\cU_j}^*)\circ\phi^*$ and $(s^*\circ\iota_{\cU_j}^*)\circ\psi^*$ differ. Indeed, as the $\Z_2^n$-morphisms $(\phi\circ\iota_{\cU_j})_{\cV_i}$ and $(\psi\circ\iota_{\cU_j})_{\cV_i}$ are fully determined by the pullbacks of the target coordinates, their pullbacks at $\cV_i$ differ for at least one coordinate~$y^b$,~$\theta^B$. It follows from the proof of Lemma~\ref{LemGenerating} that the pullback $s^*_{\cU_j}\circ(\iota_{\cU_j,|U_i|}^*\circ\phi^*_{\cV_i})$ at $\cV_i$ of $(\phi\circ\iota_{\cU_j})_{\cV_i}\circ s$ and the similar pullback for $\psi$ differ for the same coordinate. However, the pullback at $\cV_i$ considered is also the algebra map at $\cV_i$ of the pullback $(s^*\circ\iota_{\cU_j}^*)\circ\phi^*$, so that the pullbacks $(s^*\circ\iota_{\cU_j}^*)\circ\phi^*$ and $(s^*\circ\iota_{\cU_j}^*)\circ\psi^*$ are actually distinct.
\end{proof}

It remains to prove the following

\begin{Lemma}\label{LemGenerating} The statement of Theorem~{\rm \ref{thm:generatingset}} holds for any two distinct $\Z_2^n$-morphisms between $\Z_2^n$-domains.\end{Lemma}
\begin{proof}

We consider two $\Z_2^n$-domains $\mathcal{U}^{p|\ul{q}}$ and $\mathcal{V}^{r|\ul{s}}$ together with two distinct $\Z_2^n$-morphisms
 \begin{gather*} \mathcal{U}^{p|\ul{q}} \mathrel{\mathop{ \overset{\longrightarrow}{\longrightarrow} }^{\phi}_{\psi}} \mathcal{V}^{r|\ul{s}}.\end{gather*}
As in the general case above, there are two cases to consider: either $|\phi|\neq |\psi|$, or $|\phi| = |\psi|$ and $\phi^*\neq \psi^*$. In the proof of Theorem~\ref{thm:generatingset}, we showed that in the first case, the maps $\phi$ and $\psi$ can be separated. In the second case, since a $\Z_2^n$-morphism valued in a $\Z_2^n$-domain is fully defined by the pullbacks of the coordinates, these global $\Z_2^n$-functions $\phi^*_{\cV^r}(Y^i), \psi^*_{\cV^r}(Y^i)\in\Ci_{\cU^p}(x)[[\zx]]$ differ for at least one coordinate $Y^i=y^b$ or $Y^i=\theta^B$. Let $B$ be an index, such that
\begin{gather*}
\phi^*_{\cV^r}\big(\theta^B\big) = \sum_{|\za|=1}^\infty \phi_\za^B(x) \zx^{\za},\qquad
\psi^*_{\cV^r}\big(\theta^B\big) = \sum_{|\za|=1}^\infty \psi_\za^B(x) \zx^{\za},
\end{gather*}
where we denoted the coordinates of $\cU^{p|\ul{q}}$ by $\big(x^a, \zx^A\big)$ and used the standard multi-index notation, differ. This means that the functions $\phi^B_\za(x)$ and $\psi^B_\za(x)$ differ for at least one $\za$ and at least one $x\in\cU^p$, say for $\za=\mathfrak{a}$ and $x=\mathfrak{x}\in\cU^p\subset\R^p$. From this, we can construct the separating $\Z_2^n$-morphism
 \begin{gather*} \R^{0|\ul{q}} \stackrel{s}{\longrightarrow} \mathcal{U}^{p|\ul{q}} \mathrel{\mathop{ \overset{\longrightarrow}{\longrightarrow} }^{\phi}_{\psi}} \mathcal{V}^{r|\ul{s}}.\end{gather*}
Let us denote the coordinates of $\R^{0|\ul{q}}$ by $\chi^A$. We then define the $\Z_2^n$-morphism $s$ by setting
 \begin{alignat*}{3}
& s_{\cU^p}^*x^a = \mathfrak{x}^a\in\R[[\chi]], \qquad && \deg\big(\mathfrak{x}^a\big)=\deg\big(x^a\big),&\\
& s_{\cU^p}^*\zx^A = \chi^A\in\R[[\chi]], \qquad && \deg\big(\chi^A\big)=\deg\big(\zx^A\big).&
 \end{alignat*}
It is clear that $\phi\circ s\neq\psi\circ s$, since \begin{gather*}\sum_{|\za|=1}^\infty \phi_\za^B(\mathfrak{x})\chi^\za=s_{\cU^p}^*\big(\phi_{\cV^r}^*\big(\theta^B\big)\big)\neq s_{\cU^p}^*\big(\psi_{\cV^r}^*\big(\theta^B\big)\big)=\sum_{|\za|=1}^\infty \psi_\za^B(\mathfrak{x})\chi^\za.\end{gather*}
The case where $\phi^*_{\cV^r}(Y^i)\neq\psi^*_{\cV^r}(Y^i)$ for $Y^i=y^b$ is almost identical. In particular, we then have
 \begin{gather*}
 \phi_{\cV^r}^*\big(y^b\big) = |\phi|^b(x) + \sum_{|\za|=2}^\infty\phi_\za^b(x)\zx^\za,\\
 \psi_{\cV^r}^*\big(y^b\big) = |\psi|^b(x) + \sum_{|\za|=2}^\infty\psi_\za^b(x)\zx^\za.
\end{gather*}
Since we know that $|\phi|=|\psi|$, we can proceed as for $Y^i=\theta^B$.
\end{proof}

In view of Proposition \ref{RestYonProp}, we get the
\begin{Corollary}\label{cor:RestYon}
The restricted Yoneda functor
\begin{gather*}\mathcal{Y}_{\Z_2^n\catname{Pts}}\colon \ \Z_2^n\catname{Man}\ni M\mapsto \emph{\Hom}_{\Z_2^n\catname{Man}}\big(-, M\big)\in \big [\Z_2^n\catname{Pts}^{\rmo},\catname{Set} \big]\end{gather*}
is faithful.
\end{Corollary}

Above, we wrote $M(-)\in[\Z_2^n{\tt Man}^{\op{op}},{\tt Set}]$ for the image of $M\in\Z_2^n\tt Man$ by the non-restricted Yoneda functor. If no confusion arises, we will use the same notation $M(-)$ for the image $\cY_{\Z_2^n\tt Pts}(M)\in[\Z_2^n{\tt Pts}^{\op{op}},{\tt Set}]$ of $M$ by the restricted Yoneda functor.

\begin{Definition}\label{def:LambdaPtsM} Let $M$ be an object of $\Z_2^n\tt Man$ and $\Lambda\simeq\R^{0|\ul{m}}$ an object of $\Z_2^n{\tt GrAlg}\simeq\Z_2^n{\tt Pts}^{\op{op}}$. We refer to the set
\begin{gather*}
M(\Lambda):= \Hom_{\Z_2^n\catname{Man}}\big(\R^{0|\ul{m}}, M\big) \simeq \Hom_{\Z_2^n\catname{Alg}} \big ( \cO(|M|), \Lambda\big)\end{gather*} as the set of \emph{$\Lambda$-points} of $M$.
\end{Definition}

\begin{Proposition}\label{prop:LocalNature}Let \begin{gather*}m^* \in \emph{\Hom}_{\Z_2^n \catname{Alg}}\big( \cO(|M|), \Lambda\big)\end{gather*} be a $\zL$-point of $M$ and let $s \in \cO(|M|)$. The $\zL$-point $m^*$ can equivalently be viewed as a $\Z_2^n$-morphism \begin{gather*}m=(|m|,m^*)\in\emph{\Hom}_{\Z_2^n\tt Man}\big(\R^{0|\ul{m}},M\big) \end{gather*} and therefore it defines a unique topological point $x:=|m|(\star)\in|M|$. If $|U|\subset|M|$ is an open neighbourhood of $x$, such that $s|_{|U|} =0$, then $m^*(s)=0$.
\end{Proposition}

\begin{proof}Since $m^*\colon \cO_M\to \cO_{\R^{0|\ul m}}$ is a sheaf morphism, it commutes with restrictions, i.e., for any open subsets $|V|\subset |U|\subset |M|$ and any $s\in\cO_M(|U|)$, we have $m^*_{|U|}(s)\in\cO_{\R^{0|\ul m}}\big(|m|^{-1}(|U|)\big)$ and
\begin{gather*}(m^*_{|U|}(s))|_{|m|^{-1}(|V|)}= m^*_{|V|}(s|_{|V|})\in\cO_{\R^{0|\ul m}}\big(|m|^{-1}(|V|)\big).\end{gather*}
It follows that $m^*(s)=m^*_{|M|}(s)\in\zL=\cO_{\R^{0|\ul m}}(\{\star\})$ reads
\begin{gather*}m^*(s)=(m^*_{|M|}(s))|_{\{\star\}}=(m^*_{|M|}(s))|_{|m|^{-1}(|U|)}=m^*_{|U|}(s|_{|U|})=0.\tag*{\qed}\end{gather*}\renewcommand{\qed}{}
\end{proof}

\begin{Lemma}\label{RestTarget} There is a $1:1$ correspondence \begin{gather*}M(\Lambda) \simeq \bigcup_{x\in |M|} \emph{\Hom}_{\Z_2^n\catname{Alg}}\big( \cO_{M,x} , \Lambda \big) \end{gather*} between the set of $\zL$-points of $M$ and the set of morphisms from the stalks of $\cO_M$ to $\zL$. The set \begin{gather*}M_x(\Lambda) := \emph{\Hom}_{\Z_2^n\catname{Alg}}\big( \cO_{M,x} , \Lambda \big)\end{gather*} is referred to as the set of $\Lambda$-points near~$x$.\end{Lemma}

\begin{proof}Any $\Lambda$-point $m^*$ or $m=(|m|,m^*)$ defines a topological point $x=|m|(\star)\in|M|$, as well as a $\Z_2^n\tt Alg$-morphism $\phi_x\in\Hom_{\Z_2^n\tt Alg}(\cO_{M,x},\zL)$ between stalks. This morphism is given, for any $t_U\in\cO(|U|)$ defined in some neighbourhood $|U|$ of $x$ in $|M|$, by \begin{gather*}\phi_x[t_U]_x=m^*_\star[t_U]_x=[m^*_{|U|}t_U]_\star=m^*_{|U|}t_U.\end{gather*}

Conversely, any morphism $\psi_y\in\Hom_{\Z_2^n\tt Alg}(\cO_{M,y},\zL)$ ($y\in|M|$) between stalks defines a $\zL$-point $\zm^*\in\Hom_{\Z_2^n \catname{Alg}}\big( \cO(|M|),$ $\Lambda\big)$. It suffices to set \begin{gather*}\zm^*t=\psi_y[t]_y\in\zL,\end{gather*} for all $t\in\cO(|M|)$.

It remains to check that the composites $m^*\mapsto \phi_x\mapsto \zm^*$ and $\psi_y\mapsto \zm^*\mapsto \phi_x$ are identities. In the first case, for any $t\in\cO(|M|)$, we get $\zm^*t=\phi_x[t]_x=m^*t$, so that $\zm^*=m^*$. In the second case, we need the following reconstruction results. Let $|U|\subset|M|$ be an open subset and set \begin{gather*}S_U=\big\{s\in\cO^0(|M|)\colon (\zve s)|_{|U|} \ \text{is invertible in} \ \Ci(|U|)\big\}.\end{gather*} Then the localization map $\zl_U\colon \cO(|M|)\cdot S_U^{-1}\to \cO(|U|)$ is an isomorphism in $\Z_2^n\tt Alg$. More precisely, for any $t_U\in\cO(|U|)$, there is a unique $Fs^{-1}\in\cO(|M|)\cdot S_U^{-1}$, such that $t_U=F|_{|U|}s|_{|U|}^{-1}$ (if $s\in S_U$, then $s|_{|U|}$ is invertible in $\cO(|U|)$), and we identify $Fs^{-1}$ with $t_U$. For the proof of these statements or more details on them, see \cite[Proposition~3.5.]{Bruce:2018b}. It is further clear from the results of \cite[Proposition~3.1.]{Bruce:2018b} that $x=|\zm|(\star)$ is the topological point~$y$.

We now compute the second composite above. For any $t_U$ defined in a neighborhood $|U|$ of~$x$, we get
\begin{gather*}\phi_x[t_U]_x=\zm^*_{|U|}\big(Fs^{-1}\big)=\zm^*(F) \zm^*(s)^{-1}\\
\hphantom{\phi_x[t_U]_x}{}= \psi_x[F]_x (\psi_x[s]_x)^{-1}=\psi_x[F]_x \psi_x\big([s]_x^{-1}\big)=\psi_x\big([F|_{|U|}]_x\big[s|_{|U|}^{-1}\big]_x\big)=\psi_x[t_U]_x,\end{gather*}
where the second equality is part of the reconstruction theorem of $\Z_2^n$-morphisms~\cite{Bruce:2018b}.
\end{proof}

Let us consider an open cover $(|U_I|)_{I \in \A}$ of the smooth manifold $|M|$, as well as the open $\Z_2^n$-submanifolds $U_I := \big(|U_I|, \cO_M|_{|U_I|} \big)$ of the $\Z_2^n$-manifold $M$ (which need \emph{not} be coordinate charts).

\begin{Proposition}\label{prop:CoverM}For any $\Z_2^n$-Grassmann algebra $\Lambda$ and $\Z_2^n$-manifold $M = \big( |M|, \cO_M\big)$, we have a natural $1:1$ correspondence
\begin{gather*}M(\Lambda) \simeq \bigcup_{I\in \A} U_I(\Lambda),\end{gather*}
so that the family of sets $(U_I(\Lambda))_{I \in \A}$ is a cover of the set $M(\Lambda)$.
\end{Proposition}

\begin{proof}Since it is clear from the definition of a stalk that $\cO_{U_I,x} = \cO_{M,x}$, for any $x\in|U_I|$, it follows from Lemma~\ref{RestTarget} that
\begin{gather*}
\bigcup_{I \in \A} U_{I}(\Lambda) \simeq \bigcup_{I \in \A}\bigcup_{x\in |U_I|} \Hom_{\Z_2^n\catname{Alg}}\big( \cO_{M,x} , \Lambda \big) = \bigcup_{x\in |M|}\Hom_{\Z_2^n\catname{Alg}}\big( \cO_{M,x} , \Lambda \big) \simeq M(\Lambda).\tag*{\qed}
\end{gather*}\renewcommand{\qed}{}
\end{proof}

Recall that \begin{gather*}\Hom_{\Z_2^n\tt Man}(-,-)\in\big[\Z_2^n{\tt Man}, \big[\Z_2^n{\tt Pts}^{\op{op}},{\tt Set}\big]\big],\end{gather*} so that,
\begin{enumerate}\itemsep=0pt
\item[(i)] any $\Z_2^n$-morphism $\phi=(|\phi|,\phi^*):M\to N$ is mapped (injectively) to a natural transformation
\begin{gather*}\phi\simeq\Hom_{\Z_2^n\tt Man}(-,\phi)\colon \ \Hom_{\Z_2^n\tt Man}(-,M)\to \Hom_{\Z_2^n\tt Man}(-,N),\end{gather*}
whose $\zL$-component ($\zL\simeq\R^{0|\ul{m}}$) is the $\tt Set$-map given by
\begin{gather*}
\phi_\zL:=\Hom_{\Z_2^n\tt Man}(\zL,\phi)\colon \ M(\zL)=\Hom_{\Z_2^n\tt Man}\big(\R^{0|\ul{m}},M\big)\simeq\Hom_{\Z_2^n\tt Alg}(\cO(|M|),\zL)\ni m^*\\
\hphantom{\phi_\zL}{} \mapsto
m^*\circ \phi^*\in\Hom_{\Z_2^n\tt Alg}(\cO(|N|),\zL)\simeq\Hom_{\Z_2^n\tt Man}\big(\R^{0|\ul{m}},N\big)=N(\zL), \qquad \text{and}, \end{gather*}
\item[(ii)] for any fixed $M\in\Z_2^n\tt Man$, given a morphism $\psi=(|\psi|,\psi^*)\colon \R^{0|\ul{m}'}\to \R^{0|\ul{m}}$ of $\Z_2^n$-points, or, equivalently, a morphism $\psi^* \colon \Lambda \rightarrow \Lambda'$ of $\Z_2^n$-Grassmann algebras, we get the induced $\tt Set$-map
\begin{gather} M(\psi^*):=\Hom_{\Z_2^n\tt Man}(\psi,M)\colon \ M(\zL)=\Hom_{\Z_2^n\tt Man}\big(\R^{0|\ul{m}},M\big)\nonumber\\
\hphantom{M(\psi^*)}{} \simeq\Hom_{\Z_2^n\tt Alg}(\cO(|M|),\zL)\ni m^* \mapsto \psi^*\circ m^*\in\Hom_{\Z_2^n\tt Alg}(\cO(|M|),\zL')\nonumber\\
\hphantom{M(\psi^*)}{} \simeq\Hom_{\Z_2^n\tt Man}\big(\R^{0|\ul{m}'},M\big)=M(\zL'). \label{AlgMorLbdPts}\end{gather}
\end{enumerate}
When reading the maps $\phi_\zL$ and $M(\psi^*)$ through the $1:1$ correspondence
\begin{gather*}M(\Lambda)\ni m^* \mapsto (x,m^*_\star)\in\bigcup_{y\in |M|} {\Hom}_{\Z_2^n\catname{Alg}}\big( \cO_{M,y} , \Lambda \big), \end{gather*} where $x=|m|(\star)$, we obtain
\begin{gather*}
\phi_\Lambda \colon \ M(\Lambda) \longrightarrow N(\Lambda),\\
\hphantom{\phi_\Lambda \colon}{} \ (x, m_\star^*) \mapsto (|\phi|(x), m_\star^* \circ \phi^*_x), \qquad \text{and}
\\ 
M(\psi^*)\colon \ M(\Lambda) \longrightarrow M(\Lambda'),\\
\hphantom{M(\psi^*)\colon}{} \ (x, m_\star^*) ~ \mapsto (x, \psi^* \circ m_\star^*).
\end{gather*}

\subsection{Restricted Yoneda functor and fullness}

\newcommand{\mr}{\mathrm}
\newcommand{\rim}{\mathring}

The Yoneda functor from any locally small category $\tt C$ into the category of $\tt Set$-valued contravariant functors on $\tt C$, is fully faithful. This holds in particular for ${\tt C}=\Z_2^n\tt Man$. When we restrict the contravariant functors to the generating set $\Z_2^n\tt Pts$, the resulting restricted Yoneda functor is automatically faithful. In the following, we show that it is not full, i.e., that not all natural transformations are induced by a $\Z_2^n$-morphism.

Naturality of any transformation $\phi : M(-) \rightarrow N(-) $ between $\tt Set$-valued contravariant (resp., covariant) functors on $\Z_2^n\tt Pts$ (resp., $\Z_2^n\tt GrAlg$), means that the diagram
\begin{gather*}
\leavevmode
\begin{xy}
(0,20)*+{ M(\Lambda)}="a"; (40,20)*+{N(\Lambda)}="b";
(0,0)*+{M(\Lambda')}="c"; (40,0)*+{N(\Lambda')}="d";%
{\ar "a";"b"}?*!/_3mm/{\phi_\Lambda};%
{\ar "a";"c"}?*!/^6mm/{M(\psi^*)};%
{\ar "b";"d"}?*!/_6mm/{N(\psi^*)};%
{\ar "c";"d"}?*!/^3mm/{\phi_{\Lambda'}};%
\end{xy}
\end{gather*}
commutes, for any morphism $\psi^*\colon \Lambda \rightarrow \Lambda'$ of $\Z_2^n$-Grassmann algebras.

A $\Lambda$-point of a $\Z_2^n$-manifold $M$ is denoted by $m^*$ or $m=(|m|,m^*)$. If the manifold is a~$\Z_2^n$-domain $\cU^{p|\ul{q}}$, we use the notation $\mathrm{x}^*$ or $\mathrm{x}=(|\mathrm{x}|,\mathrm{x}^*)$. If $\big(x^a,\zx^A\big)$ are the coordinates of $\mathcal{U}^{p|\ul{q}}$, a~$\Lambda$-point $\mathrm{x}^*$ in $\cU^{p|\ul{q}}$ is completely determined by the degree-respecting pullbacks \begin{gather*}\big(x^a_\zL,\zx^A_\zL\big):=\big(\mathrm{x}^*\big(x^a\big),\mr{x}^*\big(\zx^A\big)\big).\end{gather*} Since $x^a_\zL\in\zL_0=\R\oplus\mathring{\Lambda}_0$, we write $x^a_\zL=\big(x^a_{||},\mathring{x}^a_\zL\big)$. Hence, any $\zL$-point $\mr{x}^*$ in $\cU^{p|\ul{q}}$ can be identified with
\begin{gather} \label{eqn:LambdaCoords}
\mr{x}^*\simeq\big(x^a_\zL,\zx^A_\zL\big)=\big(x^a_{||},\mathring{x}_\Lambda^a,\zx^A_\Lambda \big) \in \R^p \times \mathring{\Lambda}_{0}^p \times \mathring{\Lambda}_{\gamma_1}^{q_1} \times \dots \times \mathring{\Lambda}_{\gamma_N}^{q_N},
\end{gather}
where \begin{gather*}x_{||}=\big(x^a_{||}\big)=\big(\dots, x^a_{||}, \dots\big)\in\cU^p,\end{gather*} and where $\zg_1,\dots,\zg_N$ denote the non-zero $\Z_2^n$-degrees in standard order. Here the $\mathring{x}_\Lambda^a$ (resp., the~$\zx^A_\Lambda$) are formal power series containing at least $2$ (resp., at least $1$) of the generators $\big(\theta^C\big)$ of the $\Z_2^n$-Grassmann algebra $\Lambda$.

As mentioned above, any $\Z_2^n$-morphism, in particular any morphism $\phi\colon \mathcal{U}^{p|\ul{q}} \rightarrow \mathcal{V}^{r|\ul{s}}$ between $\Z_2^n$-domains, naturally induces a natural transformation, with $\zL$-component
\begin{gather*}\phi_\Lambda \colon \ \mathcal{U}^{p|\ul{q}}(\Lambda)\ni \mr{x}^* \mapsto \mr{x}^*\circ\phi^*\in \mathcal{V}^{r|\ul{s}}(\Lambda).\end{gather*}
If $\big(y^b,\zh^B\big)$ are the coordinates of $\cV^{r|\ul{s}}$, the morphism $\phi$ reads
\begin{gather*}
 \phi^*\big(y^b\big)=\sum_{|\za|\ge 0}\phi_\za^b(x) \zx^\za,\qquad
\phi^*\big(\zh^B\big)=\sum_{|\za|> 0}\phi_\za^B(x) \zx^\za,
\end{gather*}
where the right-hand sides have the appropriate degrees and where $\phi_0(\cU^p)\subset\cV^r$. Further, the image $\zL$-point $\mr{x}^*\circ\phi^*$ in $\cV^{r|\ul{s}}$ by $\phi_\zL$ of the $\zL$-point $\mr{x}^*\simeq\big(\mr{x}^*(x^a);\mr{x}^*\big(\zx^A\big)\big)=\big(x^a_{||},\mathring{x}_\Lambda^a;\zx^A_\Lambda \big)$ in $\cU^{p|\ul{q}}$, is given by
\begin{subequations}
\begin{gather}
 y^b_\Lambda = \sum_{|\za|\ge 0}\sum_{|\zb|\ge 0}\frac{1}{\zb!} \big(\partial_{x}^\zb\phi_\za^b\big)(x_{||}) \mathring{x}_\zL^\zb \zx_\zL^\za, \label{eqn:smoothzero}\\
 \zh^B_\Lambda =\sum_{|\za|> 0}\sum_{|\zb|\ge 0}\frac{1}{\zb!} \big(\partial_x^\zb\phi_\za^B)(x_{||}\big) \mathring{x}_\zL^\zb \zx_\zL^\za. \label{eqn:smoothnonzero}
\end{gather}
\end{subequations}
Let us recall that there is no convergence issue with terms in $x_{||}$~\cite{Covolo:2016}. Thus the components of a~natural transformation implemented by a $\Z_2^n$-morphism between $\Z_2^n$-domains, are very particular formal power series in the formal variables $\mathring{x}^a_{\Lambda}$ and $\zx_\Lambda^A$, which are themselves formal power series in the generators $\big(\theta^C\big)$ of $\Lambda$.

We are now able to prove that not all natural transformations between the restricted functors $M(-),N(-)\in[\Z_2^n\tt Pts,Set]$ associated with $M,N\in\Z_2^n\tt Man$, arise from a $\Z_2^n$-morphism $M\to N$. Since it suffices to give one counter-example, we choose $M=N=\R^{p|\ul{0}}=\R^p$.

\begin{Example}Consider an arbitrary diffeomorphism $\phi\colon \R^p\longrightarrow \R^p$. The $\zL$-component of the associated natural transformation is
\begin{gather*}
\phi_\Lambda \colon \ \R^{p|\ul{0}}(\Lambda) \longrightarrow \R^{p|\ul{0}}(\Lambda),\\
 \hphantom{\phi_\Lambda \colon}{} \ (x^b_\Lambda, 0)\mapsto \bigg(\phi^b(x_{||})+\sum_{|\zb|> 0}\frac{1}{\zb!} (\partial_{x}^\zb\phi^b)(x_{||}) \mathring{x}_\zL^\zb, \, 0\bigg).
\end{gather*}
From this data we obtain another natural transformation
\begin{gather*}
\alpha_\Lambda \colon \ \R^{p|\ul{0}}(\Lambda) \longrightarrow \R^{p|\ul{0}}(\Lambda),\\
\hphantom{\alpha_\Lambda \colon}{} \ \big(x^b_\Lambda, 0\big) \mapsto \big(\phi^b(x_{||}), \, 0\big).
\end{gather*}
The natural transformation $\alpha$ is not implemented by a morphism $\psi\colon \R^p\to\R^p$. Indeed, otherwise $\za_\zL=\psi_\zL$, for all $\zL$. This means that
\begin{gather*}
\big(\phi^b(x_{||}), 0\big) = \bigg(\psi^b(x_{||})+\sum_{|\zb|> 0}\frac{1}{\zb!} (\partial_{x}^\zb\psi^b)(x_{||}) \mathring{x}_\zL^\zb, \, 0\bigg),
\end{gather*}
for all $\zL$ and all $\zL$-points. Since $\phi^b(x)\equiv\psi^b(x)$, we have $\partial_x^\zb\phi^b\equiv\partial_x^\zb\psi^b$. Take now any $\zb\colon |\zb|=1$, so that $\zb_a=1$, for some fixed $a\in\{1,\ldots,p\}$. As we can choose $\zL$ and $x_\zL^b$, for all $b\in\{1,\ldots,p\}$, arbitrarily, we can choose $\rim{x}_\zL^b=0$, for all $b\neq a$, and $\rim{x}_\zL^a=\zvy^D\zvy^E$, where $\zvy^D$ and $\zvy^E$ are two different generators of $\zL$ that have the same degree. The coefficient of $\zvy^D\zvy^E$ in the sum over all~$\zb$ is then $\big(\partial_{x^a}\psi^b\big)(x_{||})$, hence $\partial_{x^a}\phi^b\equiv\partial_{x^a}\psi^b\equiv 0$. The latter observation is a contradiction, since the Jacobian determinant of $\phi$ does not vanish anywhere in $\R^p$.
\end{Example}

We now generalise a technical result \cite[Theorem~1]{Voronov:1984} to $\Z_2^n$-domains $\cU^{p|\ul{q}}$. Let
\begin{gather*}\mathcal{B}_{p |\ul{q}}\big(\mathcal{U}^p\big) := \mathcal{F}\big(\mathcal{U}^p, \R\big)[[X, \Xi]],\end{gather*}
be the $\Z_2^n$-graded $\Z_2^n$-commutative associative unital $\R$-algebra of formal power series in $p$ parameters $X^a$ of $\Z_2^n$-degree $0$ and $q_1,\dots,q_N$ parameters $\Xi^A$ of non-zero $\Z_2^n$-degree $\gamma_1,\dots,\zg_N$, and with coefficients in arbitrary $\R$-valued functions on $\mathcal{U}^p$, i.e., we do \emph{not} ask that these functions be continuous let alone smooth. Following \cite{Schwarz:1982,Schwarz:1984,Voronov:1984}, we will refer to this algebra as a $\Z_2^n$-\emph{Berezin algebra}. Any element of this algebra is of the form
\begin{gather}\label{NatXXI}F = \sum_{|\za|\ge 0}\sum_{|\zb|\ge 0}F_{\za\zb}(x)X^\zb\Xi^\za,\end{gather}
where the $x^a$ are coordinates in $\mathcal{U}^p$.

\begin{Theorem}\label{prop:NaturalTrans}
For any $\Z_2^n$-domains $\mathcal{U}^{p|\ul{q}}$ and $\mathcal{V}^{r|\ul{s}}$, there is a $1:1$ correspondence \begin{gather*}\op{Nat}\big(\cU^{p,\ul{q}},\cV^{r,\ul{s}}\big)\to\big(\mathcal{B}_{p|\ul{q}}\big(\mathcal{U}^p\big)\big)^{r|\ul{s}} \end{gather*} between
\begin{itemize}\itemsep=0pt
\item[--] the set of natural transformations in $[\Z_2^n\catname{Pts}^{\textnormal{op}}, \catname{Set}]$ between $\mathcal{U}^{p|\ul{q}}(-)$ and $\mathcal{V}^{r|\ul{s}}(-)$, and
\item[--] the set of `vectors' $\mathbf{F}$ with $r$ $($resp., with $s_1,\ldots,s_N)$ components $F^b$ of degree $0$ $($resp., components $F^B$ of degrees $\zg_1,\ldots,\zg_N)$ of the type \eqref{NatXXI}, such that the $r$-tuple $\big(F^b_{00}\big)$ made of the coefficients $F^b_{00}(x)$ of the $r$ series $F^b$ satisfies
\begin{gather*}\big(F^b_{00}\big)\big(\cU^p\big)\subset\cV^r.\end{gather*}
\end{itemize}
\end{Theorem}

\begin{proof}
Let $\mathbf{F}$ be such a `vector'. For any $\zL$, we define the map \begin{gather*}\zb_\zL\colon \ \cU^{p|\ul{q}}(\zL)\ni \big(x^a_{||},\mathring{x}^a_\zL,\zx^A_\zL\big)\mapsto \big(y^b_\zL,\zh^B_\zL\big)\in\cV^{r|\ul{s}}(\zL),\end{gather*} where
\begin{gather*}\label{NatBeta}
y^b_\Lambda := \sum_{|\za|\ge 0}\sum_{|\zb|\ge 0}F^b_{\za\zb}(x_{||}) \mathring{x}_\zL^\zb \zx_\zL^\za\qquad\text{and}\qquad \zh^B_\Lambda := \sum_{|\za|\ge 0}\sum_{|\zb|\ge 0}F^B_{\za\zb}(x_{||}) \mathring{x}_\zL^\zb \zx_\zL^\za .
\end{gather*}
Since $\mathring{x}^a_\zL$, $\zx^A_\zL$ have the same degrees as $X^a$, $\Xi^A$, the right-hand sides of \eqref{NatBeta} have the same degrees as $F^b$, $F^B$, hence, $y^b_\zL$, $\zh^B_\zL$ have the degrees required to be a $\zL$-point in $\cV^{r|\ul{s}}$. Moreover, we have \begin{gather*}y^b_{||}=F^b_{00}(x_{||}),\end{gather*} so that $y_{||}\in\cV^r$. The target of the map $\zb_\zL$ is thus actually $\cV^{r|\ul{s}}(\zL)$. The naturality of $\zb$ under morphisms of $\Z_2^n$-Grassmann algebras is obvious: $\zb$ is a natural transformation in $[\Z_2^n{\tt Pts}^{\op{op}},\tt Set]$ between $\cU^{p|\ul{q}}(-)$ and $\cV^{r|\ul{s}}(-)$. Finally, we defined a map
\begin{gather*}\mathcal{I}\colon \ \big(\mathcal{B}_{p|\ul{q}}\big(\mathcal{U}^p\big)\big)^{r|\ul{s}}\to\op{Nat}\big(\cU^{p,\ul{q}},\cV^{r,\ul{s}}\big).\end{gather*}

We will explain now that any natural transformation $\beta\colon \mathcal{U}^{p|\ul{q}}(-) \longrightarrow \mathcal{V}^{r|\ul{s}}(-)$ is the image by~$\mathcal{I}$ of a unique `vector' $\mathbf{F}$. We first show that, for any $\zL\simeq\R^{0|\ul{m}}$, the image $\zb_\zL(\mr{x}^*)\in\cV^{r|\ul{s}}(\zL)$ of any $\zL$-point
\begin{gather*}\mr{x^*}\simeq\big(x^a_{||},\rim{x}^a_\zL,\zx_\zL^A\big)\in\mathcal{U}^p \times \mathring{\Lambda}_{0}^p \times \rim{\Lambda}^{q_1}_{\gamma_1} \times \dots \times \rim{\Lambda}^{q_N}_{\gamma_N}
\end{gather*}
in $\cU^{p|\ul{q}}$, has components $y^b_\zL$ and $\zh^B_\zL$ of the type~\eqref{NatBeta}.

{\it Step 1}. {\it We prove that any $\zL$-point in $\cU^{p|\ul{q}}$ is the image by a $\Z_2^n$-Grassmann algebra map $\zf^*\colon \zL'\to\zL$ of a $\zL'$-point in $\cU^{p|\ul{q}}$, some of whose defining series are series in formal pairings.}

Let $\big(\theta^C\big)$ be the generators of $\Lambda$. The $\zL$-point $\mr{x}^*$ then reads
\begin{gather*}\mr{x}^*\simeq \big(x^a_{||}, \textstyle\sum_{\lambda \kappa} \zvy^\lambda \zvy^\kappa K_{\kappa \lambda}^a,\zx^A_\Lambda\big),\end{gather*}
where the degree of $K_{\kappa \lambda}^a \in \Lambda$ is the sum of the degrees of $\zvy^\zl$ and $\zvy^\zk$. Recall that $a$ (resp.,~$A$) runs through $\{1,\ldots,p\}$ (resp., through $\{1,\ldots,|\ul{q}|\}$), and that $\zl$, $\zk$ run through $\{1,\ldots,|\ul{m}|\}$. Consider now the set $S$ of generators
\begin{gather*}\zvy'=\big( \eta^{a\lambda} , \zeta_\kappa^b , \psi^A \big),\end{gather*}
where $b$ has the same range as $a$, and define their (non-zero) $\Z_2^n$-degrees by
\begin{gather*}
 \deg\big(\eta^{a\lambda}\big) = \deg\big(\zvy^\zl\big), \qquad \deg\big(\zeta^b_\kappa\big) = \deg\big(\zvy^\kappa\big), \qquad \deg\big(\psi^A\big) = \deg\big(\zx^A_\zL\big) = \deg \big(\zx^A\big).
\end{gather*}
Let $\zL'$ be the $\Z_2^n$-Grassmann algebra defined by $S$, and set
\begin{gather*}\label{Dependence1} \mathrm{x'^*} \simeq \big( x_{||}^a, \textstyle\sum_{\lambda}\eta^{a\lambda}\zeta_\lambda^a , \psi^A\big)\in\mathcal{U}^p \times \mathring{\Lambda}'^p_{0} \times \rim\Lambda'^{q_1}_{\gamma_1} \times \dots \times \rim\Lambda'^{q_N}_{\gamma_N} \end{gather*}
(no sum over $a$ in the formal pairings $\textstyle\sum_{\lambda}\zh^{a\zl}\zz^{a}_\zl$). The degree-respecting equalities
\begin{gather*}
\varphi^*\big(\eta^{a\lambda}\big) = \theta^\lambda, \qquad \varphi^*\big(\zeta^b_\kappa\big) = \sum_{\lambda} \zvy^\lambda K_{\lambda \kappa}^b, \qquad \varphi^*\big(\psi^A\big) = \zx^A_\Lambda
\end{gather*}
define a morphism of $\Z_2^n$-Grassmann algebras $\varphi^* \colon \Lambda' \longrightarrow \Lambda$. It suffices to set
\begin{gather*}\zf^*\bigg(\sum_\ze r_\zve\zvy'^\ze\bigg) := \sum_\ze r_\ze (\zf^*\zvy')^\ze.\end{gather*}
Indeed, any term of the right-hand side is a series in $\zvy$ whose terms contain at least $|\ze|$ gene\-ra\-tors. Hence, for any $\zve$, only the terms $|\ze|\le |\zve|$ can contribute to $\zvy^\zve$, and therefore there is no convergence issue with the coefficient of $\zvy^\zve$. Since the $\zL$-point $\zf^*\circ\mr{x}'^*$ in $\cU^{p|\ul{q}}$ reads
\begin{gather*}\zf^*\circ\mr{x}'^* \simeq \zf^*\big( x_{||}^a,\textstyle\sum_{\lambda} \eta^{a\lambda}\zeta_\lambda^a , \psi^A\big) = \big( x_{||}^a, \textstyle\sum_{\lambda \kappa}\zvy^\zl\zvy^\zk K^a_{\zk\zl}, \zx^A_\zL\big) \simeq \mr{x}^*,\end{gather*}
naturality of the transformation $\beta \colon \mathcal{U}^{p|\ul{q}}(-) \longrightarrow \mathcal{V}^{r|\ul{s}}(-)$ implies that
\begin{gather}\big(y^b_{\zL},\zh^B_\zL\big) :\simeq \zb_\zL(\mr{x}^*) = \zb_\zL(\zf^*\circ{\mr{x}'^*}) = \zb_\zL\big( \cU^{p|\ul{q}}(\zf^*) (\mr{x}'^*)\big)\nonumber\\
\hphantom{\big(y^b_{\zL},\zh^B_\zL\big)}{} =\cV^{r|\ul{s}}(\zf^*)(\zb_{\zL'}(\mr{x}'^*)) = \zf^*\circ (\zb_{\zL'}(\mr{x}'^*))\simeq \zf^*\big(y^b_{\zL'},\zh^B_{\zL'}\big),\label{NTMain}\end{gather}
where $y^b_{\zL'}$ and $\zh^B_{\zL'}$ are series in the generators of $\zL'$.

{\it Step 2}. {\it We define formal rotations under which the formal pairings are invariant. Moreover, we show that any formal series that is invariant under the formal rotations is a series in the formal pairings.}

The formal part of each degree $0$ component of $\mathrm{x}'^*$ can be viewed as a formal pairing $\eta^a \cdot \zeta^a = \sum_\zl\eta^{a\lambda}\zeta^a_\lambda$, which is stable under formal rotations $R^*$. More precisely, we set
\begin{gather*}
R^*\big(\eta^{a\lambda}\big) = \sum_{\kappa}\eta^{a\kappa}(O^a)_\kappa^{ \lambda}, R^*\big(\zeta_\kappa^b\big) =\sum_{\lambda} \big(O^{b \st}\big)_{\kappa}^{ \lambda} \zeta_\lambda^b, R^*(\psi^A) = \psi^A,
\end{gather*}
where $O^a$ and $O^{b \st}$ are any $(m_1+\dots+m_N)\times(m_1+\dots+m_N)$ block-diagonal matrices with entries in $\R$ that satisfy
\begin{gather}\label{Ortho} \sum_{\lambda}(O^a)_\rho^{ \lambda}\big(O^{a \st}\big)_{\lambda}^{ \omega} = \delta_\rho^{ \omega}.\end{gather}
Since, for any fixed $a$ (resp.,~$b$), the components $\zh^{a\zl}$ (resp.,~$\zz^b_\zk$) are ordered such that the~$m_1$ first components have degree $\zg_1$, the next $m_2$ degree $\zg_2$, and so on, these equalities are degree-preserving. Hence, they define a $\Z_2^n$-Grassmann algebra morphism $R^*\colon \zL'\to\zL'$ via
\begin{gather*} R^*\bigg(\sum_{\ze}r_\ze\zvy'^\ze\bigg) = R^*\bigg(\sum_{\za\zb\zg}r_{\za\zb\zg} \zh^\za \zz^\zb \psi^\zg\bigg) := \sum_{\za\zb\zg}r_{\za\zb\zg} (R^*\zh)^\za (R^*\zz)^\zb \psi^\zg.\end{gather*}
Since the images $R^*(\zh^{a\zl})$ (resp., $R^*(\zz^b_\zk)$) are linear in the $\zh^{a\zk}$ (resp.,~$\zz^b_\zl$) (of the same degree), the term indexed by $\za\zb\zg$ is a homogeneous polynomial of order $|\za|+|\zb|+|\zg|$ in the generators~$\zvy'$. Hence, for any~$\zve$, only the terms $|\za|+|\zb|+|\zg|=|\zve|$ can contribute to $\zvy'^\zve$, so that no convergence problems arise. In view of~\eqref{Ortho}, it is clear that, as mentioned above, the formal pairing $\eta^a \cdot \zeta^a = \sum_\zl\eta^{a\lambda}\zeta^a_\lambda$ is invariant under $R^*$. As any $\Z_2^n$-Grassmann algebra morphism, the formal rotation~$R^*$ induces maps $\cU^{p|\ul{q}}(R^*)$ and $\cV^{r|\ul{s}}(R^*)$, and due to naturality of $\zb$, we find
\begin{gather*}
\cV^{r|\ul{s}}(R^*)(\zb_{\zL'}\mr{x}'^*) = \zb_{\zL'}\big( \cU^{p|\ul{q}}(R^*)(\mr{x}'^*)\big) = \zb_{\zL'}( R^*\circ\mr{x}'^*) \\
\hphantom{\cV^{r|\ul{s}}(R^*)(\zb_{\zL'}\mr{x}'^*)}{} \simeq \zb_{\zL'}\bigg(R^*\bigg( x_{||}^a, \sum_\zl\eta^{a\lambda}\zeta_\lambda^a , \psi^A\bigg)\bigg) \simeq \zb_{\zL'}\mr{x}'^*,\end{gather*}
so that $\zb_{\zL'}\mr{x}'^*$ is invariant under rotations.

We are now prepared to continue the computation \eqref{NTMain}. Since \begin{gather} \label{Dependence2}\zb_{\zL'}(\mr{x}'^*) \simeq \big(y_{\zL'}^b,\zh_{\zL'}^B\big)=\big(y^b_{||},\rim{y}_{\zL'}^b,\zh_{\zL'}^B\big)\end{gather} is invariant under the rotations $R^*$, the series $\rim{y}_{\zL'}^b$, $\zh_{\zL'}^B$ in the generators $\zvy'$ are invariant. More explicitly, for each series, we have an equality of the type
\begin{gather*}\sum_\zg\bigg(\sum_{k,\ell}\sum_{|\za|=k,\, |\zb|=\ell}F_{\za\zb\zg} \zh^\za \zz^\zb\bigg) \psi^\zg = \sum_\zg\bigg(\sum_{k,\ell}\sum_{|\za|=k,\, |\zb|=\ell}F_{\za\zb\zg} (R^*\zh)^\za (R^*\zz)^\zb\bigg) \psi^\zg,\end{gather*}
which is equivalent to
\begin{gather*}\sum_{|\za|=k, \, |\zb|=\ell}F_{\za\zb\zg} \cdots \zh^{a\zl}\zh^{b\zm}\zz^c_\zn\cdots = \sum_{|\za|=k,\, |\zb|=\ell}F_{\za\zb\zg} \zh^\za \zz^\zb =\sum_{|\za|=k,\, |\zb|=\ell}F_{\za\zb\zg} (R^*\zh)^\za (R^*\zz)^\zb \\
\qquad {} = \sum_{|\za|=k,\, |\zb| =\ell}F_{\za\zb\zg} \cdots \eta^{a\zd}(O^a)_\zd^{ \lambda} \eta^{b\zd'}\big(O^{b}\big)_{\zd'}^{ \zm} \big(O^{c \st}\big)_{\zn}^{ \zd''} \zeta_{\zd''}^c \cdots,\end{gather*}
and holds for all (!) formal rotations. This is only possible, if the power series considered, i.e., the series $\rim{y}^b_{\zL'}$ and $\zh_{\zL'}^B$, are series in pairings $\zh^a\cdot\zz^a=\sum_\zl\zh^{a\zl}\zz^a_\zl$. In the classical setting, the result is known under the name of first fundamental theorem of invariant theory for the orthogonal group \cite{OInvariants1, OInvariants2}. It has been extended to the graded situation in \cite[Proposition~4.13]{Balduzzi:2010}. In view of~\eqref{NTMain}, we thus get
\begin{gather*}\big(y^b_{||},\rim{y}^b_\zL,\zh_\zL^B\big)=\zb_\zL(\mr{x}^*)=\zb_\zL\big(x^a_{||},\rim{x}^a_\zL,\zx^A_\zL\big) =\big(y^b_{||},\zf^*(\rim{y}^b_{\zL'}),\zf^*\big(\zh^B_{\zL'}\big)\big),\end{gather*}
where any image by $\zf^*$ is of the type
\begin{gather*}\sum_{(\za,\zb)\neq (0,0)} F_{\za\zb} \zf^* \big((\zh\cdot\zz)^\zb\big) \zf^* \big(\psi^\za\big) = \sum_{(\za,\zb)\neq (0,0)} F_{\za\zb} \rim{x}_\zL^\zb \zx_\zL^\za.\end{gather*}
It is clear from \eqref{Dependence2} and \eqref{Dependence1} that the coefficients \begin{gather*}F^b_{\za\zb}, F_{\za\zb}^B ((\za,\zb)\neq(0,0)), \qquad \text{and} \qquad F^b_{00}:=y^b_{||}\end{gather*} depend (only) on~$x_{||}\in\cU^p$. Hence, the image \begin{gather*}\big(y_\zL^b,\zh_\zL^B\big)=\zb_\zL(\mr{x}^*)=\big(F^b(x_{||},\rim{x}_\zL,\zx_\zL),F^B(x_{||},\rim{x}_\zL,\zx_\zL)\big)\end{gather*} is actually of the type \eqref{NatBeta}. Since $\zb_\zL(\mr{x}^*)$ is a $\zL$-point in $\cV^{r|\ul{s}}$, the $r$ series $F^b(x_{||},\rim{x}_\zL,\zx_\zL)$ and the~$s_i$ series $F^B(x_{||},\rim{x}_\zL,\zx_\zL)$ are of degree 0 and degree $\zg_i$, respectively, i.e., the $r$ series $F^b(x,X,\Xi)$ and the~$s_i$ series $F^B(x,X,\Xi)$ are of degree $0$ and degree $\zg_i$, respectively. For the same reason, we have $F_{00}(x_{||})\in\cV^r$, for all $x_{||}\in\cU^p$, so that we constructed a `vector' $\mathbf{F}\in(\mathcal{B}_{p|\ul{q}}(\cU^p))^{r|\ul{s}} $, whose image by $\mathcal{I}$ is obviously $\zb$.

{\it Step 3}. {\it We show that $\mathbf{F}$ is unique $($which concludes the proof$)$}. If there is another `vector'~$\mathbf{F}'$, such that $\mathcal{I}(\mathbf{F'})=\zb$, we have
\begin{gather}\label{UniquenessStep3}
\sum_{|\za|\ge 0,|\zb|\ge 0}F^{\mathfrak{b}}_{\za\zb}(x_{||}) \rim{x}_\zL^\zb \zx_{\zL}^\za=\sum_{|\za|\ge 0,|\zb|\ge 0} F'^{\mathfrak{b}}_{\za\zb}(x_{||}) \rim{x}_\zL^\zb \zx_{\zL}^\za,
\end{gather}
for all $\mathfrak{b}\in\{b,B\}$, all $\zL$, and all $\mr{x}^*$. Notice first that any $\rim{x}_\zL^a$ (resp., any $\zx_\zL^A$) is a series of degree~0 (resp., of degree $\deg\big(\zx^A\big)=\zg_A$) in the $\zvy$-s that contains at least two parameters $\zvy^{C}\zvy^{C'}$ (resp., at least one parameter $\zvy^{C''}$). Hence, both sides are series in $\zvy$, and the left-hand side and right-hand side coefficients of any monomial $\zvy^\ze$ coincide. A term $(\za,\zb)\neq(0,0)$ cannot contribute to the independent term $\zvy^0$. Hence $F_{00}^{\mathfrak{b}}(x_{||})=F_{00}'^{\mathfrak{b}}(x_{||})$. We now show that $F_{\za\zb}^{\mathfrak{b}}(x_{||})=F_{\za\zb}'^{\mathfrak{b}}(x_{||}),$ for an arbitrarily fixed $(\za,\zb)\neq(0,0)$. Since $\zL$ is arbitrary, we can choose as many different generators~$\zvy$ in each non-zero degree as necessary, and, since $\mr{x}^*$ is arbitrary, we can choose $x_{||}$ arbitrarily in $\cU^p$ and we can choose the coefficients of the series~$\rim{x}_\zL^a$ and $\zx_\zL^A$ arbitrarily (except that we have to observe that the coefficient of a monomial~$\zvy^\ze$, which does not have the required degree, {\it must be zero}). Let now $\za_1,\ldots,\za_{\zm}$ and $\zb_1, \ldots, \zb_{\zn}$ be the non-zero components in the fixed $\za$ and $\zb$. For each factor $\zx_\zL^{A_i}$ of
\begin{gather*}\zx_\zL^\za=\big(\zx_\zL^{A_1}\big)^{\za_1}\cdots\big(\zx_\zL^{A_\zm}\big)^{\za_{\zm}},\end{gather*}
we choose a monomial in one generator $\zvy^{C_i}$ of degree $\zg_{A_i}$, set its coefficient $r_{C_i}$ to~1, and all the other coefficients in the series $\zx_\zL^{A_i}$ to zero. Further, for different $\zx_\zL^{A_i}$, we choose different generators~$\zvy^{C_i}$. Similarly, for each factor $\rim{x}_\zL^{a_j}$ of
\begin{gather*}\rim{x}_\zL^\zb=\big(\rim{x}_\zL^{a_1}\big)^{\zb_1}\cdots\big(\rim{x}_\zL^{a_{\zn}}\big)^{\zb_{\zn}},\end{gather*}
we choose monomials $\zvy^{D_{jk}}\zvy^{E_{jk}}$ ($k\in\{1,\ldots,\zb_j\}$) in two generators of the same odd degree (for all $\Z_2^n$-manifolds with $n\ge 1$, there is at least one odd degree), set their coefficient $r_{D_{jk}E_{jk}}$ to 1, and all the other coefficients in the series $\rim{x}_\zL^{a_j}$ to zero. Further, we choose the generators so that all generators $\zvy^{C_i},$ $\zvy^{D_{jk}}$, and $\zvy^{E_{jk}}$ are different. When setting \begin{gather*}\zvy^{\zw}=\prod_{j=1}^\zn\zvy^{D_{j1}}\zvy^{E_{j1}}\cdots\zvy^{D_{j\zb_j}}\zvy^{E_{j\zb_j}}\prod_{i=1}^\zm\big(\zvy^{C_i}\big)^{\za_i}\neq 0,\end{gather*} the terms indexed by (the fixed) $(\za,\zb)$ in both sides of \eqref{UniquenessStep3}, read \begin{gather*}\zb!F^{\mathfrak{b}}_{\za\zb}(x_{||})\zvy^\zw \qquad \text{and} \qquad \zb!F'^{\mathfrak{b}}_{\za\zb}(x_{||})\zvy^\zw.\end{gather*} For any term $(\za',\zb')\neq(\za,\zb)$, we either get a new series $\zx_\zL^A$ or $\rim{x}_\zL^a$ (i.e., a series that is not present in $\zx_\zL^\za$ or $\rim{x}_\zL^\zb$), or we get an old series a different number of times. In the second case, the term $(\za',\zb')$ does not contribute to the coefficient of $\zvy^\zw$; in the first, we set all the coefficients of the new series to 0, so that the term $(\za',\zb')$ vanishes. Finally, we obtain $F^{\mathfrak{b}}_{\za\zb}(x_{||})=F'^{\mathfrak{b}}_{\za\zb}(x_{||})$, for any $x_{||}\in\cU^p$.
\end{proof}

We now show that $\R^{p|\ul{q}}(\Lambda)$ is a Fr\'echet space and that $\mathcal{U}^{p|\ul{q}}(\Lambda)$ is an open subset of $\R^{p|\ul{q}}(\Lambda)$. This means that we have a notion of directional derivative, as well as a notion of smoothness of continuous maps between the $\Lambda$-points of $\Z_2^n$-domains. For more details on Fr\'{e}chet objects, we refer the reader to Appendix~\ref{appx:AManifolds}.

\begin{Proposition}\label{LambdaFrechet}For any $\Lambda \in \Z_2^n \catname{GrAlg} $, the set $\R^{p|\ul{q}}(\Lambda)$ is a~nuclear Fr\'{e}chet space and a~Fr\'{e}chet $\Lambda_{0}$-module. Moreover, the set $\cU^{p|\ul{q}}(\zL)$ is an open subset of $\R^{p|\ul{q}}(\zL)$.
\end{Proposition}

\begin{proof}Let $\Lambda \in \Z_2^n\catname{GrAlg} $. As explained above, there is a $1:1$ correspondence between the $\zL$-points $\mr{x}^*$ of $\R^{p|\ul{q}}$ (resp., of $\cU^{p|\ul{q}}$) and the $(p+|\ul{q}|)$-tuples
\begin{gather*}\mathrm{x}^* \simeq \big( {x}^a_\Lambda , \zx^A_\Lambda\big)\in {\Lambda}_{0}^p \times \Lambda_{\gamma_1}^{q_1} \times \dots \times \Lambda_{\gamma_N}^{q_N} \end{gather*}
(resp., the same $(p+|\ul{q}|)$-tuples, but with the additional requirement that the $p$-tuple $(x^a_{||})$ made of the independent terms of $(x^a_\zL)$ be a point in $\cU^p\subset\R^p$). Note now that $\zL$ is the $\Z_2^n$-graded $\Zn$-commutative nuclear Fr\'echet $\R$-algebra of global $\Z_2^n$-functions of some $\R^{0|\ul{m}}$. Hence, all its homogeneous subspaces $\zL_{\zg_i}$ ($i\in{0,\ldots, N}$, $\zg_0=0$) are nuclear Fr\'echet vector spaces. Since any product (resp., any countable product) of nuclear (resp., Fr\'echet) spaces is nuclear (resp., Fr\'echet), the set $\R^{p|\ul{q}}(\zL)$ of $\zL$-points of $\R^{p|\ul{q}}$ is nuclear Fr\'echet. The latter statements can be found in~\cite{Bruce:2018}.

As for the second claim in Proposition~\ref{LambdaFrechet}, recall that $\Lambda_{0}$ is a (commutative) Fr\'{e}chet algebra, see Corollary~\ref{Coro:Lambda0Frec}. The Fr\'{e}chet $\Lambda_{0}$-module structure on $\R^{p|\ul{q}}(\zL)$ is then defined by
\begin{gather}\label{ActionCompWise}
 m\colon \ \Lambda_{0} \times \R^{p|\ul{q}}(\Lambda)\ni (\mathrm{a} , \mathrm{x}^*)\mapsto \big(\mathrm{a} \cdot x_\Lambda^a , \mathrm{a} \cdot \zx_\Lambda^{A} \big) \in\R^{p|\ul{q}}(\Lambda) .
\end{gather}
Since this action is defined using the continuous associative multiplication $\cdot\colon \zL_{\zg_i}\times\zL_{\zg_j}\to\zL_{\zg_i+\zg_j}$ of the Fr\'echet algebra $\zL$, it is (jointly) continuous.

As any closed subspace of a Fr\'echet space is itself a Fr\'echet space, the space \begin{gather*}\rim{\zL}_0\simeq \{0\}\times \rim{\zL}_0\subset \R\times\rim{\zL}_0=\zL_0\end{gather*} is Fr\'echet. We thus see that \begin{gather}\label{IsomPtProd}\cU^{p|\ul{q}}(\zL)\simeq \cU^p\times {\rim\Lambda}_{0}^p \times \prod_{i=1}^N\Lambda_{\gamma_i}^{q_i}\subset\R^p\times {\rim\Lambda}_{0}^p \times \prod_{i=1}^N\Lambda_{\gamma_i}^{q_i}\simeq\R^{p|\ul q}(\zL) \end{gather} is open.
\end{proof}

\begin{Remark} In the following, we will use the isomorphisms \eqref{IsomPtProd} (and similar ones) without further reference.\end{Remark}

The just described $\Lambda_0$-module structure is vital in understanding the structure of the $\Lambda$-points of any $\Z_2^n$-manifold. In particular, morphisms between $\Z_2^n$-domains induce natural transformations between the associated functors that respect this module structure. The converse is also true, that is, any natural transformation between the associated functors that respects the $\Lambda_0$-module structure comes from a morphism between the underlying $\Z_2^n$-domains. More carefully, we have the following proposition.

\begin{Theorem}\label{prop:AMorphisms}Let $\mathcal{U}^{p|\ul{q}}$ and $\mathcal{V}^{r|\ul{s}}$ be $\Z_2^n$-domains. A natural transformation $\beta\colon \mathcal{U}^{p|\ul{q}}(-) \longrightarrow \mathcal{V}^{r|\ul{s}}(-)$ comes from a $\Z_2^n$-manifold morphism $\cU^{p|\ul q}\to \cV^{r|\ul s}$ if and only if $\beta_{\Lambda} \colon \mathcal{U}^{p|\ul{q}}(\Lambda) \longrightarrow \mathcal{V}^{r|\ul{s}}(\Lambda)$ is $\Lambda_{0}$-smooth, for all $\Lambda \in \Z_2^n\catname{GrAlg}$. That is, for all $\Lambda$, the map $\beta_{\Lambda}$ must be a smooth map $($from the open subset $\cU^{p|\ul q}(\zL)$ of the Fr\'echet space $\R^{p|\ul q}(\zL)$ to the Fr\'echet space $\R^{r|\ul s}(\zL)$, see Appendix~{\rm \ref{appx:AManifolds})} and its G\^ateaux derivative $($see Appendix~{\rm \ref{appx:AManifolds})} must be $\Lambda_{0}$-linear, i.e.,
\begin{gather*}\rmd_{\mathrm{x}^*}\beta_\Lambda(\mathrm{a} \cdot \mathrm{v}) = \mathrm{a} \cdot \rmd_{\mathrm{x}^*}\beta_\Lambda( \mathrm{v}),\end{gather*}
for all $\mathrm{x}^* \in \mathcal{U}^{p|\ul{q}}(\Lambda)$, $\mathrm{a} \in \Lambda_{0}$, and $\mathrm{v} \in \R^{p|\ul{q}}(\Lambda)$.
\end{Theorem}

\begin{proof}{\it Part I}. Let $\zb\colon \mathcal{U}^{p|\ul{q}}(-) \longrightarrow \mathcal{V}^{r|\ul{s}}(-)$ be a natural transformation with $\zL_0$-smooth components $\zb_\zL$, $\zL\in\Z_2^n\tt GrAlg$. From Theorem~\ref{prop:NaturalTrans}, we know that $\zb_\zL$ is completely specified by the systems
\begin{gather}\label{Theybs}
y^b_\Lambda =\sum_{|\za|\ge 0, |\zb|\ge 0}F^b_{\za\zb}(x_{||}) \rim{x}_\zL^\zb\zx_\zL^\za \qquad\text{and} \qquad \zh_\zL^B=\sum_{|\za|> 0, |\zb|\ge 0}F^B_{\za\zb}(x_{||}) \rim{x}_\zL^\zb\zx_\zL^\za,
\end{gather}
where the coefficients $F_{\za\zb}^{\mathfrak{b}}$ $(\mathfrak{b}\in\{b,B\})$ are set-theoretical maps from $\cU^p$ to $\R $.

{\it Part Ia}. Smoothness of $\zb_\zL$ implies that these coefficients are smooth. Indeed, we will show that $F_{\za\zb}^{\mathfrak{b}}\in C^0(\cU^p)$ and that, if $F_{\za\zb}^{\mathfrak{b}}\in C^k(\cU^p)$ $(k\ge 0)$, then $F_{\za\zb}^{\mathfrak{b}}\in C^{k+1}(\cU^p) $.

{\it Step 1}. Since \begin{gather*}\zb_\zL\colon \ \cU^{p|\ul q}(\zL)\to \Lambda_{0}^r \times \prod_{i=1}^N\Lambda_{\gamma_i}^{s_i}\end{gather*} is continuous, any of its components \begin{gather*}y_\zL^{\mathfrak{b}}\colon \ \cU^{p|\ul q}(\zL)\to \zL_{\zg_{i(\mathfrak b)}}=\R[[\zvy]]_{\zg_{i(\mathfrak b)}}\simeq\prod_{\zg_i(\mathfrak{b})}\R \end{gather*} is continuous. For simplicity, we wrote $y_\zL^B$ instead of $\zh_\zL^B$, and we will continue doing so. Moreover, the target space are the formal power series in $\zvy$ with coefficients in $\R$, all whose terms have the degree $\zg_{i(\mathfrak{b})}$ of $y^{\mathfrak b}$, and this space is identified with the corresponding space of families of reals. For any $\zw$ such that $\zvy^\zw$ has the degree $\zg_{i(\mathfrak b)}$, the corresponding real coefficient gives rise to a~continuous map \begin{gather*}y_\zL^{\mathfrak{b},\zw}\colon \ \cU^{p|\ul q}(\zL)\to \R.\end{gather*}
Since this joint continuity implies separate continuity with respect to $x_{||}\in\cU^p$, for any fixed $(\rim{x}_\zL,\zx_\zL)$ and any $\zL$, we can proceed as at the end of the proof of Theorem~\ref{prop:NaturalTrans}. More precisely, select any $(\za,\zb)$ and select (for an appropriate $\zL$) the pair $(\rim{x}_\zL,\zx_\zL)$ such that $\rim{x}_\zL^\zb\zx_\zL^\za=\zb! \zvy^\zw$, where~$\zvy^\zw$ is now the degree $\zg_{i(\mathfrak{b})}$ monomial defined in the proof just mentioned. The real coefficient of this monomial is $\zb! F_{\za\zb}^{\mathfrak{b}}(x_{||})$, which, as said, is an $\R$-valued continuous map on $\cU^p$, so that $F_{\za\zb}^{\mathfrak{b}}\in C^0(\cU^p)$, for all~$\mathfrak b$ and all $(\za,\zb)$.

{\it Step 2}. Since \begin{gather*}\cU^{p|\ul q}(\zL)\subset \R\times\left(\R^{p-1}\times {\rim\Lambda}_{0}^p \times \prod_{i=1}^N\Lambda_{\gamma_i}^{q_i}\right) \end{gather*} is an open subset of a product of two Fr\'echet spaces, smoothness of $\zb_\zL$ implies (via an iterated application of Proposition~\ref{partialFrechet}) that, for any $\mathfrak{b}\in\{b,B\}$, any $\ell\in\N$ and any $\zg\in\N^p$ $(|\zg|=\ell)$, the partial derivative \begin{gather*}\op{d}^\zg_{x_{||}}y^{\mathfrak{b}}_\zL\colon \ \cU^{p|\ul q}(\zL)\times\R^{\times \ell}\to \prod_{\zg_i(\mathfrak{b})}\R\end{gather*} is continuous.

Assume now that $F_{\za\zb}^{\mathfrak{b}}\in C^{k}(\cU^p)$ $(k\ge 0)$, for any $\mathfrak{b}$ and any $(\za,\zb)$, as well as that, for any $\zg\in\N^p$ $(|\zg|=k)$ and any $\mathfrak{b}$, the continuous partial G\^ateaux derivative \begin{gather*}\op{d}_{x_{||}}^\zg y^{\mathfrak{b}}_\zL(1,\ldots,1)\colon \ \cU^{p|\ul q}(\zL)\to \prod_{\zg_i(\mathfrak{b})}\R\end{gather*} is given by
\begin{gather}\label{IndHyp}\op{d}_{x_{||},\mr{x}^*}^\zg y^{\mathfrak{b}}_\zL(1,\ldots,1)=\sum_{\za\zb} \big(\partial_{x}^\zg F_{\za\zb}^{\mathfrak{b}}\big)(x_{||}) \rim{x}_\zL^\zb\zx_\zL^\za.\end{gather} Observe that for $k=0$, this condition is automatically satisfied. We will now show that, under these assumptions, the same statements hold at order $k+1$. In view of \eqref{IndHyp}, any order $k+1$ continuous partial G\^ateaux derivative \begin{gather*}\op{d}_{x_{||}^a}\op{d}_{x_{||}}^\zg y_\zL^{\mathfrak{b}}(1,\ldots,1)\colon \ \cU^{p|\ul q}(\zL)\to \prod_{\zg_{i}(\mathfrak{b})}\R \end{gather*} $(a\in\{1,\ldots,p\}, |\zg|=k)$ is given, at any $\mr{x}^*\simeq(x_{||},\rim{x}_\zL,\zx_\zL)\in\cU^{p|\ul q}(\zL)$, by
\begin{gather} \label{IndConcl}\sum_{\za\zb}\lim_{t\to 0}\frac{1}{t}\big(\big(\partial_x^\zg F^{\mathfrak{b}}_{\za\zb}\big)(x^1_{||},\ldots, x^a_{||}+t,\ldots, x^p_{||})-\big(\partial_x^\zg F^{\mathfrak{b}}_{\za\zb}\big)\big(x^1_{||},\ldots, x^a_{||},\ldots, x^p_{||}\big)\big)\rim{x}_\zL^\zb\zx_\zL^\za.\end{gather} When proceeding as in Step 1, we get that the limit is an $\R$-valued continuous function in $\cU^p$. In other words, the partial derivative $\partial_{x^a}\partial_x^\zg F_{\za\zb}^{\mathfrak{b}}$ exists and is continuous in $\cU^p$, i.e., $F_{\za\zb}^{\mathfrak{b}}\in C^{k+1}(\cU^p)$. Moreover, formula~\eqref{IndHyp} pertaining to order $k$ derivatives, extends to the order $k+1$ derivatives, see~\eqref{IndConcl}.

{\it Part Ib}. We examine the further consequences of $\Lambda_{0}$-smoothness, in particular those of $\Lambda_{0}$-linearity. Since $\zb_\zL$ is of class $C^1$, its components $y^{\mathfrak{b}}_\zL\colon \cU^{p|\ul q}(\zL)\to\prod_{\zg_{i}(\mathfrak{b})}\R$ are of class $C^1$. Further, as \begin{gather*}\cU^{p|\ul q}(\zL)\subset \big(\R\times\rim{\zL}_0 \big)\times \left(\R^{p-1}\times \rim{\zL}_0^{p-1}\times\prod_{i=1}^N\zL_{\zg_i}^{q_i}\right)\end{gather*} is an open subset of a product of two Fr\'echet spaces, the partial G\^ateaux derivative \begin{gather*}\op{d}_{(x^a_{||},\rim{x}_\zL^a)}y_\zL^{\mathfrak{b}}\colon \ \cU^{p|\ul q}(\zL)\times\big(\R\times\rim\zL_0\big)\to \prod_{\zg_{i}{(\mathfrak{b})}}\R\end{gather*} is continuous. It is given by
\begin{gather*}\op{d}_{(x^a_{||}, \rim{x}_\zL^a), \mr{x}^*}y_\zL^{\mathfrak{b}}(v_{||},\rim v_\zL)=\op{d}_{x^a_{||}, \mr{x}^*}y_\zL^{\mathfrak{b}}(v_{||})+\op{d}_{\rim{x}_\zL^a, \mr{x}^*}y_\zL^{\mathfrak{b}}(\rim v_\zL)\\
\qquad {} =v_{||}\sum_{\za\zb}\big(\partial_{x^a}F^{\mathfrak{b}}_{\za\zb}\big)(x_{||})\rim{x}_\zL^\zb\zx_\zL^\za+\sum_{\za\zb}F^{\mathfrak{b}}_{\za\zb}(x_{||}) \lim_{t\to 0}\frac{1}{t}\big((\rim{x}^a_\zL+t \rim v_\zL)^{\zb_a}-(\rim{x}^a_\zL)^{\zb_a}\big)\prod_{b\neq a} \big(\rim x^b_\zL\big)^{\zb_b}\zx_\zL^\za \\
\qquad{} = : v_{||}T_1 + \mathfrak{T}_2.\end{gather*}
As $\rim\zL_0$ is a commutative algebra, it follows from the binomial formula that \begin{gather*}\mathfrak{T}_2=\rim v_\zL\sum_{\za\zb}\zb_a F^{\mathfrak{b}}_{\za\zb}(x_{||})\rim x_\zL^{\zb-e_a}\zx_\zL^\za=:\rim v_\zL T_2,\end{gather*} where $(e_a)_a$ is the canonical basis of $\R^p$. Observe now that, in view of \eqref{ActionCompWise}, the $\zL_0$-linearity of the total G\^ateaux derivative of $y_\zL^{\mathfrak{b}}$ with respect to $\mr{x}^*$ is equivalent to the $\zL_0$-linearity of all its partial G\^ateaux derivatives with respect to the $x_\zL^a=\big(x_{||}^a,\rim x_\zL^a\big)$ and the $\zx_\zL^A$. For $a=0+\rim v_\zL\in \zL_0$ and $\mr{v}=1+0\in\R+\rim\zL_0=\zL_0$, this implies that
\begin{gather*}\rim v_\zL T_2=\op{d}_{(x^a_{||}, \rim{x}_\zL^a), \mr{x}^*}y_\zL^{\mathfrak{b}}(\rim v_\zL\cdot 1)=\rim v_\zL\cdot \op{d}_{(x^a_{||}, \rim{x}_\zL^a), \mr{x}^*}y_\zL^{\mathfrak{b}}(1)=\rim v_\zL T_1,\end{gather*} i.e., that
 \begin{gather*}
 \rim v_\zL\sum_{\za\zb}(\zb_a+1) F^{\mathfrak{b}}_{\za,\zb+e_a}(x_{||})\rim x_\zL^{\zb}\zx_\zL^\za\\
 \qquad {} =\rim v_\zL\sum_{\za,\zg\colon \zg_a\neq 0}\zg_a F^{\mathfrak{b}}_{\za\zg}(x_{||})\rim x_\zL^{\zg-e_a}\zx_\zL^\za=\rim v_\zL\sum_{\za\zb}\big(\partial_{x^a}F^{\mathfrak{b}}_{\za\zb}\big)(x_{||})\rim{x}_\zL^\zb\zx_\zL^\za.\end{gather*}
 Since $\zL\in\Z_2^n\tt GrAlg$, $\rim v_\zL\in\rim\zL_0$, and $\mr{x}^*\in\cU^{p|\ul q}(\zL)$ are arbitrary, we can repeat the $\zvy^\zw$-argument used above. More precisely, we select $(\za,\zb)$, select $(\rim x_\zL,\zx_\zL)$ such that $\rim x_\zL^\zb\zx_\zL^\za=\zb! \zvy^\zw$, and select $\rim v_\zL=\zvy^D\zvy^E\in\rim\zL_0$ such that $\zvy^D\zvy^E\zvy^\zw\neq 0.$ The coefficients of the latter monomial in the left and right hand sides do coincide, which means that
 \begin{gather*} (\zb_a+1) F^{\mathfrak{b}}_{\za,\zb+e_a}(x_{||})=\big(\partial_{x^a}F^{\mathfrak{b}}_{\za\zb}\big)(x_{||}),\end{gather*}
 or, equivalently,
 \begin{gather} F^{\mathfrak{b}}_{\za\zg}(x_{||})=\frac{1}{\zg_a}\big(\partial_{x^a}F^{\mathfrak{b}}_{\za,\zg-e_a}\big)(x_{||}),\label{Propagation}\end{gather}
 for all $\mathfrak{b}$, $\za$, $a$, all $\zg\colon \zg_a\neq 0$, and all $x_{||}\in\cU^p$. For any $\mathfrak{b}$, $\za$, and $x_{||}$, we now set \begin{gather*}\zvf^{\mathfrak{b}}_\za(x_{||}):=F^{\mathfrak{b}}_{\za 0}(x_{||})\in C^\infty\big(\cU^p\big).\end{gather*} An iterated application of~\eqref{Propagation} shows that \begin{gather*}F_{\za\zg}^{\mathfrak{b}}(x_{||})=\frac{1}{\zg!}\big(\partial^\zg_{x}\zvf^{\mathfrak{b}}_\za\big)(x_{||}).\end{gather*} Hence, the $y_\zL^{\mathfrak{b}}$ have the form~\eqref{eqn:smoothzero} and~\eqref{eqn:smoothnonzero}. This means that the natural transformation~$\zb$ is implemented by the $\zvf^{\mathfrak{b}}_\za$, which define actually a $\Z_2^n$-morphism from $\cU^{p|\ul q}$ to $\cV^{r|\ul s}$. Indeed, the property $\big(\zvf^b_0\big)(\cU^p)\subset\cV^r$ follows from the similar property of $\big(F^b_{00}\big)$. On the other hand, the pullback \begin{gather*}\zvf^*\big(y^{\mathfrak{b}}\big):=\sum_\za\zvf_\za^{\mathfrak{b}}(x)\zx^\za\end{gather*} must have the same degree as $y^{\mathfrak{b}}$. However, if $\deg(\zx^\za)\neq\deg\big(y^{\mathfrak{b}}\big)$, then $\deg(\zx^\za_\zL)\neq\deg\big(y^{\mathfrak{b}}_\zL\big)$, whatever $\zx_\zL$. It follows therefore from~\eqref{Theybs} that $\zvf^{\mathfrak{b}}_\za=F^{\mathfrak{b}}_{\za 0}=0$.

{\it Part II}. The proof of the converse implication is less demanding. Let $\zb\colon \cU^{p|\ul q}(-)\to \cV^{r|\ul s}(-)$ be a natural transformation that is induced by a $\Z_2^n$-morphism $\zvf\colon \cU^{p|\ul q}\to\cV^{r|\ul s} $, i.e., that is of the form~\eqref{eqn:smoothzero} and~\eqref{eqn:smoothnonzero}. For any $\zL\in\Z_2^n\tt GrAlg$, the map $\zb_\zL$ is smooth and its derivative is $\zL_0$-linear if and only if its components $y_\zL^{\mathfrak{b}}$ have these properties. The total derivative of $y_\zL^{\mathfrak{b}}$ with respect to $\mr{x}^*$ exists, is continuous, and is $\zL_0$-linear if and only if its partial derivatives with respect to the $x_\zL^a$ and the $\zx_\zL^A$ exist, are continuous, and are $\zL_0$-linear. When computing the derivative $y_\zL^{\mathfrak{b}}$ with respect to $\zx_\zL^{A_i}\in\zL_{\zg_i}$ at $\mr{x}^*\in\cU^{p|\ul q}(\zL)$ in the direction of $w_\zL\in\zL_{\zg_i}$, we get \begin{gather*}\sum_{\za\zb}\frac{1}{\zb!}\big(\partial_x^\zb\zvf_\za^{\mathfrak{b}}\big)(x_{||}) \rim x_\zL^\zb \big(\zx_\zL^{A_1}\big)^{\za_1}\cdots \lim_{t\to 0}\frac{1}{t}\big(\big(\zx_\zL^{A_i}+tw_\zL\big)^{\za_i}-\big(\zx_\zL^{A_i}\big)^{\za_i}\big) \cdots\big(\zx_\zL^{A_{|\ul q|}}\big)^{\za_{|\ul q|}}.\end{gather*} If $\zg_i$ is odd, the exponent $\za_i$ is $0$ or $1$. In the first (resp., the second) case, the limit vanishes (resp., is $w_\zL$). If $\zg_i$ is even, the multiplication of vectors in $\zL_{\zg_i}$ is commutative and the binomial formula shows that the limit is $w_\zL \za_i\big(\zx_\zL^{A_i}\big)^{\za_i-1} $. The derivative thus exists, is continuous, and is $\zL_0$-linear. Similarly, the derivative of $y_\zL^{\mathfrak{b}}$ with respect to $x_\zL^a$ exists if and only if its derivatives with respect to~$x_{||}^a$ and with respect to $\rim x_\zL^a$ exist. The (standard) computation of the derivative with respect to $x_\zL^a$ at $\mr{x}^*$ in the direction of \begin{gather*}v_\zL=(v_{||},\rim v_\zL)\in\R\times \rim\zL_0\end{gather*} thus leads to the sum of the terms \begin{gather*}v_{||}\sum_{\za\zb}\frac{1}{\zb!}\big(\partial_x^{\zb+e_a}\zvf_\za^{\mathfrak{b}}\big)(x_{||}) \rim x_\zL^\zb \zx_\zL^\za \end{gather*} and \begin{gather*}\rim v_\zL\sum_{\za,\zg\colon \zg_a\neq 0}\frac{1}{\zg!}\big(\partial_x^\zg\zvf_\za^{\mathfrak{b}}\big)(x_{||}) \zg_a \rim x_\zL^{\zg-e_a} \zx_\zL^\za=\rim v_\zL\sum_{\za\zb}\frac{1}{\zb!}\big(\partial_x^{\zb+e_a}\zvf_\za^{\mathfrak{b}}\big)(x_{||}) \rim x_\zL^{\zb} \zx_\zL^\za.\end{gather*} The derivative considered does therefore exist, is continuous, and is $\zL_0$-linear (note that it is essential that the derivative is the series over $\za\zb$ multiplied by~$v_\zL$~-- as $a\in\zL_0$ does not act on~$v_{||}$).\end{proof}

\begin{Remark}The $\Lambda_0$-linearity is a strong constraint that takes us from the category of ge\-ne\-ralized $\Z_2^n$-manifolds to the one of $\Z_2^n$-manifolds. A similar phenomenon exists in complex analysis. Indeed, for any real differentiable function $f = u + \rmi v \colon \Omega\subset\C\simeq \R^2 \rightarrow \C\simeq \R^2$, the Jacobian is an $\R$-linear map $J_f \colon \R^2 \rightarrow \R^2$. However, if we further insist that the Jacobian be $\C$-linear, then we see that $f$ must be holomorphic, that is, it must satisfy the Cauchy--Riemann equations on $\Omega$. Imposing $\C$-linearity thus greatly restricts class of functions and takes us from real analysis to complex analysis.
\end{Remark}

It will also be important to understand what happens to the $\Lambda$-points of a given $\Z_2^n$-domain under morphisms of $\Z_2^n$-Grassmann algebras.
\begin{Proposition}\label{prop:varphiSmooth} Let $\mathcal{U}^{p|\ul{q}}$ be a $\Z_2^n$-domain and let $\psi^* \colon \Lambda \rightarrow \Lambda'$ be a morphism of $\Z_2^n$-Grassmann algebras. The induced map $($see \eqref{AlgMorLbdPts}$)$
\begin{gather*}
\Psi := \mathcal{U}^{p|\ul{q}}(\psi^*)\colon \ \mathcal{U}^{p|\ul{q}}(\Lambda)\ni \mr{x}^*\simeq(x_\zL,\zx_\zL) \mapsto \psi^*\circ\mr{x}^*\simeq\psi^*(x_\zL,\zx_\zL)\in \mathcal{U}^{p|\ul{q}}(\Lambda')
\end{gather*}
is a smooth map from the open subset $\cU^{p|\ul q}(\zL)$ of the Fr\'echet space and Fr\'echet $\zL_0$-module $\R^{p|\ul q}(\zL)$ to the open subset $\cU^{p|\ul q}(\zL')$ of the Fr\'echet space and Fr\'echet $\zL'_0$-module $\R^{p|\ul q}(\zL')$, such that
\begin{gather*}\rmd_{\mathrm{x}^*} \Psi(\mathrm{a} \cdot \mathrm{v}) = \psi^*(\mathrm{a}) \cdot \rmd_{\mathrm{x}^*} \Psi( \mathrm{v}),\end{gather*}
for all $\mathrm{x}^*\in \mathcal{U}^{p|\ul{q}}(\Lambda)$, $\mathrm{v} \in \R^{p|\ul{q}}(\Lambda)$ and $\mathrm{a} \in \Lambda_0 $.
\end{Proposition}

\begin{proof}Since $\zL=\cO_{\R^{0|\ul m}}(\{\star\})$, so that \begin{gather*}\psi^*\in\Hom_{\Z_2^n{\tt Alg}}(\cO_{\R^{0|\ul m}}(\{\star\}),\cO_{\R^{0|\ul{m}'}}(\{\star\})),\end{gather*} there is a unique morphism \begin{gather*}\Phi=(|\phi|,\phi^*)\in\Hom_{\Z_2^n{\tt Man}}\big(\R^{0|\ul m'},\R^{0|\ul m}\big),\end{gather*} such that $\psi^*=\phi^*_{\{\star\}}$. Hence, the morphism $\psi^*$ is continuous from $\zL=\R[[\zvy]]$ to $\zL'=\R[[\zvy']]$ endowed with their standard locally convex topologies~\cite{Bruce:2018}, and so are its restrictions $\psi^*|_{\zL_{\zg_i}}$ from $\zL_{\zg_i}$ to $\zL'_{\zg_i}$. We thus see that the induced map {\samepage \begin{gather*}\Psi=(\psi^*|_{\zL_0})^{\times p}\times\prod_{i=1}^N(\psi^*|_{\zL_{\zg_i}})^{\times q_i} \end{gather*} is continuous.}

At $\mr{x}^*\simeq (x_\zL,\zx_\zL)=:u_\zL\in\cU^{p|\ul q}(\zL)$ and $\mr{v}\simeq v_\zL\in \R^{p|\ul q}(\zL)$, the derivative
\begin{gather*}\rmd \Psi \colon \ \mathcal{U}^{p|\ul{q}}(\Lambda) \times \R^{p|\ul{q}}( \Lambda) \longrightarrow \R^{p|\ul{q}}( \Lambda')\end{gather*}
is defined as
\begin{gather*}
\rmd_{\mathrm{x}^*} \Psi(\mathrm{v}) = \lim_{t \rightarrow 0}\frac{\Psi(\mathrm{x}^* + t \mathrm{v}) {-} \Psi(\mathrm{x}^*)}{t}
 = \lim_{t \rightarrow 0} \left(\dots,\frac{\psi^*(u^{\mathfrak{a}}_\zL + t v^{\mathfrak{a}}_\zL) {-} \psi^*(u^{\mathfrak{a}}_\zL) }{t},\dots\right)\\
\hphantom{\rmd_{\mathrm{x}^*} \Psi(\mathrm{v})}{} = \big(\dots,\psi^*(v_\zL^{\mathfrak{a}}),\dots\big) =: (\psi^*(v_\zL^{\mathfrak{a}})),
\end{gather*}
where $\mathfrak{a}$ is the label $a\in\{1,\ldots,p\}$ or $A\in\{1,\ldots,|\ul q|\}$ of any coordinate in $\R^{p|\ul q}(\zL)$, and where we used the $\R$-linearity of the $\Z_2^n$-algebra morphism $\psi^*\colon \zL\to\zL'$. Hence, for any $\mathrm{a}\in\zL_0$, we get
\begin{gather*}\rmd_{\mathrm{x}^*} \Psi(\mathrm{a}\cdot\mathrm{v}) = (\psi^*(\mathrm{a}\cdot v_\zl^{\mathfrak{a}})) = (\psi^*(\mathrm{a})\cdot \psi^*(v_\zl^{\mathfrak{a}})) = \psi^*(\mathrm{a}) \cdot \rmd_{\mathrm{x}^*} \psi(\mathrm{v}).\end{gather*}
Since the higher order derivatives of $\Psi$ vanish, all its derivatives exist and are continuous, hence, the map $\Psi$ is actually smooth.
\end{proof}

\subsection[The manifold structure on the set of $\Lambda$-points]{The manifold structure on the set of $\boldsymbol{\Lambda}$-points}\label{MSLP}

The next theorem generalizes Propositions \ref{LambdaFrechet} and \ref{prop:varphiSmooth}. For information about Fr\'echet manifolds, we refer to Appendix~\ref{appx:AManifolds}. We recall that the $\zL$-points $M(\zL)$ of a $\Z_2^n$-manifold~$M$ can be equivalently viewed as the maps $m=(|m|,m^*)\in\Hom_{Z_2^n\tt Man}\big(\R^{0|\ul m},M\big)$, as the global pullbacks $m^*=m^*_{|M|}\in\Hom_{\Z_2^n \tt Alg}(\cO_M(|M|),\zL)$, or as the induced morphisms \begin{gather*}m^*_\star\in\Hom_{\Z_2^n \tt Alg}(\cO_{M,x},\zL),\end{gather*} where $x=|m|(\star)\in|M| $. If $M=\cU^{p|\ul q}$ is a $\Z_2^n$-domain, we often write $\mr{x}$ instead of $m$ and we can identify $\mr{x}\simeq\mr{x}^*\simeq\mr{x}^*_\star$ with the pullbacks \begin{gather*}(u_{||}, \rim{u}_\zL,\zr_\zL)\in\cU^p\times\rim{\zL}_0^p \times\prod_i\zL_{\zg_i}^{q_i} \end{gather*} by $\mr{x}^*$ of the coordinate functions $(u,\zr)$ in $\cU^{p|\ul q}$. Recall as well that $\Z_2^n$-morphisms $\zvf\colon M\to N$ are mapped injectively to natural transformations $\zvf\colon M(-)\to N(-)$ with $\zL$-component \begin{gather}\label{InducedComponent}\zvf_\zL\colon \ M(\zL)\ni (x,m^*_\star)\mapsto (|\zvf|(x), m^*_\star\circ\zvf^*_x)\in N(\zL),\end{gather} and that, for any fixed~$M$, a~$\Z_2^n$-Grassmann algebra morphism $\psi^*\colon \zL\to\zL'$ induces a map \begin{gather*}M(\psi^*)\colon \ M(\zL)\ni(x,m^*_\star)\mapsto (x,\psi^*\circ m^*_\star)\in M(\zL').\end{gather*}

\begin{Theorem}\label{thm:ManStruct}Let $M$ be a $\Z_2^n$-manifold, and let $\Lambda$ and $\Lambda'$ be $\Z_2^n$-Grassmann algebras. Then
\begin{enumerate}\itemsep=0pt
\item[$(i)$] $M(\Lambda)$ has the structure of a nuclear Fr\'{e}chet $\Lambda_0$-manifold, and,
\item[$(ii)$] given a morphism of $\Z_2^n$-Grassmann algebras $\psi^* \colon \Lambda \longrightarrow \Lambda'$, the induced mapping $M(\psi^*)$ is $\psi^*$-smooth.
\end{enumerate}
\end{Theorem}

\begin{proof}(i) Let $p|\ul q$ be the dimension of the $\Z_2^n$-manifold $M$. The local $\Z_2^n$-isomorphisms \begin{gather*}h_\za=(|h_\za|,h^*_\za)\colon \ U_\za=(|U_\za|,\cO_M|_{|U_\za|})\to \cU_\za^{p|\ul q}=\big(\cU^p_\za,\Ci_{\R^p}|_{\cU^p_\za}[[\zr]]\big),\end{gather*} where $\za$ varies in some $\mathcal{A}$ and where $|U_\za|\subset|M|$ is open, provide an atlas on $M$ (see paragraph below Definition~\ref{AtlasDefi}). As recalled above, the $\Z_2^n$-isomorphisms \begin{gather*}h_\za\colon \ U_\za\to\cU_\za^{p|\ul q}\end{gather*} implement natural isomorphisms $h_{\za}$ with $\zL$-components \begin{gather}\label{InducedComponentSpec}h_{\za,\zL}\colon \ U_\za(\zL)\ni (x,m^*_\star)\mapsto (|h_\za|(x), m^*_\star\circ (h_\za)^*_x)\in \cU_\za^{p|\ul q}(\zL),\end{gather} whose inverses are the similar maps defined using \begin{gather*}\big|h_\za^{-1}\big|=|h_\za|^{-1}\qquad\text{and}\qquad\big(h_\za^{-1}\big)^*_y=((h_\za)^*_{|h_\za|^{-1}(y)})^{-1} (y\in\cU_\za^p).\end{gather*} The family $(U_\za(\zL),h_{\za,\zL})$ ($\za\in\mathcal{A}$) is an atlas that endows $M(\zL)$ with a nuclear Fr\'echet $\zL_0$-manifold structure. Indeed:
\begin{enumerate}\itemsep=0pt
\item[(a)] Any $h_{\za,\zL}\colon U_\za(\zL)\to \cU_\za^{p|\ul q}(\zL)$ is a bijection valued in the open subset $\cU_\za^{p|\ul q}(\zL)$ of the nuclear Fr\'echet vector space $\R^{p|\ul q}(\zL)$, which is also a Fr\'echet module over the nuclear Fr\'echet algebra~$\zL_0$. Moreover, as the $|U_\za|$ are an open cover of~$|M|$, we have \begin{gather*}M(\Lambda) = \bigcup_{\za \in \mathcal{A}} U_{\za}(\Lambda),\end{gather*} in view of Proposition~\ref{prop:CoverM}.
\item[(b)] The image $h_{\za,\zL}(U_\za(\zL)\cap U_\zb(\zL))$ is open in $\R^{p|\ul q}(\zL)$. To see this, set $|U_{\za\zb}|=|U_\za|\cap|U_\zb|\subset|U_\za|$ and consider the open $\Z_2^n$-submanifold $U_{\za\zb}=(|U_{\za\zb}|,\cO_M|_{|U_{\za\zb}|})$ of $U_\za$. The $\Z_2^n$-isomorphism $h_\za$ restricts to a $\Z_2^n$-isomorphism \begin{gather*}h_\za\colon \ U_{\za\zb}\to \cU_{\za\zb}^{p|\ul q},\end{gather*} where the target is the open $\Z_2^n$-subdomain $\cU_{\za\zb}^{p|\ul q}$ of $\cU_\za^{p|\ul q}$ defined over the open subset \begin{gather*}\cU_{\za\zb}^{p}:=|h_\za|(|U_{\za\zb}|)\subset\cU^{p}_\za,\end{gather*} obtained as the image of the open subset $|U_{\za\zb}|\subset|U_\za|$ by the diffeomorphism $|h_\za|$. The restricted $\Z_2^n$-isomorphism $h_\za$ induces a natural isomorphism $h_\za$, whose $\zL$-component is a~bijection \begin{gather*}h_{\za,\zL}\colon \ U_{\za\zb}(\zL)\to\cU_{\za\zb}^{p|\ul q}(\zL).\end{gather*} Further, we have \begin{gather*}U_{\za\zb}(\zL)=\bigcup_{x\in |U_{\za\zb}|}\Hom_{\Z_2^n\tt Alg}(\cO_{M,x},\zL)\\
 \hphantom{U_{\za\zb}(\zL)}{} = \bigcup_{x\in|U_{\za}|}\Hom_{\Z_2^n\tt Alg}(\cO_{M,x},\zL) \bigcap \bigcup_{x\in|U_{\zb}|}\Hom_{\Z_2^n\tt Alg}(\cO_{M,x},\zL)=U_\za(\zL)\cap U_{\zb}(\zL).\end{gather*} Hence, the image $h_{\za,\zL}(U_\za(\zL)\cap U_\zb(\zL))=\cU^{p|\ul q}_{\za\zb}(\zL)\subset \R^{p|\ul q}(\zL)$ is open.
\item[(c)] We have still to prove that the transition bijections \begin{gather*}h_{\zb,\zL}(h_{\za,\zL})^{-1}\colon \ \cU^{p|\ul q}_{\za\zb}(\zL)\to \cU^{p|\ul q}_{\zb\za}(\zL)\end{gather*} are $\zL_0$-smooth. In view of Theorem~\ref{prop:AMorphisms}, the $\Z_2^n$-isomorphism \begin{gather*}h_\zb h_\za^{-1}\colon \ \cU^{p|\ul q}_{\za\zb}\to\cU^{p|\ul q}_{\zb\za}\end{gather*} induces a natural isomorphism $h_\zb h_\za^{-1}$ with a $\zL_0$-smooth $\zL$-component \begin{gather*}\big(h_\zb h_\za^{-1}\big)_\zL\colon \ \cU^{p|\ul q}_{\za\zb}(\zL)\to\cU^{p|\ul q}_{\zb\za}(\zL).\end{gather*} In view of equations~\eqref{InducedComponent} and~\eqref{InducedComponentSpec}, we get \begin{gather*}\big(h_\zb h_\za^{-1}\big)_\zL(u,\mr{x}^*_\star)=\big(\big|h_\zb h_\za^{-1}\big|(u), \mr{x}^*_\star\circ \big(h_\zb\circ h_\za^{-1}\big)^*_u\big)\\
 \hphantom{\big(h_\zb h_\za^{-1}\big)_\zL(u,\mr{x}^*_\star)}{} = \big(|h_\zb|\big(|h_\za|^{-1}(u)\big), \mr{x}^*_\star\circ \big((h_\za)^*_{|h_\za|^{-1}(u)}\big)^{-1}\circ (h_\zb)^*_{|h_\za|^{-1}(u)}\big)\\
 \hphantom{\big(h_\zb h_\za^{-1}\big)_\zL(u,\mr{x}^*_\star)}{}=h_{\zb,\zL}\big((h_{\za,\zL})^{-1}(u,\mr{x}^*_\star)\big),\end{gather*} for any $(u,\mr{x}^*_\star)\in\cU^{p|\ul q}_{\za\zb}(\zL).$ It follows that $h_{\zb,\zL}(h_{\za,\zL})^{-1}=\big(h_\zb h_\za^{-1}\big)_\zL$ is $\zL_0$-smooth.
\end{enumerate}

(ii) The statement of part (ii) is purely local, see Appendix~\ref{appx:AManifolds}. Let $(x,m^*_\star)\in M(\zL)$, let $(U_\za(\zL),h_{\za,\zL})$ be a chart of $M(\zL)$ around $(x,m^*_\star)$, and let $(U_\zb(\zL'),h_{\zb,\zL'})$ be a chart of $M(\zL')$, such that $M(\psi^*)(U_{\za}(\zL))\subset U_{\zb}(\zL')$. We must show that the local form \begin{gather*}h_{\zb,\zL'}\circ M(\psi^*)\circ(h_{\za,\zL})^{-1}\end{gather*} of $M(\psi^*)$ is $\psi^*$-smooth. Actually, we can choose $(U_\za(\zL'),h_{\za,\zL'})$ as second chart, since the image by $M(\psi^*)$ of a point $(y,n^*_\star)$ in $U_\za(\zL)$, i.e., a point \begin{gather*}(y,n^*_\star)\in\Hom_{\Z_2^n\tt Alg}(\cO_{M,y},\zL)\end{gather*} with $y\in|U_\za|$, is the point \begin{gather*}(y,\psi^*\circ n^*_\star)\in\Hom_{\Z_2^n\tt Alg}(\cO_{M,y},\zL'),\end{gather*} i.e., in $U_\za(\zL')$. From here, we omit subscript $\za$. Since $h\colon U(-)\to\cU^{p|\ul q}(-)$ is a natural transformation, the diagram
\begin{gather*}
\leavevmode
\begin{xy}
(0,15)*+{U(\Lambda)}="a"; (35,15)*+{U(\Lambda')}="b";%
(0,0)*+{\mathcal{U}^{p|\ul{q}}(\Lambda) }="c"; (35,0)*+{\mathcal{U}^{p|\ul{q}}(\Lambda')}="d";%
{\ar "a";"b"}?*!/_3mm/{M(\psi^*)};%
{\ar "a";"c"}?*!/^3mm/{h_\Lambda};{\ar "b";"d"}?*!/_5mm/{h_{\Lambda'}};%
{\ar "c";"d"} ?*!/^5mm/{\cU^{p|\ul q}(\psi^*)};%
\end{xy}
\end{gather*}
commutes. Since $h$ is in fact a natural isomorphism, we get that \begin{gather*}h_{\Lambda'} \circ M(\psi^*) \circ (h_\Lambda)^{-1}=\cU^{p|\ul q}(\psi^*).\end{gather*} From Proposition~\ref{prop:varphiSmooth} we conclude that this local form is indeed $\psi^*$-smooth.
\end{proof}

In view of \eqref{eqn:LambdaCoords}, in general, the local model $\R^{p|\ul{q}}(\zL)$ of $M(\zL)$ is infinite-dimensional, due to the non-zero degree even coordinates of $\zL$. If the particular $\Z_2^n$-Grassmann algebra has no non-zero degree even coordinates, then it is a polynomial algebra and the resulting local model~$\R^{p|\ul{q}}(\zL)$ will, of course, be finite-dimensional. Further, we have the

\begin{Corollary}\label{corl:FunctorLMan}For any $\Z_2^n$-manifold $M$, the associated functor \begin{gather*}M(-)\in\big[\Z_2^n{\tt Pts}^{\op{op}},{\tt Set}\big]\end{gather*} can be considered as a functor \begin{gather*}M(-)\in\big[\Z_2^n{\tt Pts}^{\op{op}},{\tt A(N)FMan}\big],\end{gather*} where the target category is either the category $\tt AFMan$ of Fr\'echet manifolds over a Fr\'echet algebra or the category $\tt ANFMan$ of nuclear Fr\'echet manifolds over a nuclear Fr\'echet algebra, see Appendix~{\rm \ref{appx:AManifolds}}. Therefore, the faithful restricted Yoneda functor $\mathcal{Y}_{\Z_2^n\tt Pts}$, see Corollary~{\rm \ref{cor:RestYon}}, can be viewed as a faithful functor \begin{gather*}\mathcal{Y}_{\Z_2^n\tt Pts}\colon \ \Z_2^n{\tt Man}\to\big[\Z_2^n{\tt Pts}^{\op{op}},{\tt A(N)FMan}\big].\end{gather*}
\end{Corollary}

The latter statement requires that the natural transformation $\zvf\colon M(-)\to N(-)$ induced by a $\Z_2^n$-morphism $\zvf\colon M\to N$ have components $\zvf_\zL\colon M(\zL)\to N(\zL)$ that are morphisms in $\tt A(N)FMan$ between the Fr\'echet $\zL_0$-manifolds $M(\zL)$ and $N(\zL)$, i.e., that the $\zvf_\zL$ be $\zr$-smooth for some morphism $\zr\colon \zL_0\to\zL_0$ of Fr\'echet algebras. We will show in the next subsection that this condition is satisfied for $\zr=\id_{\zL_0}$, i.e., we will show that:

\begin{Proposition}\label{prop:NatMorp} Any natural transformation $\phi\colon M(-)\to N(-)$ that is implemented by a~$\Z_2^n$-morphism $\phi\colon M\to N$ has $\zL_0$-smooth components $\zvf_\zL\colon M(\zL)\to N(\zL)$.
\end{Proposition}

\begin{Theorem}Let $M\in\Z_2^n\tt Man$ be of dimension $p|\ul q$ and let $\zL\in\Z_2^n\tt GrAlg$.
\begin{enumerate}\itemsep=0pt
\item[$(i)$] The nuclear Fr\'echet $\zL_0$-manifold $M(\zL)$ is a fiber bundle in the category $\tt ANFMan$. Its base is the nuclear Fr\'echet $\R$-manifold $M(\R)$, i.e., the smooth manifold $|M|$, and its typical fiber is the nuclear Fr\'echet $\zL_0$-manifold \begin{gather*}
 \zL^{p|\ul q}:=\rim{\zL}_0^p\times \prod_{i=1}^N\zL_{\zg_i}^{q_i}.\end{gather*}
\item[$(ii)$] The topology of $M(\zL)$, which is defined, as in the case of smooth manifolds, by the atlas providing the nuclear Fr\'echet $\zL_0$-structure, is a Hausdorff topology, so that~$M(\zL)$ is a~genuine Fr\'echet manifold.
\end{enumerate}
\end{Theorem}

\begin{proof}(i) We think of fiber bundles in $\tt ANFMan$ exactly as of fiber bundles in the category of smooth manifolds. Of course, in such a fiber bundle, all objects and arrows are $\tt ANFMan$-objects and $\tt ANFMan$-morphisms.

Let $p^*\colon \zL\to \R$ be, as above, the canonical $\Z_2^n\tt GrAlg$-morphism. The induced map \begin{gather*}\zp:=M(p^*)\colon \ M(\zL)\ni(x,m^*_\star)\mapsto (x,p^*\circ m^*_\star)\simeq x\in M(\R)\simeq |M|\end{gather*} is $p^*$-smooth, i.e., is a morphism in the category $\tt ANFMan$.

We will show that $\zp$ is surjective and that the local triviality condition is satisfied.

Let $z\in|M|$. There is a $\Z_2^n$-chart $(U,h)$ of $M$, such that $|U|\subset|M|$ is a neighborhood of $z$. The $\Z_2^n$-isomorphism $h\colon U\to \cU^{p|\ul q}$ induces a natural isomorphism $h$, whose $\zL$-components are $\zL_0$-diffeomorphisms, i.e., $\zL_0$-smooth maps that have a $\zL_0$-smooth inverse. We have the following commutative diagram:
\begin{gather*}
\leavevmode
\begin{xy}
(0,20)*+{U(\Lambda)}="a"; (55,20)*+{\cU^{p|\ul q}(\zL)\simeq \mathcal{U}^p \times\zL^{p|\ul q}} ="b";
(0,0)*+{U(\R)\simeq|U|}="c"; (55,0)*+{\cU^{p|\ul q}(\R)\simeq\mathcal{U}^p ,}="d";%
{\ar@{<->} "a";"b"}?*!/_4mm/{h_\Lambda};%
{\ar "a";"c"}?*!/^13mm/{U(p^*)=\zp|_{U(\zL)}};%
{\ar "b";"d"}?*!/_16mm/{\cU^{p|\ul q}(p^*)\simeq\textnormal{prj}_1};%
{\ar@{<->} "c";"d"}?*!/^3mm/{h_\R=|h|};%
\end{xy}
\end{gather*}
where $\op{prj}_1$ is the canonical projection. Let us explain that $\cU^{p|\ul q}(p^*)\simeq\op{prj}_1$, when read through $\flat\colon \cU^p\times\zL^{p|\ul q}\leftrightarrow\cU^{p|\ul q}(\zL)$. We need a more explicit description of the equivalent views on $\zL$-points of a $\Z_2^n$-domain, see beginning of Section~\ref{MSLP}. As elsewhere in this text, we denote a~$\Z_2^n$-morphism $\R^{0|\ul m}\to \cU^{p|\ul q}$ by $\mr{x}=(|\mr{x}|,\mr{x}^*)$ and we denote the morphism it induces between the stalks $\cO_{\cU^{p|\ul q},|\mr{x}|(\star)}\to\zL$ by $\mr{x}^*_\star$. The morphism $\flat$ is the succession of identifications
\begin{gather}\label{Nice}
\cU^p\times\zL^{p|\ul q}\ni(x_{||},\rim{x}_\zL,\zx_\zL)\simeq \mr{x}=(|\mr{x}|,\mr{x}^*)\simeq (|\mr{x}|(\star),\mr{x}^*_\star)\in\cU^{p|\ul q}(\zL),
\end{gather}
where the components of the base morphism $|\mr{x}|$ are obtained (see~\cite{Covolo:2016}) by applying the base projection $\ze_{\star}\colon \zL\!\to\! \R$ of $\R^{0|\ul m}$, i.e., the canonical morphism $p^*$, to the components \smash{$x^a_\zL=(x^a_{||},\rim{x}^a_\zL)\!\in\!\zL_0$}. Hence, we get \begin{gather}\label{Interesting}|\mr{x}|(\star)=|\mr{x}|=\big(\dots, p^*(x^a_\zL),\dots\big)=x_{||}.\end{gather} Therefore, we actually obtain that \begin{gather*}\cU^{p|\ul q}(p^*)(\flat(x_{||},\rim{x}_\zL,\zx_\zL))=(|\mr{x}|(\star),p^*\circ\mr{x}^*_\star)\simeq |\mr{x}|(\star)=x_{||}=\op{prj}_1(x_{||},\rim{x}_\zL,\zx_\zL).\end{gather*}

Since $\zp|_{U(\zL)}=|h|^{-1}\circ\op{prj}_1\circ h_\zL$, the local projection $\zp|_{U(\zL)}$ is surjective, so that $z$ is in the image of $\zp$, which is thus surjective as well.

As just mentioned, we started from $z\in|M|$ and found a neighborhood $|U|$ of $z$ and a $\zL_0$-diffeomorphism $h_\zL$. When identifying $|U|$ with $\cU^p$ via $|h|$ (which then becomes $\id$), we get the $\zL_0$-diffeomorphism \begin{gather}\label{Trivialization}h_\zL\colon \ \zp^{-1}(|U|)\simeq U(\zL)\ni (y,m^*_\star)\mapsto (y,m^*_\star\circ h^*_y)\in|U|\times\zL^{p|\ul q}.\end{gather} Observe that in equation \eqref{Trivialization} we used $\flat^{-1}$ defined in equations~\eqref{Nice} and~\eqref{Interesting}, thus identifying \begin{gather*}(y,m^*_\star\circ h^*_y)\in\Hom_{\Z_2^n\tt Alg}(\cO_{\cU^{p|\ul q},y},\zL) \subset \cU^{p|\ul q}(\zL)\end{gather*} with $h\circ m\in\Hom_{\Z_2^n\tt Man}\big(\R^{0|\ul m},\cU^{p|\ul q}\big)$, and then with \begin{gather*}\big(y,\op{pr}_2(m^*(h^*(x))),m^*(h^*(\zx))\big)\in |U|\times\zL^{p|\ul q},\end{gather*} where we denoted the projection of $\zL_0$ onto $\rim{\zL}_0$ by $\op{pr}_2$. Notice also that the conclusion that $\zL^{p|\ul q}$ is a nuclear Fr\'echet $\zL_0$-manifold comes from the facts that any subspace (resp., any closed subspace) of a nuclear (resp., a Fr\'echet) space is a nuclear (resp., a Fr\'echet) space.

Hence, the trivialization condition is satisfied as well, and $M(\zL)$ is a fiber bundle in $\tt ANFMan$, as announced.

(ii) Now consider two different $\Lambda$-points $m^* = (x, m^*_\star)$ and $n^* = (y, n^*_\star)$ in $M(\zL)$. If $x\neq y$, then, as $|M|$ is Hausdorff, there exist open neighborhoods $|U|$ of $x$ and $|V|$ of $y$, such that $|U|\cap |V| = \varnothing$. When denoting the corresponding open $\Z_2^n$-submanifolds by $U$ and $V$, respectively, we get open neighborhoods $U(\zL)$ and $V(\zL)$ of $m^*$ and $n^*$, such that $U(\Lambda)\cap V(\Lambda) = \varnothing$. We have of course to check that, for any $\Z_2^n$-chart $(U_\za,h_\za)$, the image \begin{gather*}h_{\za,\zL}(U_\za(\zL)\cap U(\zL))\end{gather*} is open in $\R^{p|\ul q}(\zL)$, and similarly for $V(\zL)$. To see this, it suffices to proceed as in the proof of Theorem~\ref{thm:ManStruct}.

Next, consider the situation where $x=y=:z\in |M|$, use the trivialization constructed in~(i), and denote the canonical projection from $\cU^p\times\zL^{p|\ul q}$ onto $\zL^{p|\ul q}$ by $\op{prj}_2$. As $m^*\neq n^*$, we have $h_\zL(m^*)\neq h_\zL(n^*)$, i.e., \begin{gather*}(|h|(z),\op{prj}_2(h_\zL(m^*)))\neq(|h|(z),\op{prj}_2(h_\zL(n^*))).\end{gather*} Since $\op{prj}_2(h_\zL(m^*))\neq\op{prj}_2(h_\zL(n^*))$ are points in the Hausdorff space $\zL^{p|\ul q}$, there are open neighborhoods $V_{m^*}$ and $V_{n^*}$ of these projections that do not intersect. The preimages $U_{m^*}$ and $U_{n^*}$ of $V_{m^*}$ and $V_{n^*}$ by the continuous map \begin{gather*}\op{prj_2}\circ h_\zL\colon \ U(\zL)\to\zL^{p|\ul q}\end{gather*} are then open neighborhoods of $m^*$ and $n^*$ that do not intersect.

Finally, the space $M(\Lambda)$ is indeed a Hausdorff topological space.
\end{proof}

\subsection{The Schwarz--Voronov embedding}
In order to get a fully faithful functor, hence, to embed the category $\Z_2^n\tt Man$ as full sub\-ca\-tegory into a functor category, we need to replace the target category $[\Z_2^n{\tt Pts}^{\op{op}},{\tt A(N)FMan}]$ by a~subcategory that we denote by $\SN{\Z_2^n{\tt Pts}^{\op{op}},{\tt A(N)FMan}}$ and that we define as follows:

\begin{Definition}The category $\SN{\Z_2^n\catname{Pts}^{\textnormal{op}}, {\tt A(N)FMan}}$ is the subcategory of the category $[\Z_2^n\catname{Pts}^{\textnormal{op}},\!\allowbreak {\tt A(N)FMan}]$,
\begin{enumerate}\itemsep=0pt
\item[(i)] whose objects are the functors $\cF$, such that, for any $\zL\in\Z_2^n{\tt Pts}^{\op{op}}$, the value $\cF(\zL)$ is a~(nuclear) Fr\'echet $\zL_0$-manifold, and
\item[(ii)] whose morphisms are natural transformations $\zh\colon \cF\to\mathcal{G}$, such that, for any $\zL$, the component $\zh_\Lambda\colon \mathcal{F}(\Lambda) \to \mathcal{G}(\Lambda)$ is $\Lambda_{0}$-smooth.\end{enumerate}
\end{Definition}

\begin{Proposition}\label{RYF}The restricted Yoneda functor $\mathcal{Y}_{\Z_2^n{\tt Pts}}$ can be considered as a faithful functor \begin{gather*}\mathcal{S}\colon \ \Z_2^n{\tt Man}\to \SN{\Z_2^n\catname{Pts}^{\textnormal{op}}, {\tt A(N)FMan}}.\end{gather*}
\end{Proposition}

\begin{proof}The image $\mathcal{Y}_{\Z_2^n{\tt Pts}}(M)$ of an object $M\in\Z_2^n{\tt Man}$ is a functor $M(-)\!\in\![\Z_2^n{\tt Pts}^{\op{op}},{\tt A(N)FMan}]$, such that, for any $\zL,$ the value $M(\zL)$ is a (nuclear) Fr\'echet $\zL_0$-manifold. Further, the image $\mathcal{Y}_{\Z_2^n{\tt Pts}}(\phi)$ of a $\Z_2^n$-morphism $\phi\colon M\to N$ is a natural transformation $\phi\colon M(-)\to N(-)$, such that, for any $\zL$, the component $\phi_\zL\colon M(\zL)\to N(\zL)$ is $\zL_0$-smooth.

The proof of $\zL_0$-smoothness uses the following construction, which we will also need later on. Let $M,N\in\Z_2^n{\tt Man}$ be manifolds of dimension $p|\ul q$ and $r|\ul s $, respectively, let $|\phi|\in\Ci(|M|,|N|)$, and let $(|V_\zb|)_\zb$ be an open cover of $|N|$ by $\Z_2^n$-charts \begin{gather*}g_\zb\colon \ V_\zb\to\cV_\zb^{r|\ul s}, \qquad \text{where}\quad V_\zb=(|V_\zb|,\cO_N|_{|V_\zb|}).\end{gather*} The open subsets $|U_\zb|:=|\phi|^{-1}(|V_\zb|)\subset|M|$ cover $|M|$, and each $|U_\zb|$ can be covered by $\Z_2^n$-charts \begin{gather*}h_{\zb\za}\colon \ U_{\zb\za}\to \cU_{\zb\za}^{p|\ul q}, \qquad \text{where} \quad U_{\zb\za}=(|U_{\zb\za}|,\cO_M|_{|U_{\zb\za}|}).\end{gather*}

The $\Z_2^n$-morphism $\phi\colon M\to N$ restricts to a $\Z_2^n$-morphism $\phi|_{U_{\zb\za}}\colon U_{\zb\za}\to V_\zb$. In particular, the composite \begin{gather*}g_\zb\circ\phi|_{U_{\zb\za}}\circ (h_{\zb\za})^{-1}\colon \ \cU_{\zb\za}^{p|\ul q}\to \cV_{\zb}^{r|\ul s}\end{gather*} is a $\Z_2^n$-morphism.

We now show that $\phi_\zL$ is $\zL_0$-smooth. Therefore, let $(x,m^*_\star)\in M(\zL)$. There is a $\Z_2^n$-chart $(V_\zb,g_\zb)$ of $N$ such that $|\phi|(x)\in|V_\zb|$, and there is a $\Z_2^n$-chart $(U_{\zb\za},h_{\zb\za})$ of $M$ such that $x\in|U_{\zb\za}|$. These charts (we omit in the following the subscripts $\zb$ and $\za$) induce charts $(U(\zL),h_\zL)$ of $M(\zL)$ around $(x,m^*_\star)$, and $(V(\zL),g_{\zL})$ of $N(\zL)$ such that $\phi_\zL(U(\zL))\subset V(\zL)$. It suffices to show (see Appendix \ref{appx:AManifolds}) that the local form \begin{gather*}g_\zL\circ\phi_\zL\circ (h_\zL)^{-1}=\big(g\circ\phi|_{U}\circ h^{-1}\big)_\zL\end{gather*} is $\zL_0$-smooth. This is the case in view of Theorem~\ref{prop:AMorphisms}. Finally, the faithfulness is established in Corollary~\ref{cor:RestYon}. This completes the proof.
\end{proof}

We will prove that the functor $\mathcal S$ is fully faithful, hence, injective (up to isomorphism) on objects. Therefore, it embeds the category $\Z_2^n\tt Man$ of $\Z_2^n$-manifolds as full subcategory into the larger functor category $\SN{\Z_2^n\catname{Pts}^{\textnormal{op}}, {\tt A(N)FMan}}$.

\begin{Definition}\label{def:SVMembedding}
We refer to the faithful functor \begin{gather*}\mathcal{S}\colon \ \Z_2^n\catname{Man} \longrightarrow\SN{\Z_2^n\catname{Pts}^{\textnormal{op}}, {\tt A(N)FMan}}\end{gather*}
as the \emph{Schwarz--Voronov embedding}.
\end{Definition}

\begin{Theorem}\label{thm:SVMembedding}The Schwarz--Voronov embedding $\mathcal S$ is a fully faithful functor. That is, given two $\Z_2^n$-manifolds $M$ and $N$, the injective map
\begin{gather*}\mathcal{S}_{M,N}\colon \ \emph{\Hom}_{\Z_2^n\catname{Man}}\big( M, N \big) \to \emph{\Hom}_{\SN{\Z_2^n\catname{Pts}^{\textnormal{op}}, {\tt A(N)FMan}}} \big( M(-), N(-)\big)\end{gather*}
is bijective.
\end{Theorem}

\begin{proof}Notice first that it follows from the results of~\cite{Bruce:2018b} and Lemma~\ref{RestTarget} that there is a $1:1$ correspondence \begin{gather*}|M|\simeq \Hom_{\Z_2^n\tt Alg}(\cO_M(|M|),\R)\simeq \bigcup_{x\in|M|}\Hom_{\Z_2^n\tt Alg}(\cO_{M,x},\R)= M(\R),\end{gather*} which is given by \begin{gather*}x\mapsto \ze_x\mapsto (x,\ze_x),\end{gather*} where $\ze_x$ is the evaluation map $\ze_x(f)=(\ze f)(x)$ ($f\in\cO_M(|M|)$) and where $\ze$ is the base map $\ze\colon \cO_M\to\Ci_{|M|}$. Hence, any $(x,m^*_\star)\in M(\R)$ is equal to $(x,\ze_x)$ and can be identified with $x$. In view of \eqref{InducedComponentSpec}, this $1:1$ correspondence identifies the nuclear Fr\'echet $\R$-manifold structure on $M(\R)$ with the smooth manifold structure on $|M|$.

Let now
\begin{gather*}\zh\colon \ M(-) \to N(-)\end{gather*}
be a natural transformation in the target set of $\mathcal{S}_{M,N}$, i.e., a natural transformation such that, for any $\zL$, the $\zL$-component $\zh_\zL$ is $\zL_0$-smooth. In particular, the map
\begin{gather*}|\phi| := \zh_\R \colon \ |M| \to |N|,\end{gather*}
is a smooth map between the reduced manifolds. As in the proof of Proposition \ref{RYF}, let $(V_\zb,g_\zb)_\zb$ be an open cover of $|N|$ by $\Z_2^n$-charts, and, for any $\zb$, let $(U_{\zb\za},h_{\zb\za})_\za$ be an open cover of $|U_\zb|:=|\phi|^{-1}(|V_\zb|)$ by $\Z_2^n$-charts. When denoting the canonical $\Z_2^n$-Grassmann algebra morphism $\zL\to\R$ by $p^*$, we get the commutative diagram
\begin{gather*}
\leavevmode
\begin{xy}
(0,20)*+{\bigcup_{\zb\za}U_{\zb\za}(\zL)}="a"; (40,20)*+{\bigcup_{\zb}V_\zb(\zL)}="b";
(0,0)*+{\bigcup_{\zb\za}|U_{\zb\za}|}="c"; (40,0)*+{\bigcup_{\zb}|V_\zb|,}="d";%
{\ar "a";"b"}?*!/_3mm/{\zh_\Lambda};%
{\ar "a";"c"}?*!/^6mm/{M(p^*)};%
{\ar "b";"d"}?*!/_10mm/{N(p^*) \quad};%
{\ar "c";"d"}?*!/^3mm/{|\phi|};%
\end{xy}
\end{gather*}
which shows that, for any $\zb$, $\za$, we get the $\zL_0$-smooth map \begin{gather*}(\zh_\zL)|_{U_{\zb\za}(\zL)}\colon \ U_{\zb\za}(\zL)\to V_\zb(\zL).\end{gather*} Indeed, if, for $(x,m^*_\star)\in U_{\zb\za}(\zL)$, we set $\zh_\zL(x,m^*_\star)=(y,n^*_\star)$, the commutativity of the diagram implies that \begin{gather*}y\simeq (y,p^*\circ n^*_\star)=(N(p^*)\circ \zh_\zL)(x,m^*_\star)=(\zh_\R\circ M(p^*))(x,m^*_\star)=\zh_\R(x,p^*\circ m^*_\star)\\
\hphantom{y}{}\simeq|\phi|(x)\in|V_\zb|.\end{gather*} Therefore, the restriction \begin{gather*}\zh|_{U_{\zb\za}(-)}\colon \ U_{\zb\za}(-)\to V_\zb(-)\end{gather*} is a natural transformation with $\zL_0$-smooth components.

Note that \begin{gather*}h_{\zb\za}\colon \ U_{\zb\za}\to\cU_{\zb\za}^{p|\ul q}\qquad\text{and}\qquad g_\zb\colon \ V_\zb\to\cV^{r|\ul s}_\zb\end{gather*} are $\Z_2^n$-isomorphisms and induce natural isomorphisms, also denoted by $h_{\zb\za}$ and $g_\zb$, whose components are chart diffeomorphisms \begin{gather*}h_{{\zb\za},\zL}\colon \ U_{\zb\za}(\zL)\to\cU_{\zb\za}^{p|\ul q}(\zL)\qquad\text{and}\qquad g_{\zb,\zL}\colon \ V_\zb(\zL)\to\cV^{r|\ul s}_\zb(\zL)\end{gather*} of nuclear Fr\'echet $\zL_0$-manifolds. The local form
\begin{gather*}g_{\zb,\Lambda} \circ (\zh_\Lambda)|_{U_{\zb\za}(\Lambda)}\circ(h_{\zb\za,\zL})^{-1} \colon \ \mathcal{U}^{p|\ul{q}}_{\zb\za}(\Lambda) \to \mathcal{V}_\zb^{r|\ul s}(\Lambda) \end{gather*}
of $\zh_\zL$ is thus $\zL_0$-smooth. In other words, any $\zL$-component of the natural transformation \begin{gather}\label{LocRestNat} \zf_{\zb\za}:=g_\zb\circ\zh|_{U_{\zb\za}(-)}\circ h_{\zb\za}^{-1}\colon \ \cU^{p|\ul q}_{\zb\za}(-)\to\cV^{r|\ul s}_\zb(-) \end{gather} between functors associated to $\Z_2^n$-domains, is $\zL_0$-smooth. It therefore follows from Theorem~\ref{prop:AMorphisms} that $\zf_{\zb\za}$ is implemented by a $\Z_2^n$-morphism \begin{gather*}\zf_{\zb\za}\colon \ \cU_{\zb\za}^{p|\ul q}\to\cV_\zb^{r|\ul s},\end{gather*} so that the composite \begin{gather}\label{LocZMor}\phi_{\zb\za}:=g_\zb^{-1}\circ\zf_{\zb\za}\circ h_{\zb\za}\colon \ U_{\zb\za}\to N\end{gather} is a $\Z_2^n$-morphism that is defined on an open $\Z_2^n$-submanifold of $M$. The question is whether we can patch together these locally defined $\Z_2^n$-morphisms, which are inherited from $\zh$, and get a~globally defined $\Z_2^n$-morphism $\phi\colon M\to N$ that induces $\zh$.

Let $\phi_{\zb\za}|_{U_{\zb\za,\zn\zm}}$ and $\phi_{\zn\zm}|_{U_{\zb\za,\zn\zm}}$ be the $\Z_2^n$-morphisms obtained by restriction to the open $\Z_2^n$-submanifold $U_{\zb\za,\zn\zm}$ with base manifold $|U_{\zb\za,\zn\zm}|:=|U_{\zb\za}|\cap |U_{\zn\zm}|$. They coincide as $\Z_2^n$-morphisms, if they do as associated natural transformations, i.e., if all $\zL$-components of those transformations coincide. This is the case since both $\zL$-components are equal to $\zh_\zL|_{U_{\zb\za,\zn\zm}(\zL)}$. It follows that the $\Z_2^n$-algebra morphisms \begin{gather*}\phi_{\zb\za}|_{U_{\zb\za,\zn\zm}}^*,\phi_{\zn\zm}|_{U_{\zb\za,\zn\zm}}^*\colon \ \cO_N(|N|)\to \cO_M(|U_{\zb\za,\zn\zm}|) \end{gather*} coincide. This implies that we can glue the $\Z_2^n$-algebra morphisms $\zvf_{\zb\za}^*\colon \cO_N(|N|)\to \cO_M(|U_{\zb\za}|)$ and get a $\Z_2^n$-algebra morphism \begin{gather*}\zvf^*\colon \ \cO_N(|N|)\to \cO_M(|M|).\end{gather*} Indeed, for any $f\in\cO_N(|N|)$, the $\zvf_{\zb\za}^*(f)\in\cO_M(|U_{\zb\za}|)$ are a family of $\Z_2^n$-functions on an open cover of $|M|$, which do coincide on the intersections. To see this, note that \begin{gather*}(\zvf_{\zb\za}^*(f))|_{|U_{\zb\za,\zn\zm}|}=\zvf_{\zb\za}|_{U_{\zb\za,\zn\zm}}^*(f)=\zvf_{\zn\zm}|_{U_{\zb\za,\zn\zm}}^*(f) =(\zvf_{\zn\zm}^*(f))|_{|U_{\zb\za,\zn\zm}|}.\end{gather*} Hence, there is a unique global section $F\in\cO_M(|M|)$ of the sheaf $\cO_M$, such that $F|_{|U_{\zb\za}|}=\zvf_{\zb\za}^*(f)$. The $\tt Set$-morphism, which is defined by \begin{gather*}\zvf^*_{|N|}\colon \ \cO_N(|N|)\ni f\mapsto F\in\cO_M(|M|),\end{gather*} is actually a morphism of $\Z_2^n$-algebras. Indeed, note that \begin{gather*}\zr^{|M|}_{|U_{\zb\za}|}\circ\phi^*_{|N|}=\phi^*_{\zb\za}\end{gather*} ($\zr$ is the restriction) and observe that, for any element $|U_{\zb\za}|$ of the open cover of $|M|$ considered, we have \begin{gather*}(\zvf^*_{|N|}(f\cdot g))|_{|U_{\zb\za}|}=\zvf^*_{\zb\za}(f)\cdot \zvf^*_{\zb\za}(g)=(\zvf^*_{|N|}(f)\cdot\zvf^*_{|N|}(g))|_{|U_{\zb\za}|}.\end{gather*} The $\Z_2^n$-algebra morphism $\zvf^*_{|N|}$ fully characterizes a $\Z_2^n$-morphism $\zvf=(||\phi||,\phi^*)\colon M\to N $. We will show that $\zvf$ induces the natural transformation $\zh$, which then completes the proof.

Since $\zvf$ is glued from the $\Z_2^n$-morphisms $\zvf_{\zb\za}$, we get, in view of equations \eqref{LocRestNat} and \eqref{LocZMor}, in particular that \begin{gather}\label{Induction1}||\zvf|| |_{|U_{\zb\za}|}=|\zvf_{\zb\za}|=\zh_\R|_{U_{\zb\za}(\R)}=|\zvf| |_{|U_{\zb\za}|},\end{gather} so that $||\zvf||=|\zvf|$. Further, for any $|V_\zb|,$ \begin{gather}\label{Induction2}\zr^{|U_\zb|}_{|U_{\zb\za}|}\circ\phi^*_{|V_{\zb}|}=\phi^*_{\zb\za,|V_\zb|}\colon \ \cO_N(|V_\zb|)\to \cO_M(|U_{\zb\za}|).\end{gather} Let now $\zL$ be any $\Z_2^n$-Grassmann algebra and let $(x,m^*_\star)\in U_{\zb\za}(\zL)$. As $x\in|U_{\zb\za}|$ and $|\phi|(x)\in|V_\zb|$, it follows from equations \eqref{Induction1}, \eqref{Induction2}, \eqref{LocRestNat}, and~\eqref{LocZMor}, that the image of $(x,m^*_\star)$ by the $\zL$-component of the natural transformation induced by $\zvf$ is \begin{gather*}\phi_\zL(x,m^*_\star)=(|\phi|(x),m^*_\star\circ\phi^*_x)=(|\phi_{\zb\za}|(x),m^*_\star\circ\phi^*_{\zb\za,x})
=(\phi_{\zb\za})_\zL(x,m^*_\star)=\zh_\zL(x,m^*_\star).\tag*{\qed} \end{gather*}\renewcommand{\qed}{}
\end{proof}

The following theorem is of importance in the study of $\Z_2^n$-Lie groups.

\begin{Theorem}The Schwarz--Voronov embedding $ \mathcal{S}$ sends $\Z_2^n$-Lie groups $G$ to functors $\mathcal{S}(G) = G(-)$ from the category $\Z_2^n\tt Pts^{\op{op}}\!$ of $\Z_2^n$-Grassmann algebras to the category $\tt ANFLg$ of nuclear Fr\'echet Lie groups over nuclear Fr\'echet algebras.
\end{Theorem}

The proof is not entirely straightforward and will be given in a paper on $\Z_2^n$-Lie groups, which is currently being written down.

\subsection{Representability and equivalence of categories}
As the Schwarz--Voronov embedding is fully faithful, the category $\Z_2^n\tt Man$ can be viewed as a full subcategory of the category $\SN{\Z_2^n\catname{Pts}^{\textnormal{op}}, {\tt A(N)FMan}}$. Functor categories are known to be well-suited for geometric constructions. Hence, when trying to build a $\Z_2^n$-manifold $M$ (possibly from other $\Z_2^n$-manifolds $M_\iota$), it is often easier to build a functor $\cF$ in $\SN{\Z_2^n\catname{Pts}^{\textnormal{op}}, {\tt A(N)FMan}}$ (from the given $\Z_2^n$-manifolds interpreted as functors $M_\iota(-)$). However, one has then to show that $\cF$ can be represented by a $\Z_2^n$-manifold $M$, i.e., that there is a $\Z_2^n$-manifold $M$, such that $M(-)\simeq\cF$.

\begin{Definition}A functor \begin{gather*}\mathcal{F}\in\SN{\Z_2^n\catname{Pts}^{\textnormal{op}}, {\tt A(N)FMan}}\end{gather*} is said to be \emph{representable}, if there exists a $\Z_2^n$-manifold $M\in\Z_2^n\tt Man$ (which is then unique up to unique isomorphism), such that \begin{gather*}M(-)\simeq \mathcal{F}\quad\text{in}\quad\SN{\Z_2^n\catname{Pts}^{\textnormal{op}}, {\tt A(N)FMan}}.\end{gather*}
\end{Definition}

We define the restriction $\cF|_{|U|}$ of a functor $\mathcal{F}\in\SN{\Z_2^n\catname{Pts}^{\textnormal{op}}, {\tt A(N)FMan}}$ to an open subset $|U|\subset\cF(\R)\in\tt(N)FMan$.

For any $\Lambda \in \Z_2^n\catname{GrAlg}$, let
\begin{gather*}p^*_\Lambda \colon \ \Lambda \longrightarrow \R\end{gather*}
be the canonical projection, let
\begin{gather*}
\mathcal{F}(p^*_\Lambda ) \colon \ \mathcal{F}(\Lambda) \longrightarrow \mathcal{F}(\R)\end{gather*}
be the corresponding smooth map. The preimage
\begin{gather}\label{RestObj}\cF|_{|U|}(\zL):=(\mathcal{F}(p^*_\Lambda ))^{-1}(|U|)\end{gather}
is an open (nuclear) Fr\'echet $\Lambda_{0}$-submanifold of $\mathcal{F}(\Lambda)$.

Consider now a morphism $\varphi^* \colon \Lambda \longrightarrow \Lambda'$ in $\Z_2^n\tt GrAlg$. As $p^*_{\Lambda'} \circ \varphi^*=p^*_\Lambda$, we get the restriction
\begin{gather}\label{RestMor}\cF|_{|U|}(\zf^*):= \mathcal{F}(\varphi^*)\big|_{\cF|_{|U|}(\zL)} \colon \ \cF|_{|U|}(\zL) \longrightarrow \cF|_{|U|}(\zL'),\end{gather}
which is a morphism in $\tt A(N)FMan$.

\begin{Definition}\label{def:Localfunc}For any functor \begin{gather*}\mathcal{F}\in\SN{\Z_2^n\catname{Pts}^{\textnormal{op}}, {\tt A(N)FMan}}\end{gather*} and any open subset $|U| \subset \mathcal{F}(\R)$, the \emph{restriction of $\mathcal{F}$ to $|U|$} is the functor \begin{gather*}\cF|_{|U|}\in\SN{\Z_2^n\catname{Pts}^{\textnormal{op}}, {\tt A(N)FMan}} \end{gather*} that is defined by equations~\eqref{RestObj} and~\eqref{RestMor}.\end{Definition}

\begin{Example}\label{ExampleRest}Let $M\in\Z_2^n\tt Man$, let $M(-)$ be the corresponding functor, and let $|U|\subset|M|\simeq M(\R)$ be an open subset. The restriction $M(-)|_{|U|}$ is given:
\begin{enumerate}\itemsep=0pt
\item[(i)] on objects $\zL$, by \begin{gather}\label{=1}M(-)|_{|U|}(\zL):=\{(x,m^*_\star)\in M(\zL)\colon (x,p^*_\zL\circ m^*_\star)\simeq x\in|U|\}=U(\zL),\end{gather} where $U=(|U|,\cO_M|_{|U|})$ is the open $\Z_2^n$-submanifold of $M$ over $|U|$, and
\item[(ii)] on morphisms $\varphi^*\colon \zL\to\zL'$, by \begin{gather}\label{=2}M(-)|_{|U|}(\varphi^*):=M(\varphi^*)|_{U(\zL)}=U(\varphi^*),\end{gather} since both maps are given by \begin{gather*}U(\zL)\ni(x,m^*_\star)\mapsto (x,\varphi^*\circ m^*_\star)\in U(\zL').\end{gather*}
\end{enumerate}
\end{Example}

Let $\cF$ be representable, let $M$ be `its' representing $\Z_2^n$-manifold, and let \begin{gather}\label{Sim1}\zh\colon \ \cF\to M(-)\end{gather} be the corresponding natural isomorphism in $\SN{\Z_2^n\catname{Pts}^{\textnormal{op}}, {\tt A(N)FMan}}$. The maps $\zh_\zL$ and $\zh^{-1}_\zL$ are then $\zL_0$-smooth, i.e., $\zh_\zL$ is a $\zL_0$-diffeomorphism, for any~$\zL$. In particular, the map $\zh_\R\colon \cF(\R)\to M(\R)$ is a diffeomorphism of (nuclear) Fr\'echet manifolds. This means that the (nuclear) Fr\'echet structures on $\cF(\R)\simeq M(\R)$ coincide. Further, if one identifies $\cF(\R)\simeq M(\R)$ with $|M|,$ the (nuclear) Fr\'echet structure on $\cF(\R)\simeq M(\R)$ coincides with the smooth structure on $|M|$. We can therefore view $\cF(\R)$ as being the smooth manifold $|M|$. Consider now a $\Z_2^n$-atlas $(U_\za,h_\za)_\za$ of $M$. If we denote the dimension of $M$ by $p|\ul q$, the $\Z_2^n$-chart map $h_\za$ is a $\Z_2^n$-isomorphism \begin{gather*}h_\za\colon \ U_\za\to \cU_\za^{p|\ul q}\end{gather*} valued in a $\Z_2^n$-domain of dimension $p|\ul q$, which implies that \begin{gather}\label{Sim2}h_\za\colon \ U_\za(-)\to \cU_\za^{p|\ul q}(-)\end{gather} is a natural isomorphism in $\SN{\Z_2^n\catname{Pts}^{\textnormal{op}}, {\tt A(N)FMan}}$. In view equations \eqref{Sim1}, \eqref{=1}, \eqref{=2}, and~\eqref{Sim2}, the family $(|U_\za|)_\za$ is an open cover of $|M|\simeq\cF(\R)$, such that, for any $\za$, we have \begin{gather*}\cF|_{|U_\za|}\simeq M(-)|_{|U_\za|}=U_\za(-)\simeq \cU_\za^{p|\ul q}(-)\end{gather*} in $\SN{\Z_2^n\catname{Pts}^{\textnormal{op}}, {\tt A(N)FMan}}$.

\begin{Theorem}\label{trm:ReprCon}A functor $\mathcal{F}\in\SN{\Z_2^n\catname{Pts}^{\textnormal{op}}, {\tt A(N)FMan}}$ is representable if and only if there exists an open cover $(|U_\za|)_\za$ of $\mathcal{F}(\R)$, such that, for each $\za$, we have \begin{gather}\label{ReprCond}\mathcal{F}|_{|U_\za|} \simeq \mathcal{U}^{p|\ul{q}}_\za(-)\end{gather} in $\SN{\Z_2^n\catname{Pts}^{\textnormal{op}}, {\tt A(N)FMan}}$, where $ \mathcal{U}^{p|\ul{q}}_\za$ is a $\Z_2^n$-domain in a fixed $\R^{p|\ul{q}}$.
\end{Theorem}

\begin{proof}We showed already that the condition is necessary. Assume now that condition~\eqref{ReprCond} is satisfied, i.e., that we have natural isomorphisms \begin{gather*}k_\za\colon \ \cF|_{|U_\za|}\to\cU^{p|\ul q}_\za(-)\end{gather*} in $\SN{\Z_2^n\catname{Pts}^{\textnormal{op}}, {\tt A(N)FMan}}$. This means that the $\zL$-components
\begin{gather*}k_{\za,\Lambda} \colon \ \mathcal{F}|_{|U_\za|}(\Lambda) \rightarrow \mathcal{U}^{p|\ul{q}}_\za(\Lambda)\end{gather*}
are $\Lambda_0$-diffeomorphisms.

In particular, we have a diffeomorphism
\begin{gather*}|h_\za|:=k_{\za,\R} \colon \ \cF|_{|U_\za|}(\R)=(\cF(p^*_\R))^{-1}(|U_\za|)=|U_\za| \rightarrow \mathcal{U}_\za^{p|\ul q}(\R)\simeq \cU_\za^{p}.\end{gather*}
Notice that $(|U_\za|,|h_\za|)_\za$ can be interpreted as a smooth atlas on $|M|:=\cF(\R)$. The direct image of the structure sheaf $\cO_{\cU_\za^{p|\ul q}}$ over~$\cU^p_\za$ by the continuous map $|h_\za|^{-1}\colon \cU_\za^p\to|U_\za|$ is a sheaf over~$|U_\za|$, which we denote by $\cO_{U_\za}$:
\begin{gather*}\cO_{U_\za} := \big({|h_\za|}^{-1}\big)_* \cO_{\mathcal{U}^{p|\ul{q}}_\za}.\end{gather*}
The $\Z_2^n$-ringed space \begin{gather*}U_\za:=(|U_\za|,\cO_{U_\za})\end{gather*} is isomorphic to the $\Z_2^n$-domain $\cU_\za^{p|\ul q}$. The isomorphism is $h_\za:=(|h_\za|,h_\za^*)$, where $h_\za^*$ is the identity map (a composite of direct images is the direct image by the composite). In other words, we have an isomorphism of $\Z^n_2$-manifolds
\begin{gather*}h_\za\colon \ U_\za \rightarrow \mathcal{U}^{p|\ul{q}}_\za.\end{gather*}

Consider now an overlap $|U_{\za\zb}| : = |U_\za| \cap |U_\zb| \neq \varnothing$. Omitting restrictions, we get that $k_\zb k_\za^{-1}$ is a natural isomorphism (in $\SN{\Z_2^n\catname{Pts}^{\textnormal{op}}, {\tt A(N)FMan}}$) \begin{gather*}k_{\zb\za}:=k_\zb k_\za^{-1}\colon \ \cU^{p|\ul q}_{\za\zb}(-)\to \cU^{p|\ul q}_{\zb\za}(-)\end{gather*} between functors corresponding to $\Z_2^n$-domains (defined as usual). In view of Theorem~\ref{prop:AMorphisms}, the natural isomorphism $k_{\zb\za}$ is implemented by a $\Z_2^n$-isomorphism \begin{gather*}k_{\zb\za}\colon \ \cU_{\za\zb}^{p|\ul q}\to \cU_{\zb\za}^{p|\ul q}.\end{gather*} It follows that \begin{gather*}\psi_{\zb\za}:=h_\zb^{-1} k_{\zb\za}h_\za\colon \ U_{\za\zb}\to U_{\zb\za}\end{gather*} is an isomorphism of $\Z_2^n$-manifolds, where $U_{\za\zb}:=(|U_{\za\zb}|,\cO_{U_\za}|_{|U_{\za\zb}|})$. The $\Z_2^n$-manifolds $U_\za$ can thus be glued and provide then a $\Z_2^n$-manifold $M$ over $|M|=\cF(\R)$, such that there are $\Z_2^n$-isomorphisms $(|U_{\za}|,\cO_M|_{|U_\za|})\to U_\za$, if the $\psi_{\zb\za}$ satisfy the cocycle condition.

Since the Schwarz--Voronov embedding is fully faithful, we have that $\psi_{\zg\zb}\psi_{\zb\za}=\psi_{\zg\za}$ as $\Z_2^n$-morphisms if and only if the induced natural transformations coincide. However, for any $\zL$, we get
\begin{gather*}
(\psi_{\zg\zb}\psi_{\zb\za})_\zL=(h_{\zg,\zL})^{-1}k_{\zg,\zL}(k_{\zb,\zL})^{-1}h_{\zb,\zL}(h_{\zb,\zL})^{-1}k_{\zb,\zL}(k_{\za,\zL})^{-1}h_{\za,\zL}=\psi_{\zg\za,\zL}.\end{gather*}

It remains to show that $M$ actually represents $\cF$, i.e., that we can find a natural isomorphism $\zh\colon M(-)\to\cF$ in the category $\SN{\Z_2^n\catname{Pts}^{\textnormal{op}}, {\tt A(N)FMan}}$, i.e., that, for any $\zL\in\Z_2^n\tt GrAlg$, there is a~$\zL_0$-diffeomorphism $\zh_\zL\colon M(\zL)\to\cF(\zL)$ that is natural in~$\zL$. As $(|U_\za|)_\za$ is an open cover of~$|M|$, the source decomposes as \begin{gather*}M(\zL)=\bigcup_\za U_\za(\zL),\end{gather*} the $U_\za(\zL)$ being open (nuclear) Fr\'echet $\zL_0$-submanifolds. On any $U_\za(\zL)$, we define $\zh_\zL$ by setting \begin{gather*}\zh_\zL|_{U_\za(\zL)}:=(k_{\za,\zL})^{-1}h_{\za,\zL}\colon \ U_\za(\zL)\to \cF|_{|U_\za|}(\zL)\subset \cF(\zL).\end{gather*} These restrictions provide a~well-defined map \begin{gather*}\zh_\zL\colon \ M(\zL)\to\cF(\zL).\end{gather*} Indeed, if $(x,m^*_\star)\in U_\za(\zL)\cap U_\zb(\zL)$, we have \begin{gather*}(k_{\za,\zL})^{-1}(h_{\za,\zL}(x,m^*_\star))=(k_{\zb,\zL})^{-1}(h_{\zb,\zL}(x,m^*_\star))\qquad\text{if and only if}\quad \psi_{\zb\za,\zL}(x,m^*_\star)=(x,m^*_\star).\end{gather*} However, since we glued $M$ from the $U_\za$, the gluing $\Z_2^n$-isomorphisms $\psi_{\zb\za}$ became identities and so did the induced natural isomorphisms. The definition of $\zh_\zL^{-1}$ is similar. The source $\cF(\zL)$ decomposes as \begin{gather*}\cF(\zL)=\bigcup_\za\cF|_{|U_\za|}(\zL),\end{gather*} the $\cF|_{|U_\za|}(\zL)$ being open (nuclear) Fr\'echet $\zL_0$-submanifolds. On any $\cF|_{|U_\za|}(\zL)$, we define~$\zh_\zL^{-1}$ by setting \begin{gather*}\zh^{-1}_\zL|_{\cF|_{|U_\za|}(\zL)}:=(h_{\za,\zL})^{-1}k_{\za,\zL}\colon \ \cF|_{|U_\za|}(\zL)\to U_\za(\zL)\subset M(\zL).\end{gather*} The condition for these restrictions to give a well-defined map \begin{gather*}\zh_\zL^{-1}\colon \ \cF(\zL)\to M(\zL)\end{gather*} is equivalent to the condition for $\zh_\zL$. Clearly, the maps $\zh_\zL$ and $\zh_\zL^{-1}$ are inverses. Naturality and $\zL_0$-smoothness are local questions and are therefore consequences of the naturality and the $\zL_0$-smoothness of $(k_{\za,\zL})^{-1}h_{\za,\zL}$ and of $(h_{\za,\zL})^{-1}k_{\za,\zL}$.
\end{proof}

We are now prepared to show that the category $\Z_2^n\tt Man$ is equivalent to a functor category.
\begin{Theorem}The category $\Z_2^n\tt Man$ of $\Z_2^n$-manifolds $($defined as $\Z_2^n$-ringed spaces that are locally isomorphic to $\Z_2^n$-domains$)$ and $\Z_2^n$-morphisms $($defined as morphisms of $\Z_2^n$-ringed spaces$)$ is equivalent to the full subcategory $\SN{\Z_2^n\catname{Pts}^{\textnormal{op}}, {\tt A(N)FMan}}_{\op{rep}}$ of representable functors in $\SN{\Z_2^n\catname{Pts}^{\textnormal{op}},\allowbreak {\tt A(N)FMan}}$.
\end{Theorem}

In other words, the category $\Z_2^n\tt Man$ is equivalent to the category of locally trivial functors in the subcategory of the functor category $[\Z_2^n\catname{Pts}^{\textnormal{op}}, {\tt A(N)FMan}]$, whose objects $\cF$ have values $\cF(\zL)$ in (nuclear) Fr\'echet $\zL_0$-manifolds and whose morphisms are the natural transformations with $\zL_0$-smooth components.

\begin{Remark} This result is reminiscent of the identification of schemes with those contrava\-riant functors from affine schemes to sets that are sheaves (for the Zariski topology on affine schemes) and have a cover by open immersions of affine schemes.\end{Remark}

\begin{proof} The Schwarz--Voronov embedding viewed as functor valued in $\SN{\Z_2^n\catname{Pts}^{\textnormal{op}}, {\tt A(N)FMan}}_{\op{rep}}$ is obviously fully faithful and essentially surjective. It thus induces an equivalence of categories.
\end{proof}

\appendix

\section{Generating sets of categories}\label{appx:GenSet}

We will freely use Mac~Lane's book \cite{MacLane1998} as our source of categorical notions and proofs of general statements. For completeness, we recall the concept of generating set of a category.
 \begin{Definition}[{\cite[p.~127]{MacLane1998}}]\label{def:GeneratingSet}
Let $\catname{C}$ be a category. A set $S = \{ S_i \in \textnormal{Ob}(\catname{C})\colon i \in I \}$, where $I$ is any index set, is said to be a \emph{generating set} of $\catname{C},$ if, for any pair of distinct $\tt C$-morphisms
\begin{gather*}\phi, \psi \colon \ A \longrightarrow B,\end{gather*}
i.e., $\phi \neq \psi$, there exists some $i \in I$ and a $\tt C$-morphism
\begin{gather*} s \colon \ S_i \longrightarrow A,\end{gather*}
such that the compositions
\begin{gather*} S_i \stackrel{s}{\longrightarrow} A \mathrel{\mathop{ \rightrightarrows}^{\phi}_{\psi}} B \end{gather*}
are distinct, i.e., $\phi \circ s \neq \psi \circ s$. In this case, we say that the object~$ S_i $ \emph{separates} the morphisms~$\phi$ and~$\psi$, and that the set~$S$ \emph{generates} the category~$\catname{C}$.
 \end{Definition}

\begin{Example}The set $\{ \R\}$ is a generating set of the category of finite-dimensional real vector spaces. This is easily seen, as, if we have two distinct linear maps $\phi, \psi \colon V \rightarrow W$, then there exists a vector $v \in V$ ($v \neq 0$), such that $\phi(v) \neq \psi(v)$. Thus, the two linear maps differ on the one-dimensional subspace generated by~$v$. Now let $z$ be a basis of $\R$. Then, the linear map $s\colon \R \rightarrow V$ given by $s(z) = v$, keeps $\phi$ and $\psi$ separate.
\end{Example}

\begin{Proposition}\label{RestYonProp} For any locally small category $\tt C$, a set $S\subset\op{Ob}(\tt C)$ generates $\tt C$ if and only if the restricted Yoneda embedding \begin{gather*}\cY_S\colon \ {\tt C}\to \big[S^{\op{op}},{\tt Set}\big],\end{gather*} where $S$ is viewed as full subcategory of $\tt C$, is faithful.\end{Proposition}

\begin{proof} The restricted embedding is defined on objects by \begin{gather*}\cY_S(A)=\Hom_{\tt C}(-,A)\in\big[S^{\op{op}},\tt Set\big]\end{gather*} and on morphisms by \begin{gather*}\cY_S(\phi)=\Hom_{\tt C}(-,\phi)\colon \ \cY_S(A)\to\cY_S(B),\end{gather*} where \begin{gather*}(\cY_S(\phi))_{S_i}\colon \ \Hom_{\tt C}(S_i,A)\ni s\mapsto \phi\circ s\in\Hom_{\tt C}(S_i,B).\end{gather*} The embedding $\cY_S$ is faithful if and only if, for any different $\phi,\psi\colon A\to B$, the corresponding natural transformations are distinct, i.e., there is at least one $i\in I$ and one $s\in\Hom_{\tt C}(S_i,A)$, such that $\phi\circ s\neq\psi\circ s.$\end{proof}

\section{Fr\'{e}chet spaces, modules and manifolds} \label{appx:AManifolds}

Manifolds over algebras $A$, also known as $A$-manifolds, are manifolds for which the tangent spaces are endowed with a module structure over a given finite-dimensional commutative algebra. For details, the reader may consult Shurygin \cite{Shurygin:1996,Shurygin:1999,Shurygin:2002}, and for a discussion of the specific case of (the even part of) Grassmann algebras one may consult Azarmi \cite{Azarmi:2012}. A comprehensive introduction to the subject can be found in the book (in Russian) by Vishnevski\u{\i}, Shirokov, and Shurygin~\cite{Vishnevskii:1985}. The concept needed in this paper is a infinite-dimensional generalisation of an $A$-manifold to the category of Fr\'{e}chet algebras and Fr\'echet manifolds. For an introduction to locally convex spaces, including Fr\'{e}chet vector spaces, we refer the reader to Conway \cite[Chapter~IV]{Conway:1990}, Tr\`{e}ves \cite[Part~I]{Treves:1967}, or Rudin \cite[Chapter 1]{Rudin:1991}. A brief introduction to Fr\'{e}chet algebras can be found in Waelbroeck \cite[Chapter~VII]{Waelbroeck:1971}. For Fr\'{e}chet manifolds, the reader can consult Saunders \cite[Chapter~7]{Saunders:1989} and Hamilton \cite[Part~I.4]{Hamilton:1982}.

\begin{Definition}\label{Appx:FSpace}A \emph{Fr\'{e}chet $($vector$)$ space} is a complete Hausdorff metrizable locally convex topological vector space.
\end{Definition}

There exist a few other, equivalent, definitions of Fr\'echet spaces. The topology on a locally convex space is metrizable if and only if it can be derived from a countable family of semi-norms~$||-||_k $, $k\in\N $. The topology is Hausdorff if and only if the family of semi-norms is separating, i.e., if $||\mathrm{x}||_k =0$, for all $k$, implies $\mathrm{x}=0$. Given such a family of semi-norms, one defines a~translationally invariant metric that induces the topology by setting
\begin{gather*}d(\mathrm{x} , \mathrm{y}) = \sum_{k=0}^\infty 2^{-k}\frac{||\mathrm{x} - \mathrm{y} ||_k}{1+ || \mathrm{x} - \mathrm{y} ||_k},\end{gather*}
for all $\mathrm{x}$ and $\mathrm{y}$.

\begin{Example}Let \looseness=-1 $M = (|M|,\cO)$ be a $\Z_2^n$-manifold. For any open subset $U\subset |M|$, the space~$\cO(U)$ of $\Z_2^n$-functions on $U$ is a Fr\'echet space. An inducing family of semi-norms is given by
\begin{gather*}||f||_{C,D} = \sup_{x\in C}|\zve(D(f))(x)|,\end{gather*}
where $\zve$ is the projection $\ze\colon \cO(U)\to\Ci(U)$ of $\Z_2^n$-functions to base functions, where $C$ is any compact subset of $U$, and where $D$ is any $\Z_2^n$-differential operator over~$U$. Details on the construction of a countable family of semi-norms that is equivalent to~$(||-||_{C,D})_{C,D} $, can be found in the proof of the last lemma in~\cite{Bruce:2018}.
\end{Example}

Given two Fr\'{e}chet spaces $\big(F ,\big(||-||^F_k\big)_{k\in \mathbb N} \big)$ and $\big(G ,\big(||-||^G_k\big)_{k\in \mathbb N} \big)$, a linear map
\begin{gather*}\phi\colon \ F \longrightarrow G\end{gather*}
is continuous if and only if, for every semi-norm $||-||^G_k$, there exists a semi-norm $||-||^F_l$ and a~positive real number $C>0$, such that
\begin{gather*}||\phi(\mathrm{x})||^G_k \leq C ||\mathrm{x}||^F_l,\end{gather*}
for every $\mathrm{x} \in F$. A similar result holds for continuous bilinear maps \begin{gather*}\phi\colon \ F \times G \rightarrow H.\end{gather*} The {\it morphisms of Fr\'{e}chet spaces} are the continuous linear maps, so that the category of Fr\'echet spaces is a full subcategory of the category of topological vector spaces.

What makes Fr\'{e}chet spaces interesting, is the fact that they have just enough structure to define a derivative of a mapping between such spaces. This leads to a meaningful notion of a~smooth map between Fr\'echet spaces, and so much of finite-dimensional differential geometry can be transferred to the infinite-dimensional setting, using Fr\'{e}chet spaces as local models. The well known \emph{G\^{a}teaux $($directional$)$ derivative} is defined as follows.

\begin{Definition}\label{Appx:FGderivative}Let $F$ and $G$ be Fr\'{e}chet spaces and $U \subset F$ be open, and let $\phi\colon U \rightarrow G$ be a (nonlinear) continuous map. Then the \emph{derivative of $\phi$ in the direction of $\mathrm{v} \in F$ at $\mathrm{x} \in U$} is defined as
\begin{gather*}\rmd_\mathrm{x} \phi( \mathrm{v}) := \lim_{t \rightarrow 0}\frac{\phi(\mathrm{x} + t \mathrm{v}) - \phi(\mathrm{x}) }{t} \end{gather*}
provided the limit exists. We say that $\phi$ is \emph{continuously differentiable}, if the limit exists for all $\mathrm{x} \in U$ and $\mathrm{v} \in F$, and if the mapping
\begin{gather*}\rmd \phi \colon \ U \times F \longrightarrow G\end{gather*}
is (jointly) continuous.
\end{Definition}

Higher order derivatives are defined inductively, i.e.,
\begin{gather*}\rmd^{k+1}_\mathrm{x}\phi(\mathrm{v}_1, \mathrm{v}_2, \dots, \mathrm{v}_{k+1}) := \lim_{t \rightarrow 0} \frac{\rmd^k_{\mathrm{x} + t \mathrm{v}_{k+1}} \phi(\mathrm{v}_1, \mathrm{v}_2, \dots, \mathrm{v}_{k}) ~ {-} ~ \rmd^k_{\mathrm{x}}\phi(\mathrm{v}_1, \mathrm{v}_2, \dots ,\mathrm{v}_{k}) }{t}.\end{gather*}
A continuous map $\phi\colon U \rightarrow G$ is then said to be $k$ times continuously differentiable or to be of class $C^k$, if
\begin{gather*}\rmd^k \phi \colon \ U \times F^{\times k}\longrightarrow G \end{gather*}
is continuous (or, more explicitly, if all its derivatives of order $\le k$ exist everywhere and are continuous). If $\zvf$ is of class $C^k$, its derivative $\rmd^k_\mr{x}\zvf(\mr{v}_1,\mr{v}_2,\dots,\mr{v}_k)$ is multilinear and symmetric in $F^{\times k}$~\cite{Sharko}. Furthermore, we say that $\phi$ is \emph{smooth}, if it is of class $C^k$, for all $k$.

\begin{Proposition}\label{partialFrechet}Let $F_1$, $F_2$ be Fr\'echet spaces and let $U\subset F_1\times F_2$ be an open subset. A~continuous map $\phi\colon U \to G$ valued in a Fr\'echet space $G$ is of class $C^1$ if and only if its $($total$)$ derivative \begin{gather*}\op{d}\phi\colon \ U\times (F_1\times F_2)\ni ((f_1,f_2),(\mr{v}_1,\mr{v}_2))\mapsto \op{d}_{(f_1,f_2)}\phi (\mr{v}_1,\mr{v}_2)\in G\end{gather*} is continuous, which is the case if and only if the naturally defined partial derivatives \begin{gather*}\op{d}_{f_1}\phi\colon \ U\times F_1\ni ((f_1,f_2),\mr{v}_1)\mapsto \op{d}_{f_1,(f_1,f_2)}\phi (\mr{v_1})\in G\quad\end{gather*} and \begin{gather*} \op{d}_{f_2}\phi\colon \ U\times F_2\ni((f_1,f_2),\mr{v}_2)\mapsto \op{d}_{f_2,(f_1,f_2)}\phi (\mr{v_2})\in G\end{gather*} are continuous. In this case, we have \begin{gather*}\op{d}_{(f_1,f_2)}\phi (\mr{v}_1,\mr{v}_2) = \op{d}_{f_1,(f_1,f_2)}\phi (\mr{v_1}) + \op{d}_{f_2,(f_1,f_2)}\phi (\mr{v_2}).\end{gather*}
\end{Proposition}

The G\^ateaux or Fr\'{e}chet--G\^{a}teaux derivative gives a rather weak notion of differentiation, however, most of the standard results from calculus in the finite-dimensional setting remain true. Specifically, the fundamental theorem of calculus and the chain rule still hold. However, the inverse function theorem is in general lost. For a special class of Fr\'{e}chet spaces, known as `tame' Fr\'{e}chet spaces, there is an analogue of the inverse function theorem known as the Nash--Moser inverse function theorem, see Hamilton~\cite{Hamilton:1982} for details.

A \emph{nuclear space} is a locally convex topological vector space $F$, such that, for any locally convex topological vector space $G$, the natural map $F\widehat{\0}_\zp G\longrightarrow F\widehat{\0}_\iota G$ from the projective to the injective tensor product of $F$ and $G$ is an isomorphism of locally convex topological vector spaces. In particular, a {\it nuclear Fr\'echet space} is a locally convex topological vector space that is a nuclear space and a Fr\'echet space. Loosely, if a space $F$ is nuclear, then, for any locally convex space~$G$, the complete topological vector space $F\widehat{\0} G$ is independent of the locally convex topology considered on $F\0 G$. Because of this, and their nice dual properties, nuclear spaces provide a~reasonable setting for infinite-dimensional analysis. All the Fr\'{e}chet spaces we encounter in this paper are in fact nuclear.

The following definition is standard.
\begin{Definition}\label{Appx:FAlgebra} A \emph{Fr\'echet algebra} is a Fr\'echet vector space $A$, which is equipped with an associative bilinear and (jointly) continuous multiplication $\cdot\colon A\times A\to A $. If $(p_i)_{i\in I}$ is a family of semi-norms that induces the topology on $A$, (joint) continuity is equivalent to the existence, for any $i\in I$, of $j\in I$, $k\in I$, and $C>0$, such that \begin{gather*}p_i(x\cdot y)\le C p_j(x) p_k(y), \qquad \forall\, x,y\in A.\end{gather*} We can always consider an equivalent increasing countable family of semi-norms $(||-||_n)_{n\in\N}$. The preceding condition then requires that, for any $n\in\N$, there is $r\in\N$ $(r\ge n)$ and $C>0$, such that \begin{gather*}||x\cdot y||_n\le C ||x||_r ||y||_r, \qquad \forall\, x,y\in A.\end{gather*} In particular, the topology can be induced by a countable family of submultiplicative semi-norms, i.e., by a family $(q_n)_{n\in\N}$ that satisfies \begin{gather*}q_n(x\cdot y)\le q_n(x) q_n(y), \qquad \forall\, n\in\N,\quad \forall\, x,y\in A.\end{gather*} \end{Definition}

Note that many authors define a Fr\'echet algebra simply as a Fr\'echet vector space, which carries an associative bilinear multiplication, and whose topology can be induced by a countable family of submultiplicative semi-norms. This latter definition is equivalent to the former.

In general, a Fr\'{e}chet algebra need not be unital, and, if it is, one does not require $p_i(1_A) =1$, in contrast to what is usually required for Banach algebras.

\begin{Example}[formal power series]Consider the space \begin{gather*}\R[[z_1, z_2, \dots , z_q]] \end{gather*} of formal power series in $q$ parameters and coefficients in reals. We set $j:= (j_1, j_2, \dots , j_q) \in \N^{\times q} $ and $|j|:= j_1 +j_2 + \dots + j_q$. A general series $\mathrm{x}$ now reads
\begin{gather*} \mathrm{x} = \sum_{j} z^j a_j = \sum_{j} z_1^{j_1} z_2^{j_2} \cdots z_q^{j_q} a_{j_q \dots j_2 j_1},\end{gather*}
with no question on the convergence. The algebra structure is the standard multiplication of formal power series. The topology of coordinate-wise convergence is metrizable and given by the family of semi-norms
\begin{gather*}||\mathrm{x}||_k := \sum_{ |j| \leq k} |a_j|, \qquad \forall\, k \in \N.\end{gather*}
This algebra is unital with the obvious unit, and it is submultiplicative.
\end{Example}

Let us denote the category of \emph{Fr\'{e}chet algebras} (resp., {\it commutative Fr\'echet algebras}) as $\tt FAlg$ (resp., $\catname{CFAlg}$). Morphisms in this category are defined to be continuous algebra morphisms. If we restrict attention to nuclear Fr\'{e}chet algebras (resp., commutative nuclear Fr\'echet algebras), then we work in the full subcategory $\tt NFAlg$ (resp.,~$\catname{CNFAlg}$).

\begin{Definition}\label{Appex:FModule}Fix $A \in \catname{FAlg}$. A \emph{Fr\'{e}chet $A$-module} is a Fr\'{e}chet vector space $F$, together with a continuous action
\begin{align*}
 A \times F & \stackrel{\mu}{\longrightarrow} F,\\
 (a, \mathrm{v}) & \mapsto \mu(a,\mathrm{v}) ,
\end{align*}
which we will write as $\mu(a,\mathrm{v}) := a \cdot \mathrm{v}$ (and which is of course compatible with the multiplication in $A$).
\end{Definition}

We give a short survey on Fr\'echet manifolds.
\begin{Definition}\label{Appx:FChart0}Let $\cM$ be a set. An $F$-chart of $\cM$ is a bijective map $\phi\colon U\to\phi(U)\subset F$, where $U\subset\cM$ and $\phi(U)$ is an open subset of a Fr\'echet space $F$.
\end{Definition}

A Fr\'echet atlas can be defined using charts valued in various Fr\'echet spaces. For our purposes, it is sufficient to consider a fixed Fr\'echet model.

\begin{Definition}\label{Appx:FChart}
 A \emph{smooth $F$-atlas} on a set $\mathcal{M}$ is a collection of $F$-charts $((U_\alpha, \phi_\alpha))_{\alpha \in \mathcal{A}}$, such that
\begin{enumerate}\itemsep=0pt
\item[(i)] the subsets $U_\alpha$ cover the set $\mathcal{M}$,
\item[(ii)] the subsets $\phi_{\alpha}(U_\alpha \cap U_\beta)$ are open in $F$,
\item[(iii)] the transition maps
\begin{gather*} \phi_{ \beta \alpha} := \phi_\beta \circ \phi_\alpha^{-1} \colon \ \phi_\alpha (U_\alpha \cap U_\beta )\subset F \longrightarrow \phi_\beta (U_\zb \cap U_\za )\subset F\end{gather*}
are smooth.
\end{enumerate}
\end{Definition}

A new $F$-chart $(U, \phi)$ on $\mathcal{M}$ is \emph{compatible} with a given smooth $F$-atlas, if and only if their union is again a smooth $F$-atlas, i.e., the subsets $\phi(U\cap U_\za)\subset F$ and $\phi_\za(U_\za\cap U)\subset F$ are open, and the transition maps
\begin{gather*} \phi_\alpha \circ \phi^{-1} \colon \ \phi(U \cap U_\alpha ) \longrightarrow \phi_\alpha(U_\za \cap U)\qquad\text{and}\qquad\phi \circ \phi_\za^{-1} \colon \ \phi_\za(U_\za \cap U) \longrightarrow \phi(U \cap U_\alpha ) \end{gather*}
are smooth (for every $\alpha \in \mathcal{A}$). Similarly, two smooth $F$-atlases are compatible provided their union is also a smooth $F$-atlas. Compatibility is an equivalence relation on all possible smooth $F$-atlases on~$\mathcal{M}$.

\begin{Definition}\label{Appx:Fmanifold}
A \emph{smooth $F$-structure} on a set $\mathcal{M}$ is a choice of an equivalence class of smooth $F$-atlases on $\mathcal{M}$. We say that $\mathcal{M}$ is a \emph{Fr\'{e}chet manifold} modelled on the Fr\'{e}chet space~$F$, if $\cM$ comes equipped with a smooth $F$-structure. If the model vector space $F$ is nuclear, we speak of a \emph{nuclear Fr\'{e}chet manifold}.
\end{Definition}

A smooth $F$-atlas on a Fr\'echet manifold $\cM$ allows us to define in the obvious way a topo\-logy on $\cM$, which is independent of the atlas considered in the chosen equivalence class. The domain~$U$ of an $F$-chart $(U,\phi)$ is open in this topology and the bijective map $\phi\colon U\subset\cM\to \phi(U)\subset F$ is a homeomorphism for the induced topologies. Most authors confine themselves to Fr\'echet manifolds, whose topology is Hausdorff.

{\it Morphisms between two Fr\'{e}chet manifolds} are the {\it smooth maps} between them, where smoothness is defined, just as in the finite-dimensional case of smooth manifolds, in terms of charts and smoothness of local representatives of the maps. We denote the category of Fr\'echet manifolds and the morphisms between them by~$\tt FMan$.

Further, the {\it tangent space} $\sT_f\cM$ to a Fr\'echet manifold $\cM$ at a point $f\in\cM$ can be defined as usual, using the tangency equivalence relation for the smooth curves of $\cM$ that pass through~$f$ at time~$0$. One can easily see that $\sT_f\cM$ is a Fr\'echet space.
The concept of {\it Fr\'echet vector bundle} is the natural generalization of the notion of smooth vector bundle to the category of Fr\'echet manifolds. The {\it tangent bundle} $\sT\cM$ of a Fr\'echet manifold $\cM$ is an example of a Fr\'echet vector bundle.

In general, we must make a distinction between the {\it $( $kinematic$ )$ tangent bundle} as defined here and the \emph{operational tangent bundle} defined in terms of derivations of the algebra of functions of a~Fr\'{e}chet manifold. Indeed, the two notions do not, in general, coincide, there are derivations that do not correspond to tangent vectors. However, it is known that for nuclear Fr\'{e}chet manifolds the two concepts do coincide.

Let $\mathfrak{F}\colon \cM\to\mathcal{N}$ be a smooth map between Fr\'echet manifolds modelled on Fr\'echet spaces $F$ and $G$, respectively. There is a {\it tangent map} $\sT \mathfrak{F}$ of $\mathfrak{F}$, which is a smooth map \begin{gather*} \sT \mathfrak{F}\colon \ \sT\cM\to\sT\mathcal{N},\end{gather*} and restricts, for any $f\in\cM$, to a linear map \begin{gather*} \sT_f\mathfrak{F}\colon \ \sT_f\cM\to\sT_{\mathfrak{F}(f)}\mathcal{N}.\end{gather*} As in the finite-dimensional case, the local representative of $\sT_f\mathfrak{F}$ is the derivative $\op{d}_{\phi(f)}\big(\psi\mathfrak{F}\phi^{-1}\big)$ of the corresponding local representative
\begin{gather*} \psi\mathfrak{F}\phi^{-1}\colon \ \phi(U)\subset F\to G\end{gather*} of $\mathfrak{F}$ at the point~$\phi(f)$.

Fundamental to the work in this paper are Fr\'{e}chet manifolds with a further module structure on their tangent bundle.
\begin{Definition}\label{Appx:Aman}
Let $\mathcal{M}$ be a Fr\'{e}chet manifold, whose model Fr\'echet space $F$ is a module over a Fr\'{e}chet algebra $A$. We say that $\cM$ is a \emph{Fr\'{e}chet A-manifold}, if and only if all transition maps are $A$-linear, i.e.,
\begin{gather*} \rmd_{\phi_{\za}(f)} \phi_{\beta \alpha}(a \cdot \mathrm{v}) = a \cdot \rmd_{\phi_{\za}(f)} \phi_{\zb\za}( \mathrm{v}) ,\end{gather*}
for all $f \in U_\alpha \cap U_\beta$, $a \in A$, and $ \mathrm{v} \in F$.
\end{Definition}

{\it Morphisms between Fr\'{e}chet $A$-manifolds} $\mathcal{M}$ and $\mathcal{N}$ are the \emph{$A$-smooth} maps between them, i.e., are the smooth maps $\mathfrak{F}\colon \mathcal{M} \rightarrow \mathcal{N}$ that are $A$-linear at every point. This means that, for any point $f\in\cM$, there is an $\cM$-chart $(U,\phi)$ around $f$ and an $\mathcal N$-chart $(V,\psi)$ around $\mathfrak{F}(f)$ that contains $\mathfrak{F}(U)$, such that the local representative
\begin{gather*} \op{d}_{\phi(f)}\big(\psi\mathfrak{F}\phi^{-1}\big)\end{gather*} of the derivative $\sT_f\mathfrak{F}$ is an $A$-linear endomorphism of the $A$-module $F$. The requirement actually means that the derivative $\sT_f\mathfrak{F}$ must be $A$-linear at any point $f\in\cM$. In this way, we obtain the category of Fr\'{e}chet $A$-manifolds, which we denote as $A\catname{FMan}$.

In this paper, we will use the category ${\tt A}\catname{FMan}$, whose objects are the Fr\'{e}chet $A$-manifolds, where $A$ is not a fixed Fr\'echet algebra, but any Fr\'echet algebra. The definition of ${\tt A}\catname{FMan}$-morphisms generalizes the definition of $A\catname{FMan}$-morphisms. Suppose that $\mathcal{M}$ is a Fr\'{e}chet $A$-manifold modelled on an $A$-module $F$ and $\mathcal{N}$ is a Fr\'{e}chet $B$-manifold modelled on a $B$-module~$G$. The ${\tt A}\catname{FMan}$-morphisms from $\cM$ to $\mathcal N$ are the ${\tt A}$-smooth maps between them, i.e., those smooth maps $\mathfrak{F}\colon \mathcal{M} \rightarrow \mathcal{N}$ that are at any point compatible with the module structures of $F$ and $G$. This means that there is a Fr\'echet algebra morphism $\zr\colon A\to B$, and, for any $f\in\cM$, there exist charts $(U,\phi)$ and $(V,\psi)$ as above, such that \begin{gather*} \op{d}_{\phi(f)}\big(\psi\mathfrak{F}\phi^{-1}\big)(a\cdot \mr{v})=\zr(a)\cdot\op{d}_{\phi(f)}\big(\psi\mathfrak{F}\phi^{-1}\big)(\mr{v}),\end{gather*} for any $a \in A$ and $ \mathrm{v} \in F$. This requirement actually means that, for any~$f$, the derivative $\sT_f\mathfrak{F}$ is compatible with the induced actions on the tangent spaces. We will refer to an $\tt A$-smooth map with associated Fr\'echet algebra morphism $\zr$, as a $\zr$-smooth map. If we restrict our attention to nuclear objects, i.e., the model Fr\'{e}chet vector space and the Fr\'{e}chet algebra are both nuclear, then we denote the corresponding category as ${\tt A}\catname{NFMan}$.

\subsection*{Acknowledgements}
The authors cordially thank the anonymous referees for their valuable remarks and comments, which have served to improve this article, as well as for their suggestions for future research.

\pdfbookmark[1]{References}{ref}
\LastPageEnding

\end{document}